\documentclass[a4paper,12pt]{report}

\usepackage[utf8]{inputenc}
\usepackage[T1]{fontenc}
\usepackage[english]{babel}
\usepackage{amsmath}
\usepackage{amssymb,amsfonts,amsthm}
\usepackage[a4paper,top=2.5cm,bottom=2.8cm,left=2.8cm,right=2.8cm]{geometry}
\usepackage{graphicx}
\usepackage{booktabs}
\usepackage{array}
\usepackage[square,authoryear]{natbib}
\usepackage{csquotes}
\usepackage{tikz}
\usepackage[hidelinks]{hyperref}
\hypersetup{pdftitle={Geodesics in Black-Hole Spacetimes},pdfauthor={Oliver F. Piattella}}
\usetikzlibrary{arrows.meta,decorations.markings,calc}

\usepackage{pgfplots}
\pgfplotsset{compat=1.18}
\usepgfplotslibrary{groupplots}
\definecolor{c1}{HTML}{5E81B5}\definecolor{c2}{HTML}{E19C24}
\definecolor{c3}{HTML}{8FB032}\definecolor{c4}{HTML}{EB6235}
\newcommand{\dd}{\mathop{}\!\mathrm{d}}
\newcommand{\sgn}{\operatorname{sgn}}
\theoremstyle{plain}\newtheorem{theorem}{Theorem}[chapter]
\theoremstyle{definition}\newtheorem{definition}{Definition}[chapter]
\newtheorem{example}{Example}[chapter]

\title{Geodesics in Black-Hole Spacetimes\\
\large A qualitative classification using Vieta's formulae}
\author{Oliver F. Piattella\footnote{Dipartimento di Scienza e Alta Tecnologia, Università degli Studi dell'Insubria, Como, Italy.\\ E-mail: \url{of.piattella@uninsubria.it}}}
\date{}

\begin{document}

\hypersetup{pageanchor=false}
\maketitle

% Dedication
\thispagestyle{empty}
\vspace*{\stretch{1}}
\begin{center}
\itshape To Giulio, the turning point of my motion

\vspace{3em}

In memory of Dorotea,\\
faithful companion through many of these pages
\end{center}
\vspace*{\stretch{2}}
\clearpage
\setcounter{page}{1}
\hypersetup{pageanchor=true}

\chapter*{Preface}
\addcontentsline{toc}{chapter}{Preface}

Among the most striking predictions of General Relativity are black holes: regions of spacetime bounded by an event horizon, from which nothing, not even light, can escape, and which, in classical theory, contain spacetime singularities. Understanding the motion of freely falling test bodies and light rays in black-hole spacetimes is both a fundamental problem of general relativity and an indispensable tool of relativistic astrophysics.

Historically, geodesics provided the first contact between the theory and observation. The anomalous advance of Mercury's perihelion is the leading relativistic correction to bound timelike orbits around the Sun, and the deflection of starlight measured in 1919 is the weak-field limit of null geodesics in the same geometry. These classic tests, however, probe the weak-field regime, where the departure from Newtonian motion is small. Close to a black hole, the situation is qualitatively different. Circular orbits of light exist; there is an innermost stable circular orbit, inside which no bound motion is possible; particles can whirl around the hole many times before escaping or plunging; and in the Kerr geometry, orbits are no longer confined to a plane, while the ergoregion admits states of negative energy.

These strong-field features are no longer of purely theoretical interest. The images of M87$^\ast$ and Sgr~A$^\ast$ obtained by the Event Horizon Telescope show a bright ring whose size is closely tied to the critical curve traced by unstable photon orbits. The same orbits govern the relativistic images of strong gravitational lensing and, in the eikonal limit, the quasinormal-mode spectrum of the black hole. The innermost stable circular orbit sets the inner edge of thin accretion disks, and with it the radiative efficiency of accretion and the X-ray methods used to measure black-hole spins. Extreme-mass-ratio inspirals, among the main targets of the LISA mission, are modeled to leading order as a slow drift through a sequence of Kerr geodesics, so that the emitted waveform encodes the fundamental frequencies of the orbits the small body traverses.

In all these applications, what matters first is the qualitative structure of the motion: which kinds of orbits exist for given constants of motion, where the transitions between them occur, and what happens at the boundaries. As we shall see, this structure is entirely encoded in the roots of the polynomial governing the radial motion and, for the Kerr geometry, in the roots of the one describing the polar motion.

Systematic treatments of geodesic motion can be found in the standard textbooks by \citet{weinberg1972gravitation}, \citet{misner1973gravitation}, \citet{chandrasekhar1998mathematical}, and \citet{FrolovZelnikov}, to which we refer the reader for a comprehensive account. 

This monograph investigates the geodesics of three black-hole geometries, namely the Schwarzschild, Kerr, and Schwarzschild--de~Sitter ones, by employing a different, algebraic approach. The aim is not to solve for the geodesic trajectories explicitly or numerically, but to classify the possible motions of test particles through their turning points. This amounts to studying the pattern of the roots of the polynomials that appear in the geodesic equations. To achieve this, the key algebraic tool adopted in this monograph is provided by Vieta's formulae, which relate the aforementioned roots to the constants of motion. Together with these formulae, the employment of Descartes' rule of signs also proves to be extremely useful.

This project stemmed from an idea that came to my mind when I was a PhD student studying Chandrasekhar's book. There, the root pattern of a polynomial was treated in relation to Kerr geodesics, but I felt that it was not being fully exploited. Then came Vieta's formulae, which I found baffling at first. But, thinking about it, if one just needs the qualitative features of a motion and, in particular, of its trajectory, then it makes sense that some information can be drawn from the coefficients of a polynomial alone. I thought that the idea was not completely trivial and, at least, fun to pursue. I could not always find the time to do that, though, so the project was developed intermittently over many years. Only recently, with the help of Claude, was I able to go the last mile, especially in producing the plots that I hope the reader will enjoy. The content and the idea, however, are the same as they were many years ago, when the only AIs I knew were R.~Daneel Olivaw and HAL~9000.

\vspace{1cm}

\noindent
\emph{October 2026\\OFP}

\tableofcontents

\chapter*{Conventions}
\addcontentsline{toc}{chapter}{Conventions}
\markboth{Conventions}{Conventions}

We use geometrized units $G = c = 1$, the metric signature is $(-,+,+,+)$, Greek indices run over the four spacetime coordinates, and repeated indices are summed over. 

Along a geodesic, $\lambda$ denotes the proper time for a massive particle and an affine parameter for a photon; a dot denotes differentiation with respect to $\lambda$. 

We only consider the region outside the event horizon, the domain of outer communication, unless stated otherwise; in Schwarzschild--de~Sitter spacetime, this is the static region between the event horizon and the cosmological horizon; turning points inside a horizon do not belong to the exterior classification.

\chapter{Introduction}

Thanks to the constants of motion, the geodesic equations in the Schwarzschild and Kerr spacetimes can be reduced to first-order differential equations in which the square of the radial or polar velocity equals a polynomial. In Schwarzschild spacetime, the radial polynomial is cubic; in Kerr spacetime, it is quartic, accompanied by a quadratic polynomial that governs the polar motion; in Schwarzschild--de~Sitter spacetime, the cosmological constant raises the degree of the timelike radial polynomial to five. The coefficients of these polynomials depend only on the constants of motion, and their real roots are the candidate \emph{turning points}, where the radial or polar velocity vanishes and the motion may reverse. The qualitative behavior of a geodesic, whether it stays bound, escapes to infinity, or falls into the black hole, is therefore encoded in where these roots lie.

The roots can be computed in closed form, but the resulting expressions are rarely the clearest route to a classification. This book follows a different route and reads the qualitative motion directly off the coefficients. The main tool is provided by Vieta's formulae, which express the sums and products of the roots in terms of the constants of motion. These relations alone can tell whether the number of negative roots is even or odd, exclude entire configurations of roots, or show that the roots cannot all be real. Two further classical tools complete the picture: the discriminant, which decides whether the roots are real, and Descartes' rule of signs, which counts how many of them lie beyond a given point, such as the event horizon. Together with the values of the polynomials at the horizon and at infinity, these tools classify the generic motions without computing a single root.

Not every real root is a turning point, however. A root is physically relevant only if it lies in the region under consideration, outside the event horizon for the radial motion and between the poles for the polar one, and if the other separated equation admits real motion for the same constants of motion. Roots inside the horizon, negative roots, and complex roots still enter Vieta's formulae, and they carry useful information, but they are not places where the motion turns around. In Kerr spacetime, there are constants of motion for which the radial polynomial has roots outside the horizon, but no polar motion is possible, and no geodesic exists at all.

The Schwarzschild problem provides a transparent introduction because spherical symmetry confines every geodesic to a plane. The Kerr problem is richer: axial symmetry leaves only one manifest component of angular momentum, while the hidden symmetry discovered by Carter supplies a further constant and separates the radial and polar motions \citep{Carter:1968rr, Carter:1968ks}. Classic classifications and analytic treatments include those of \citet{Wilkins:1972rs}, \citet{Schmidt:2002qk}, \citet{DrascoHughes:2004}, and \citet{FujitaHikida:2009}. Finally, the Schwarzschild--de~Sitter problem brings back spherical symmetry but adds a cosmological horizon and a quintic radial polynomial, for which the roots cannot be written in general by radicals: there, reading the motion off the coefficients is not only convenient but practically necessary.

The plan of the book is as follows. Chapter~\ref{Sec:Vietasform} introduces Vieta's formulae, together with the discriminant and Descartes' rule of signs; it then defines root patterns and turning points. Chapter~\ref{Sec:Schgeod} classifies Schwarzschild geodesics, first through the effective potential and then through the roots of the radial cubic. Chapter~\ref{Sec:Kerrgeod} treats Kerr geodesics: the polar motion, the radial quartic in each energy regime, spherical and equatorial circular orbits, and the photon shell. Chapter~\ref{Sec:SdS} treats Schwarzschild--de~Sitter geodesics: the horizons and the static radius, the inner and outer stable circular orbits, and the root patterns of the radial quintic. Each of these chapters ends with worked examples in which the predicted root patterns are compared with numerically computed roots. Chapter~\ref{Sec:Conclusions} summarizes the results.

\chapter{Vieta's formulae and turning points}\label{Sec:Vietasform}

In all the spacetimes studied in this book, the geodesic equations can be reduced to first-order, non-linear differential equations of the form $\dot x^2=P(x)$, one for each separated coordinate, where $P$ is a polynomial whose coefficients depend on the constants of motion. Everything that is qualitatively important about the motion, namely where it is allowed, where it turns back, and whether it can stay at a constant radius, is encoded in the roots of $P$. The strategy of this book is to extract this information without computing the roots.

This chapter provides the two tools that make this possible. The first is algebraic: Vieta's formulae (Section~\ref{Sec:Vieta}) relate the symmetric combinations of the roots to the coefficients, and hence to the constants of motion. In this way, they constrain the possible configurations of the roots, which we call root patterns (Section~\ref{Sec:RootPatterns}). The second is dynamical: the definition of a turning point (Section~\ref{Sec:TurningPoints}) states which roots correspond to an actual reversal of the motion, and how multiple roots describe orbits at constant radius and their stability. Two further algebraic tools, the discriminant and Descartes' rule of signs (Section~\ref{Sec:Discriminants}), complement Vieta's formulae by determining whether the roots are real and how many of them lie beyond a given point, such as the event horizon.

\section{Vieta's formulae}\label{Sec:Vieta}

Fran\c{c}ois Vi\`ete (1540--1603), Latinized as Franciscus Vieta, was a French lawyer and mathematician. Born in Fontenay-le-Comte, in Poitou, he spent most of his life in public service as a member of the Parlement of Brittany and later as a royal councilor to Henri III and Henri IV. Mathematics occupied his free time, and his contributions were lasting. In his \emph{In artem analyticem isagoge} (1591), he introduced the systematic use of letters for both unknown and known quantities, vowels for the former and consonants for the latter. For the first time, general equations could be written and manipulated, rather than particular numerical examples. This \emph{logistica speciosa} is the ancestor of modern algebraic notation. Vi\`ete also gave the first infinite product representation of $\pi$, and during the wars of Religion, he broke the ciphers of intercepted Spanish dispatches for Henri IV so successfully that Philip II of Spain is said to have accused him of sorcery.

The relations between the roots and the coefficients of a polynomial that now bear his name appear in \emph{De aequationum recognitione et emendatione}, published posthumously in 1615 by his pupil Alexander Anderson. Vi\`ete considered positive roots only. The general statement, valid for all roots, including negative and what were then called \enquote{impossible} (complex) roots, was given by Albert Girard in his \emph{Invention nouvelle en l'alg\`ebre} (1629), which is why the formulae are sometimes attributed to both. In modern language, they read as follows. 

Consider a polynomial of degree $n$,
\begin{equation}\label{ndegpoly}
 P_n(x)=a_nx^n+a_{n-1}x^{n-1}+\cdots+a_1x+a_0\,,
 \qquad a_n\neq0\,,
\end{equation}
with complex coefficients. By the fundamental theorem of algebra, it has $n$ complex roots, counted with multiplicity. If these roots are $r_1,\ldots,r_n$, then
\begin{equation}\label{ndegpoly2}
 P_n(x)=a_n\prod_{j=1}^{n}(x-r_j)\,.
\end{equation}
Define the elementary symmetric polynomials
\begin{equation}
 s_k:=\sum_{1\leq i_1<i_2<\cdots<i_k\leq n}
 r_{i_1}r_{i_2}\cdots r_{i_k}\,.
\end{equation}
Comparing equal powers of $x$ in Eqs.~\eqref{ndegpoly} and \eqref{ndegpoly2} gives Vieta's formulae:
\begin{equation}
 s_k = (-1)^k\frac{a_{n-k}}{a_n}\,.
\end{equation}
For example, for a cubic polynomial,
\begin{equation}\label{3degpoly}
 P_3(x)=a_3x^3+a_2x^2+a_1x+a_0\,,
\end{equation}
the formulae are
\begin{align}
 r_1+r_2+r_3&=-\frac{a_2}{a_3}\,,\\
 r_1r_2+r_1r_3+r_2r_3&=\frac{a_1}{a_3}\,,\\
 r_1r_2r_3&=-\frac{a_0}{a_3}\,.
\end{align}
For a quartic polynomial
\begin{equation}\label{4degpoly}
 P_4(x)=a_4x^4+a_3x^3+a_2x^2+a_1x+a_0\,,
\end{equation}
one similarly obtains
\begin{align}
 r_1+r_2+r_3+r_4&=-\frac{a_3}{a_4}\,,\\
 \sum_{i<j}r_ir_j&=\frac{a_2}{a_4}\,,\\
 \sum_{i<j<k}r_ir_jr_k&=-\frac{a_1}{a_4}\,,\\
 r_1r_2r_3r_4&=\frac{a_0}{a_4}\,.
\end{align}

\section{Discriminants and the rule of signs}\label{Sec:Discriminants}

Vieta's formulae constrain the symmetric functions of the roots but do not, by themselves, tell whether the roots are real or where they lie. Two classical tools complete the picture: the discriminant, which decides whether the roots are real, and Descartes' rule of signs, which counts how many of them lie beyond a given point. Both are polynomial functions of the coefficients, and hence of the constants of motion, so they share the main virtue of Vieta's formulae: no root needs to be computed. Standard references are \citet{Uspensky:1948} and, for the quartic criteria, \citet{Rees:1922}.

\subsection{The discriminant}\label{Sec:Discriminant}

For the polynomial $P_n$ of Eq.~\eqref{ndegpoly}, with roots $r_1,\ldots,r_n$, the discriminant is
\begin{equation}\label{GeneralDisc}
 \Delta_n:=a_n^{2n-2}\prod_{i<j}(r_i-r_j)^2\,.
\end{equation}
It is symmetric in the roots and, therefore, by Vieta's formulae, a polynomial in the coefficients, and it vanishes if and only if $P_n$ has a multiple root. For a quadratic, $P_2(x)=a_2x^2+a_1x+a_0$, it is the familiar
\begin{equation}
 \Delta_2=a_1^2-4a_2a_0\,.
\end{equation}
When the coefficients are real and the roots are distinct, the sign of the discriminant counts the complex-conjugate pairs. A pair $z,\bar z$ contributes the factor $(z-\bar z)^2=-4(\operatorname{Im}z)^2<0$, while all the other factors combine into positive quantities: two real roots give a square, a real root $r$ and a complex pair $z,\bar z$ give $(r-z)^2(r-\bar z)^2=|r-z|^4$, and two complex pairs $z,\bar z$ and $w,\bar w$ give $|z-w|^4|z-\bar w|^4$. Hence
\begin{equation}\label{DiscSign}
 \sgn\Delta_n=(-1)^m\,,
\end{equation}
where $m$ is the number of complex-conjugate pairs. For quadratics and cubics, $m$ can only be $0$ or $1$, and the sign of the discriminant determines whether the roots are real. For quartics, $m$ can be $0$, $1$, or $2$, and a positive discriminant does not distinguish between four real roots and two complex-conjugate pairs.

\paragraph{The cubic.} For $P_3(x)=a_3x^3+a_2x^2+a_1x+a_0$, writing the product in terms of the elementary symmetric polynomials and then using Vieta's formulae gives
\begin{equation}\label{CubicDisc}
 \Delta_3=18a_3a_2a_1a_0-4a_2^3a_0+a_2^2a_1^2-4a_3a_1^3-27a_3^2a_0^2\,.
\end{equation}
If the coefficients are real, Eq.~\eqref{DiscSign} gives:
\begin{itemize}
 \item $\Delta_3>0$: three distinct real roots;
 \item $\Delta_3<0$: one real root and a complex-conjugate pair;
 \item $\Delta_3=0$: a multiple root.
\end{itemize}

\paragraph{The quartic.} For $P_4(x)=a_4x^4+a_3x^3+a_2x^2+a_1x+a_0$, the discriminant $\Delta_4$ is a homogeneous polynomial of degree six in the coefficients. By Eq.~\eqref{DiscSign}, $\Delta_4<0$ means exactly one complex-conjugate pair, whereas $\Delta_4>0$ allows either none or two. To tell these apart, one needs the auxiliary quantities
\begin{align}
 \mathcal P&:=8a_4a_2-3a_3^2\,,\\
 \mathcal D&:=64a_4^3a_0-16a_4^2a_2^2+16a_4a_3^2a_2-16a_4^2a_3a_1-3a_3^4\,.
\end{align}
For real coefficients and $\Delta_4\neq0$, the following criteria hold; \citet{Rees:1922} stated them for the reduced quartic $x^4+qx^2+rx+s$, for which $\mathcal P$ and $\mathcal D$ are proportional to $q$ and $s-q^2/4$, with positive factors:
\begin{itemize}
 \item $\Delta_4<0$: two real roots and a complex-conjugate pair;
 \item $\Delta_4>0$, $\mathcal P<0$, and $\mathcal D<0$: four distinct real roots;
 \item $\Delta_4>0$ and ($\mathcal P>0$ or $\mathcal D>0$): two complex-conjugate pairs.
\end{itemize}
In the applications of this book, the auxiliary quantities are rarely needed: whenever a real root is known to exist, for instance from a sign change between $r=0$ and the horizon, two complex-conjugate pairs are impossible, and $\Delta_4>0$ already implies four real roots.

\subsection{Descartes' rule of signs}\label{Sec:Descartes}

The rule appears in the third book of René Descartes's \emph{La Géométrie} (Leiden, 1637), published as an appendix to the \emph{Discours de la méthode}. Descartes gave it without proof, as a statement about how many \enquote{true} (positive) and \enquote{false} (negative) roots an equation \emph{can} have: at most as many true roots as there are changes of sign in the sequence of
coefficients, and at most as many false roots as there are repetitions of sign. When no coefficient vanishes, replacing $x$ by $-x$ turns every change of sign into a repetition and vice versa, so the second half is equivalent to applying the first to $P(-x)$. When coefficients are missing, it can fail: $x^2-1$ has no repetition of sign, yet $x=-1$ is a root. The modern formulation, which counts the sign changes of $P(-x)$, avoids the problem. Newton restated the rule, again without proof, in his \emph{Arithmetica universalis} (1707). A first demonstration was given by Johann Andreas Segner in a dissertation (Jena, 1728) whose title attributes the rule to Thomas Harriot, and the first proof generally regarded as rigorous is due to Jean-Paul de Gua de Malves, in a memoir for the Paris Academy of Sciences (early 1740s). De Gua's argument compares the positive roots of a polynomial with those of its derivative, and the proof given below follows the same idea. Gauss later published a short algebraic proof (\emph{Beweis eines algebraischen
Lehrsatzes}, Journal für die reine und angewandte Mathematik \textbf{3}, 1828). For historical remarks and a sharpened form of the rule, which shows that nothing more can be inferred about the positive roots from the signs of the coefficients alone, see \citet{AlbouyFu2014}.

\begin{theorem}[Descartes' rule of signs]\label{Thm:Descartes}
Let $P(x)$ be a non-zero real polynomial, and let $V(P)$ be the number of sign changes in the sequence of its non-zero coefficients, ordered by degree. Then the number of positive roots of $P$, counted with multiplicity, equals $V(P)$ or is smaller than $V(P)$ by an even integer. Applying the rule to $P(-x)$ yields the same statement for the negative roots of $P$.
\end{theorem}

\begin{proof}
Write $N(P)$ for the number of positive roots of $P$, counted with multiplicity. If $P(x) = x^k Q(x)$ with $Q(0)\neq 0$, then $P$ and $Q$ have the same positive roots and the same sequence of non-zero coefficients, so it is enough to consider $P(0) = a_0 \neq 0$. Let $a_n$ be the leading coefficient.

\emph{Parity.} For $x\to 0^+$, the polynomial $P$ has the sign of $a_0$, and for $x\to+\infty$, the sign of $a_n$. On $(0,\infty)$, it changes sign exactly at its roots of odd multiplicity, so $\operatorname{sgn}(a_0 a_n) = (-1)^{N(P)}$. Running along the non-zero coefficients from $a_0$ to $a_n$, the sign flips $V(P)$ times, so $\operatorname{sgn}(a_0 a_n) = (-1)^{V(P)}$. Hence, $N(P) \equiv V(P) \pmod 2$.

\emph{Bound.} We show $N(P)\le V(P)$ by induction on the degree $n$. For $n=0$, both sides vanish. Let $n\ge 1$, and let $r_1<\dots<r_s$ be the distinct positive roots of $P$, with multiplicities $m_1,\dots,m_s$. Each $r_i$ is a root of $P'$ of multiplicity $m_i-1$, and by Rolle's theorem, $P'$ has at least one further root in each interval $(r_i,r_{i+1})$. Therefore
\begin{equation}
N(P') \ge \sum_{i=1}^{s}(m_i-1) + s - 1 = N(P) - 1\,.
\end{equation}
The coefficients of $P'$ are $j a_j$ with $j\ge 1$: they have the signs of $a_1,\dots,a_n$, so dropping $a_0$ can only remove a sign change, and $V(P')\le V(P)$. By the induction hypothesis applied to $P'$,
\begin{equation}
N(P) \le N(P') + 1 \le V(P') + 1 \le V(P) + 1\,.
\end{equation}
Since $N(P)$ and $V(P)$ have the same parity, $N(P) = V(P)+1$ is impossible, and $N(P)\le V(P)$. Together with the parity statement, this proves the theorem. The statement for negative roots follows because the negative roots of $P$ are the positive roots of $P(-x)$.
\end{proof}

The special case $V(P)=0$ is elementary: if every coefficient is non-negative and at least one is positive, then $P(x)>0$ for all $x>0$. It is used in Chapter~\ref{Sec:Kerrgeod} to show that vortical Kerr geodesics, whether timelike or null, have no radial turning point.

\paragraph{Real-rooted polynomials.} When all the roots are real, the rule is exact. Let $P(x)=\sum_{k=0}^n a_k x^k$ have degree $n$, only real roots, and $P(0)=a_0\neq0$. If $N_+$ and $N_-$ are the numbers of positive and negative roots, counted with multiplicity, then $N_++N_-=n$, while Theorem~\ref{Thm:Descartes} gives $N_+\leq V(P)$ and $N_-\leq V(P(-x))$.

We now show that $V(P)+V(P(-x))\leq n$. Let $0=k_0<k_1<\dots<k_m=n$ be the degrees of the non-zero coefficients of $P$;
the first is $0$ because $a_0\neq0$, and the last is $n$ because $a_n\neq0$. The polynomial $P(-x)=\sum_k(-1)^ka_kx^k$ has its non-zero coefficients in the same degrees, so both sign sequences are made of the same $m$ consecutive pairs $(a_{k_{l-1}},a_{k_l})$, $l=1,\dots,m$. Let $d_l=k_l-k_{l-1}\geq1$ be the gap in degree within the $l$-th pair. In $P(x)$ the pair produces a sign change
if $a_{k_{l-1}}a_{k_l}<0$; in $P(-x)$ the corresponding product is $(-1)^{d_l}a_{k_{l-1}}a_{k_l}$. Two cases arise.
\begin{itemize}
\item If $d_l$ is odd, the two products have opposite signs, so the pair produces a sign change in exactly one of the two sequences: one change in total, and $1\leq d_l$.
\item If $d_l$ is even, the two products have the same sign, so the pair produces a sign change in both sequences or in neither: at most two changes in total, and $2\leq d_l$ because $d_l\geq2$.
\end{itemize}
In both cases, the $l$-th pair contributes at most $d_l$ sign changes to $V(P)+V(P(-x))$. Summing over the pairs, the gaps telescope:
\begin{equation}
V(P)+V(P(-x)) \leq \sum_{l=1}^m d_l = k_m-k_0 = n\,.
\end{equation}
For example, $x^2-1$ has a single pair with $d=2$, and both $x^2-1$ and $P(-x)=x^2-1$ have one sign change: the bound $n=2$ is attained.

Putting everything together,
\begin{equation}
n = N_++N_- \leq V(P)+V(P(-x)) \leq n\,,
\end{equation}
so $N_++N_-=V(P)+V(P(-x))$. Since $N_+\leq V(P)$ and $N_-\leq V(P(-x))$ separately, neither inequality can be strict: $N_+=V(P)$ and
$N_-=V(P(-x))$.

Applied to the shifted polynomial $P(c+h)=\sum_kP^{(k)}(c)\,h^k/k!$, this gives the form used in the following chapters: \emph{if all roots of $P$ are real and $P(c)\neq0$, the number of roots larger than $c$ equals the number of sign changes in the sequence of Taylor coefficients $P^{(k)}(c)/k!$, ordered by degree}. With $c=r_{\rm H}$, it counts the roots of the Kerr radial polynomial outside the event horizon (Section~\ref{Sec:KerrBound}). \citet{Wilkins:1972rs} used the rule in this shifted form, with $c=M$, to prove that Kerr geodesics with $\varepsilon^2\geq1$ cannot be bound.

\section{Root patterns}\label{Sec:RootPatterns}

By the \emph{root pattern} of a polynomial, we mean the qualitative configuration of its roots: how many are real and how many form complex-conjugate pairs, the signs and ordering of the real roots, their multiplicities, and their position relative to the boundary of the physical domain. Examples are $0<r_1<r_2<r_3$, or one negative root and a complex-conjugate pair. The classification of geodesics in this book is, at its core, a classification of the possible root patterns of the radial and polar polynomials.

For a polynomial with real coefficients, Vieta's formulae restrict the admissible patterns in three ways, none of which requires computing a root.
\begin{enumerate}
 \item \textbf{Signs.} Suppose $a_0\neq0$, so that no root vanishes, and let $N_-$ be the number of negative real roots, counted with multiplicity. In the product of all the roots, a complex-conjugate pair $z,\bar z$ contributes the factor $z\bar z=|z|^2>0$ and a positive real root contributes a positive factor, so only the negative roots affect the sign:
 \begin{equation}
 \sgn\Bigl(\prod_i r_i\Bigr)=\sgn\left[(-1)^n\frac{a_0}{a_n}\right] = (-1)^{N_-}\,.
 \end{equation}
 The sign of the product, therefore, tells us whether the number of negative roots is even or odd, though not the number itself. For instance, the roots of the cubic $x^3-b^2x+2Mb^2$ have product $-2Mb^2<0$, so there is an odd number of negative roots: one positive root accompanied by a complex-conjugate pair is impossible.
 \item \textbf{Excluded patterns.} A proposed assignment of signs must satisfy all the relations at once. Several relations taken together can exclude a pattern that each of them separately would allow; examples appear in Chapters~\ref{Sec:Schgeod} and~\ref{Sec:Kerrgeod}.
 \item \textbf{Reality.} The sum of the squares of the roots is a symmetric function, and
 \begin{equation}
 \sum_i r_i^2=s_1^2-2s_2\,.
 \end{equation}
 If the right-hand side is negative, the roots cannot all be real. The converse does not hold.
\end{enumerate}
Vieta's formulae are thus necessary constraints on a proposed root pattern: they can rule a pattern out, but a pattern compatible with them need not correspond to a geodesic. The roots may lie outside the physical domain, the pattern may not be realized for any admissible values of the constants of motion, or the other separated equation may forbid motion altogether. In the following chapters, Vieta's formulae will therefore always be used together with the discriminant and Descartes' rule of signs (Section~\ref{Sec:Discriminants}), which decide whether the roots are real and where they lie, and with the inequalities that define the physically allowed radial and polar regions.

\section{Turning points}\label{Sec:TurningPoints}

In all the spacetimes studied in this book, the motion in each separated coordinate $x$ (the radius $r$, or the polar variable $z=\cos^2\theta$) is governed by an equation of the form
\begin{equation}\label{GenericFirstIntegral}
 \left(\frac{\dd x}{\dd\sigma}\right)^2=P(x)\,,
\end{equation}
where $\sigma$ is a suitable parameter along the geodesic and $P$ is a smooth function, in practice a polynomial, on the coordinate domain $D$. Since the left-hand side is non-negative, the motion is confined to the \emph{allowed region}
\begin{equation}
 \mathcal A:=\{x\in D:\ P(x)\geq0\}\,,
\end{equation}
and, more precisely, to the connected component of $\mathcal A$ containing the initial position. Wherever $\dd x/\dd\sigma\neq0$, differentiating Eq.~\eqref{GenericFirstIntegral} gives
\begin{equation}\label{GenericSecondOrder}
 \frac{\dd^2x}{\dd\sigma^2}=\frac12\,P'(x)\,,
\end{equation}
where $P'(x) = \dd P(x)/\dd x$, and this second-order equation, which in the applications follows directly from the geodesic equation, holds along the entire motion.

\begin{definition}
A \emph{turning point} is a point $x_\ast\in D$ such that
\begin{equation}
 P(x_\ast)=0\,,\qquad P'(x_\ast)\neq0\,,
\end{equation}
that is, a simple root of $P$ lying in the coordinate domain.
\end{definition}

The name is justified by the behavior of the motion near $x_\ast$. Expanding $P(x)\simeq P'(x_\ast)(x-x_\ast)$, the parameter needed to reach $x_\ast$ from a nearby point $x_0$,
\begin{equation}
 \sigma_\ast-\sigma_0=\int_{x_0}^{x_\ast}\frac{\dd x}{\sqrt{P(x)}}
 \simeq\int_{x_0}^{x_\ast}\frac{\dd x}{\sqrt{P'(x_\ast)(x-x_\ast)}}\,,
\end{equation}
is finite. At $x_\ast$, the velocity vanishes while, by Eq.~\eqref{GenericSecondOrder}, the acceleration $P'(x_\ast)/2$ does not, and it points into the allowed region. Hence $x(\sigma)\simeq x_\ast+\tfrac14P'(x_\ast)(\sigma-\sigma_\ast)^2$: the coordinate reaches $x_\ast$, stops for an instant, and reverses. A bounded component $[x_1,x_2]$ of $\mathcal A$ with a turning point at each
end and $P>0$ in between therefore gives an oscillation, periodic in $\sigma$, between them, with period
\begin{equation}
 T=2\int_{x_1}^{x_2}\frac{\dd x}{\sqrt{P(x)}}\,,
\end{equation}
which is finite because $1/\sqrt{P}$ is integrable at a simple root. If $P$ has a double root inside the component, the integral diverges there, and the motion approaches the corresponding unstable circular orbit asymptotically
instead (cf.\ the separatrix in Section~\ref{Sec:EffPot}).

Multiple roots behave differently. If $P(x_\ast)=P'(x_\ast)=0$, then $x(\sigma)=x_\ast$ is itself a solution of both Eqs.~\eqref{GenericFirstIntegral} and \eqref{GenericSecondOrder}: a motion at constant $x$, such as a circular or spherical orbit. Near a double root, $P(x)\simeq\tfrac12P''(x_\ast)(x-x_\ast)^2$ and the integral for $\sigma_\ast-\sigma_0$ diverges logarithmically, so $x_\ast$ is never reached from a different starting point. Two cases arise.
\begin{itemize}
 \item If $P''(x_\ast)<0$, the point $x_\ast$ is an isolated point of $\mathcal A$: the constant solution is \emph{stable}, and nearby motions with slightly different constants oscillate about it.
 \item If $P''(x_\ast)>0$, the allowed region extends on both sides of $x_\ast$: the constant solution is \emph{unstable}, and there are motions that approach it asymptotically, as $\sigma\to\pm\infty$, without ever reaching it. These separate qualitatively different families of orbits.
\end{itemize}
A triple root marks the transition between the two cases and, for circular orbits, the marginally stable orbit.

Finally, an endpoint of $\mathcal A$ need not be a root of $P$ at all. When the allowed region extends to a horizon or to infinity, the motion crosses the horizon or escapes, and there is no turning point at that end.

In the applications, $P$ is often replaced by $\tilde P=g(x)P$ with $g>0$ on $D$, for instance $r^3\dot r^2$ instead of $\dot r^2$ in Schwarzschild spacetime, to obtain a polynomial. This changes neither the roots in $D$ nor their multiplicities, nor the sign of the function, so the definitions above apply to $\tilde P$ unchanged.

\chapter{Geodesics in Schwarzschild spacetime}\label{Sec:Schgeod}

The Schwarzschild solution \citep{Schwarzschild:1916uq} is, by Birkhoff's theorem, locally the unique spherically symmetric vacuum solution of Einstein's equations \citep[see, e.g.,][]{weinberg1972gravitation,FrolovZelnikov}. In its black-hole interpretation, the exterior region is $r>r_{\rm H}=2M$, where $r_{\rm H}$ denotes the event horizon. The geodesic problem has been studied since the earliest years of general relativity. Classic analyses include those of \citet{Hagihara:1931}, who classified the orbits according to the distribution of the roots of the radial cubic, and of \citet{Darwin:1959,Darwin:1961}, who located the turning points using the rule of signs and the sum of the roots. Several of the algebraic results of this chapter appear, in a different notation, in their work.

The chapter is organized as follows. Section~\ref{Sec:SchwEqs} derives the geodesic equations and reduces the radial motion to a one-dimensional problem. Section~\ref{Sec:EffPot} explains how this problem is solved qualitatively with an effective potential, and Section~\ref{Sec:SchwTimelikePot} applies the method to massive particles. Section~\ref{Sec:SchwVieta} obtains the same classification from Vieta's formulae, and Section~\ref{Sec:SchwNull} treats photons using both methods.

\section{Metric, geodesic equations, and constants of motion}\label{Sec:SchwEqs}

In Schwarzschild coordinates and geometrized units, the line element is
\begin{equation}\label{Schwmetric}
 \dd s^2=
 -f(r)\,\dd t^2
 +\frac{\dd r^2}{f(r)}
 +r^2\left(\dd\theta^2+\sin^2\theta\,\dd\varphi^2\right)\,,
 \qquad
 f(r):=1-\frac{2M}{r}\,,
\end{equation}
where $M>0$. Let $\lambda$ be the proper time for a massive particle and an affine parameter for a photon, and denote derivatives with respect to $\lambda$ by a dot. With a prime denoting $\dd/\dd r$, so that $f'=2M/r^2$, the non-vanishing Christoffel symbols are
\begin{gather}
 \Gamma^t_{tr}=\frac{f'}{2f}\,,\qquad
 \Gamma^r_{tt}=\frac{ff'}{2}\,,\qquad
 \Gamma^r_{rr}=-\frac{f'}{2f}\,,\qquad
 \Gamma^r_{\theta\theta}=-rf\,,\qquad
 \Gamma^r_{\varphi\varphi}=-rf\sin^2\theta\,,\nonumber\\
 \Gamma^\theta_{r\theta}=\Gamma^\varphi_{r\varphi}=\frac1r\,,\qquad
 \Gamma^\theta_{\varphi\varphi}=-\sin\theta\cos\theta\,,\qquad
 \Gamma^\varphi_{\theta\varphi}=\cot\theta\,,
\end{gather}
and the geodesic equation $\ddot x^\mu+\Gamma^\mu_{\nu\rho}\dot x^\nu\dot x^\rho=0$ gives one equation for each coordinate:
\begin{align}
 &\ddot t+\frac{f'}{f}\,\dot t\,\dot r=0\,,\label{SchwGeodT}\\
 &\ddot r+\frac{ff'}{2}\,\dot t^2-\frac{f'}{2f}\,\dot r^2
 -rf\left(\dot\theta^2+\sin^2\theta\,\dot\varphi^2\right)=0\,,\label{SchwGeodR}\\
 &\ddot\theta+\frac2r\,\dot r\,\dot\theta-\sin\theta\cos\theta\,\dot\varphi^2=0\,,\label{SchwGeodTheta}\\
 &\ddot\varphi+\frac2r\,\dot r\,\dot\varphi+2\cot\theta\,\dot\theta\,\dot\varphi=0\,.\label{SchwGeodPhi}
\end{align}

\paragraph{Planar motion.} Because of spherical symmetry, we may rotate the coordinates so that the initial position and velocity of the particle lie in the equatorial plane, $\theta=\pi/2$ and $\dot\theta=0$ at $\lambda=0$. Equation~\eqref{SchwGeodTheta} then gives $\ddot\theta=0$ at $\lambda=0$, and $\theta(\lambda)=\pi/2$ solves it for all $\lambda$; by the uniqueness of solutions, the particle never leaves the plane. We set $\theta=\pi/2$ from now on.

\paragraph{Constants of motion.} With $\theta=\pi/2$, Eqs.~\eqref{SchwGeodT} and \eqref{SchwGeodPhi} are total derivatives,
\begin{equation}
 \frac{\dd}{\dd\lambda}\left(f\,\dot t\right)=0\,,
 \qquad
 \frac{\dd}{\dd\lambda}\left(r^2\dot\varphi\right)=0\,,
\end{equation}
as one checks using $\dd f/\dd\lambda=f'\dot r$. They integrate to
\begin{align}
 \frac{\dd t}{\dd\lambda}
 &=\frac{\varepsilon}{f(r)}\,,\label{tschweq}\\
 \frac{\dd\varphi}{\dd\lambda}
 &=\frac{L}{r^2}\,,\label{phischweq}
\end{align}
where $\varepsilon$ and $L$ are integration constants. This is no accident: they are the conserved quantities $\varepsilon=-g_{\mu\nu}(\partial_t)^\mu\dot x^\nu$ and $L=g_{\mu\nu}(\partial_\varphi)^\mu\dot x^\nu$ associated with the Killing vectors $\partial_t$ and $\partial_\varphi$, reflecting the fact that the metric depends on neither $t$ nor $\varphi$. The minus sign makes $\varepsilon$ positive for future-directed motion outside the horizon.

For a massive particle, far from the black hole $\dot t$ reduces to the Lorentz factor of the particle relative to the static observers, so $\varepsilon$ is the energy per unit rest mass measured at infinity, rest energy included: a particle with $\varepsilon<1$ cannot reach infinity, and one with $\varepsilon=1$ reaches it with zero velocity. Similarly, $L$ is the angular momentum per unit rest mass, reducing to the Newtonian $r^2\dot\varphi$ in the weak-field limit. For a photon, the affine parameter is defined only up to $\lambda\to\alpha\lambda+\beta$, under which $\varepsilon$ and $L$ are both multiplied by $1/\alpha$. Only their ratio
\begin{equation}
 b:=\frac{L}{\varepsilon}
\end{equation}
is independent of this convention. It is the \emph{impact parameter}: the distance between the center and the straight line along which the photon travels far from the black hole.

\paragraph{The radial equation.} Substituting Eqs.~\eqref{tschweq} and \eqref{phischweq} into Eq.~\eqref{SchwGeodR} leaves a second-order equation for $r(\lambda)$ alone,
\begin{equation}\label{SchwRadialSecond}
 \ddot r+\frac{f'\left(\varepsilon^2-\dot r^2\right)}{2f}
 -\frac{fL^2}{r^3}=0\,.
\end{equation}
We need not integrate it, however, because the line element itself already provides a first integral. Along the geodesic,
\begin{equation}\label{constrlambda}
 -g_{\mu\nu}\dot x^\mu\dot x^\nu=\kappa\,,
 \qquad
 \kappa=
 \begin{cases}
 1,&\text{timelike geodesics}\,,\\
 0,&\text{null geodesics}\,,
 \end{cases}
\end{equation}
which for $\kappa=1$ is the statement that $\lambda$ is proper time, and for $\kappa=0$ that the curve is null. Inserting the metric and the two constants of motion gives
\begin{equation}\label{SchwRadialFirst}
 \left(\frac{\dd r}{\dd\lambda}\right)^2
 =\varepsilon^2-
 f(r)
 \left(\kappa+\frac{L^2}{r^2}\right)\,,
\end{equation}
or, in the form of an energy equation,
\begin{equation}\label{SchwEnergyEquation}
 \frac12\,\dot r^2+V_\kappa(r)
 =\frac{\varepsilon^2}{2}\,,
 \qquad
 V_\kappa(r):=\frac12\,
 f(r)
 \left(\kappa+\frac{L^2}{r^2}\right)\,.
\end{equation}
Equation~\eqref{SchwRadialFirst} is indeed a first integral of Eq.~\eqref{SchwRadialSecond}: using it to eliminate $\varepsilon^2-\dot r^2$ from Eq.~\eqref{SchwRadialSecond} gives
\begin{equation}\label{SchwRadialForce}
 \ddot r
 =-\frac{\dd V_\kappa}{\dd r}
 =-\frac{\kappa M}{r^2}+\frac{L^2}{r^3}
 -\frac{3ML^2}{r^4}\,.
\end{equation}
The two equations are not quite equivalent. Equation~\eqref{SchwEnergyEquation} tells us where the motion is allowed, $V_\kappa(r)\leq\varepsilon^2/2$, and where it turns back, but at a radius where $\dot r=0$ it cannot tell whether the particle turns back or stays there. That information is in Eq.~\eqref{SchwRadialForce}: a circular orbit, $\dot r=0$ for all $\lambda$, requires $\ddot r=0$ as well, that is, $V_\kappa'(r)=0$. This is the origin of the double-root condition used throughout the book.

The analogy with a one-dimensional Newtonian particle in the potential $V_\kappa$ is exact at the level of the radial equation, but only there: $\lambda$ is proper time or an affine parameter, not the time $t$ measured at infinity, and the term $-3ML^2/r^4$ has no Newtonian counterpart. It is this term that makes the potential turn over at small radii and allows capture.

\paragraph{Remark.} Nothing in this section used the explicit form of $f(r)$ before Eq.~\eqref{SchwRadialForce}. For any metric of the form~\eqref{Schwmetric} with a general function $f(r)$, the motion is planar, $\varepsilon=f\,\dot t$ and $L=r^2\dot\varphi$ are conserved, and the radial motion obeys
\begin{equation}
 \dot r^2=\varepsilon^2-f(r)\left(\kappa+\frac{L^2}{r^2}\right)\,,
 \qquad
 \ddot r=-\frac{\dd V_\kappa}{\dd r}\,,
 \qquad
 V_\kappa=\frac12\,f(r)\left(\kappa+\frac{L^2}{r^2}\right)\,.
\end{equation}
When $f$ is a rational function of $r$, multiplying by a suitable power of $r$ turns the first integral into a polynomial equation, and the method of this book applies. For the Reissner--Nordstr\"om metric, $f=1-2M/r+Q^2/r^2$, the radial polynomial is quartic for both timelike and null geodesics. For the Schwarzschild--de~Sitter metric, $f=1-2M/r-\Lambda r^2/3$, the null radial polynomial is the Schwarzschild cubic with $\varepsilon^2$ replaced by $\varepsilon^2+\Lambda L^2/3$, so the cosmological constant leaves the photon sphere at $r=3M$ unchanged (Chapter~\ref{Sec:SdS}).

\section{The effective-potential method}\label{Sec:EffPot}

Equation~\eqref{SchwEnergyEquation} has exactly the form of the energy equation of a Newtonian particle of unit mass moving on a line in the potential $V_\kappa(r)$, with total energy $\varepsilon^2/2$. This observation makes it possible to read off the qualitative behavior of every radial motion from a single graph, without solving any equations. The method, familiar from Newtonian celestial mechanics, goes as follows; it is illustrated in Fig.~\ref{SchwVschematic}.
\begin{enumerate}
 \item Draw the graph of $V_\kappa(r)$ for the given value of $L$, and the horizontal line at height $\varepsilon^2/2$. Since $\varepsilon$ is conserved, the particle stays on this line for its whole history.
 \item Since $\tfrac12\dot r^2=\varepsilon^2/2-V_\kappa(r)\geq0$, the particle can only be where the graph lies on or below the line. These are the \emph{allowed regions}. Where the graph lies above the line, the region is classically forbidden: the particle would need negative kinetic energy to be there, and the potential acts as a \emph{barrier}.
 \item Where the line crosses the graph, $\dot r=0$. At a crossing with $V_\kappa'\neq0$, the acceleration $\ddot r=-V_\kappa'$ does not vanish and points back into the allowed region: this is a \emph{turning point}, a simple root of $\varepsilon^2-2V_\kappa$ in the sense of Section~\ref{Sec:TurningPoints}. The connected components of the allowed region, bounded by turning points, by the horizon or by infinity, are the \emph{branches} of the motion. A particle stays on one branch for its entire exterior history.
 \item Where the line is tangent to the graph at a stationary point $r_c$, so that $V_\kappa(r_c)=\varepsilon^2/2$ and $V_\kappa'(r_c)=0$, the particle can stay at $r=r_c$ forever: this is a \emph{circular orbit}. At a minimum of $V_\kappa$ the circular orbit is \emph{stable}: a small change of energy produces small radial oscillations about $r_c$. At a maximum it is \emph{unstable}: a small perturbation makes the particle roll off the top of the barrier, and particles with exactly this energy approach $r_c$ asymptotically without ever reaching it. These correspond to the double roots of Section~\ref{Sec:TurningPoints}.
 \item Finally, the behavior at the ends of the domain decides whether a branch reaches the horizon or infinity. Since $f(r_{\rm H})=0$, the potential vanishes at the horizon, so a neighborhood of the horizon is always allowed. At large radii, $V_\kappa$ tends to $\kappa/2$, so a massive particle can reach infinity only if $\varepsilon^2\geq1$.
\end{enumerate}
The qualitative classification of the radial motion therefore reduces to three pieces of information: the positions of the stationary points of $V_\kappa$, the values of $V_\kappa$ there, and its asymptotic value. The method is powerful and intuitive, and its results will serve as a benchmark for the algebraic approach of Section~\ref{Sec:SchwVieta}. It works so well here because the energy enters the radial equation only through the constant $\varepsilon^2/2$, while the potential depends on $L$ alone: for each value of $L$ there is a single curve, and all energies are dealt with by moving a horizontal line up and down.

\begin{figure}[htbp]
\centering
\includegraphics{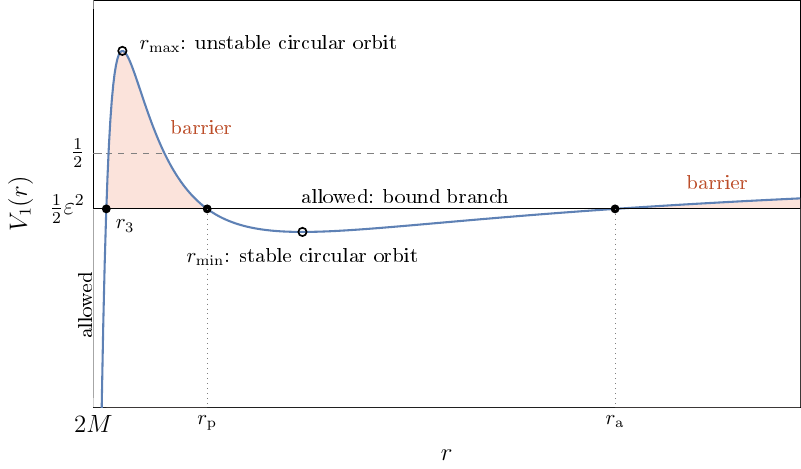}
 \caption{The effective-potential method, for a timelike geodesic with $L=4.3M$. The horizontal line is the energy level $\varepsilon^2/2$; the potential barriers, where $V_1>\varepsilon^2/2$, are shaded. The line meets the graph at three turning points, which bound two branches: a horizon-connected branch, $2M<r\leq r_3$, and a bound branch, $r_{\rm p}\leq r\leq r_{\rm a}$. Raising or lowering the line changes the turning points. When the line touches the graph at the maximum $r_{\max}$ or the minimum $r_{\min}$ (open circles), the particle can move on an unstable or a stable circular orbit, respectively.}
 \label{SchwVschematic}
\end{figure}

\section{Timelike geodesics: effective-potential analysis}\label{Sec:SchwTimelikePot}

For $\kappa=1$, the effective potential is
\begin{equation}\label{Schwpotentialmassive}
 V_1(r)=\frac12\left(1+\frac{L^2}{r^2}\right)
 \left(1-\frac{2M}{r}\right)\,.
\end{equation}
It vanishes at the horizon and approaches $1/2$ from below as $r\to\infty$. Following the method of Section~\ref{Sec:EffPot}, we first find its stationary points, which are the circular orbits, and then classify the motion.

\subsection{Stationary points and circular orbits}\label{Sec:SchwCircular}

The first two derivatives of the potential are
\begin{align}
 V_1'(r)&=\frac{M}{r^2}-\frac{L^2}{r^3}
 +\frac{3ML^2}{r^4}\,,\\
 V_1''(r)&=-\frac{2M}{r^3}+\frac{3L^2}{r^4}
 -\frac{12ML^2}{r^5}\,.
\end{align}
Stationary points satisfy $Mr^2-L^2r+3ML^2=0$, and are
\begin{equation}\label{SchwExtrema}
 r_{\max}
 =\frac{L^2}{2M}
 \left(1-\sqrt{1-\frac{12M^2}{L^2}}\right)\,,
 \qquad
 r_{\min}
 =\frac{L^2}{2M}
 \left(1+\sqrt{1-\frac{12M^2}{L^2}}\right)\,.
\end{equation}
For $L^2<12M^2$ there are no stationary points and the potential increases monotonically in the exterior region. For $L^2>12M^2$, the potential has a maximum at $r_{\max}$ and a minimum at $r_{\min}>r_{\max}$. As $L$ grows, $r_{\max}$ decreases toward $3M$, while $r_{\min}$ grows without bound, approximately as $L^2/M$, which is the radius of a Newtonian circular orbit with angular momentum $L$.

Each stationary point is the radius of a circular orbit, for the energy that puts the horizontal line at the level of the stationary point. It is instructive to reverse the logic and ask which constants of motion give a circular orbit at a prescribed radius $r_c$. The two conditions $V_1(r_c)=\varepsilon^2/2$ and $V_1'(r_c)=0$ give
\begin{equation}\label{SchwCircularConstants}
 L^2=\frac{Mr_c^2}{r_c-3M}\,,
 \qquad
 \varepsilon^2=\frac{(r_c-2M)^2}{r_c(r_c-3M)}\,,
\end{equation}
which are plotted in Fig.~\ref{SchwCircularFig}. Circular timelike orbits therefore exist only for $r_c>3M$. As $r_c\to3M$, both $L$ and $\varepsilon$ diverge: the orbiting particle approaches the speed of light, and the circular orbit approaches the photon sphere of Section~\ref{Sec:SchwNull}. Substituting Eq.~\eqref{SchwCircularConstants} into $V_1''$ gives
\begin{equation}
 V_1''(r_c)=\frac{M\,(r_c-6M)}{r_c^3\,(r_c-3M)}\,,
\end{equation}
so circular orbits are unstable for $3M<r_c<6M$, where they sit at the maximum $r_{\max}$, and stable for $r_c>6M$, where they sit at the minimum $r_{\min}$.

\paragraph{The innermost stable circular orbit.} Both constants in Eq.~\eqref{SchwCircularConstants} are minimal at $r_c=6M$. At $L^2=12M^2$ the maximum and the minimum of the potential merge into an inflection point, at $r=6M$. Since $r_{\min}$ decreases toward $6M$ as $L^2$ decreases toward $12M^2$, this radius is the smallest one at which a stable circular orbit exists. It is therefore called the \emph{innermost stable circular orbit} (ISCO):
\begin{equation}
 r_{\rm ISCO}=6M\,,\qquad
 L_{\rm ISCO}^2=12M^2\,,\qquad
 \varepsilon_{\rm ISCO}^2=\frac89\,.
\end{equation}
A particle on the ISCO has lost the fraction $1-\sqrt{8/9}\simeq5.7\%$ of its rest energy, as measured from infinity. This is the efficiency with which matter spiraling slowly through a sequence of circular orbits, as in a thin accretion disk, radiates its rest energy before plunging into the black hole.

\paragraph{The marginally bound orbit.} The circular orbit with $\varepsilon=1$ has the energy of a particle at rest at infinity. Equation~\eqref{SchwCircularConstants} gives $\varepsilon^2=1$ for $r_c=4M$, so the \emph{marginally bound} circular orbit is
\begin{equation}
 r_{\rm mb}=4M\,,\qquad L_{\rm mb}=4M\,,\qquad
 \varepsilon_{\rm mb}=1\,.
\end{equation}
Since $4M<6M$, it is unstable. In terms of the potential, for $L=4M$ the maximum lies exactly at the asymptotic value $1/2$. For $L>4M$ the maximum rises above $1/2$, since $\varepsilon^2(r_c)>1$ for $3M<r_c<4M$, and a particle coming from infinity can be reflected by the barrier; for $12M^2<L^2<16M^2$ the maximum lies below $1/2$.

For fixed $L$ such that $L^2 > 12M^2$, a particle with energy exactly $\varepsilon^2=2V_{\max}$ does not reach a turning point on the barrier: it approaches the unstable circular orbit at $r_{\max}$ asymptotically, winding around it infinitely many times. These asymptotic orbits form the \emph{separatrix}, the boundary between qualitatively different motions. When the unstable orbit lies between $4M$ and $6M$, so that $2V_{\max}<1$, the separatrix divides bound orbits from orbits that plunge into the black hole. When it lies between $3M$ and $4M$, so that $2V_{\max}>1$, it divides scattering from capture. In the plane of the constants of motion of Fig.~\ref{SchwParamSpace}, the separatrix is the red curve; in terms of the parameters of a bound orbit, it is the line $p=6+2e$ derived in Section~\ref{Sec:SchwSeparatrix}. Orbits close to the separatrix spend a long time near the unstable circular orbit before moving away, and display the familiar zoom-whirl behavior \citep{Darwin:1961,LevinPerezGiz:2008}; an example is shown in Fig.~\ref{SchwOrbits}(b).

The four qualitatively different shapes of the potential, with no stationary points, with an inflection point at the ISCO, with a maximum at the asymptotic value $1/2$, and with a maximum above it, are illustrated in Fig.~\ref{SchwV}.

\begin{figure}[htbp]
\centering
\includegraphics{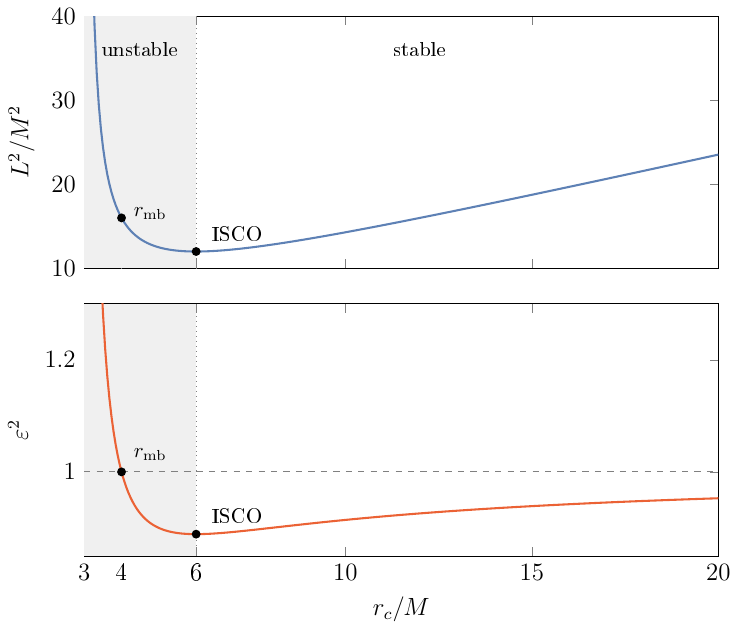}
 \caption{Angular momentum (top) and energy (bottom) of circular timelike orbits as functions of their radius, Eq.~\eqref{SchwCircularConstants}. Both diverge as $r_c\to3M$ and have their minimum at the ISCO, $r_c=6M$. Orbits in the shaded region, $3M<r_c<6M$, are unstable. The marginally bound orbit $r_{\rm mb}=4M$ has $\varepsilon^2=1$: circular orbits inside it have $\varepsilon>1$ and are unbound.}
 \label{SchwCircularFig}
\end{figure}

\begin{figure}[htbp]
\centering
 \includegraphics{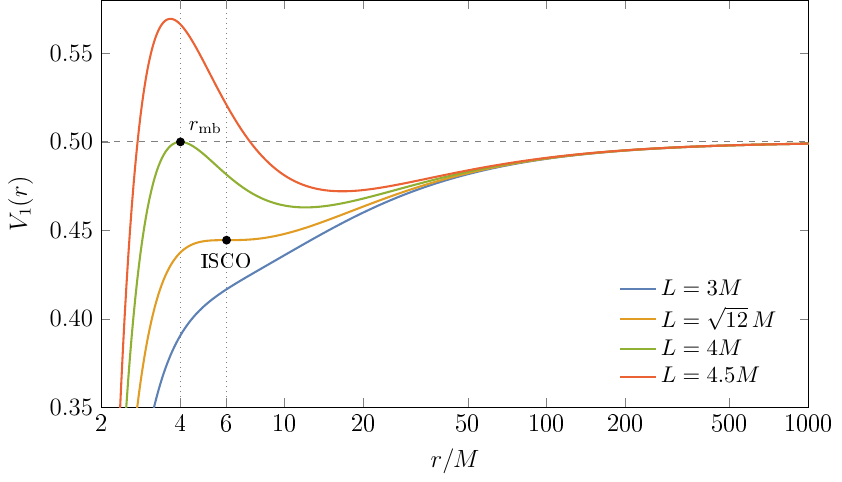}
 \caption{Effective potential $V_1(r)$ for timelike geodesics, Eq.~\eqref{Schwpotentialmassive}, for four values of the angular momentum, in the exterior region $r>2M$. Motion with energy $\varepsilon$ is allowed where $V_1(r)\leq\varepsilon^2/2$. The dashed line is the asymptotic value $V_1\to1/2$, which corresponds to $\varepsilon=1$ and separates bound from unbound motion. For $L=3M$ the potential increases monotonically and there are no circular orbits. For $L=\sqrt{12}\,M$ the maximum and minimum merge into an inflection point at the innermost stable circular orbit, $r_{\rm ISCO}=6M$, where $\varepsilon^2=8/9$. For $L=4M$ the maximum lies exactly at $V_1=1/2$, at the marginally bound radius $r_{\rm mb}=4M$: a particle falling from rest at infinity approaches this unstable circular orbit asymptotically. For $L=4.5M$ the maximum rises above $1/2$, so particles coming from infinity with energy slightly above $\varepsilon=1$ are scattered by the potential barrier, while the minimum gives a stable circular orbit.}
 \label{SchwV}
\end{figure}

\subsection{Classification of the motion}\label{Sec:SchwClassification}

Whenever the extrema exist, that is, for $L^2>12M^2$, let
\begin{equation}
 V_{\max}:=V_1(r_{\max})\,,
 \qquad
 V_{\min}:=V_1(r_{\min})
\end{equation}
denote the height of the potential barrier and the bottom of the potential well. Comparing the energy level $\varepsilon^2/2$ with $V_{\max}$, $V_{\min}$, and the asymptotic value $1/2$ gives all the generic exterior motions, summarized in Table~\ref{tab:SchwPotentialClassification}. The same information is displayed graphically in the plane of the constants of motion in Fig.~\ref{SchwParamSpace}.

\begin{table}[htbp]
\centering
\small
\begin{tabular}{@{}lll@{}}
\toprule
Angular momentum & Energy range & Exterior radial motion\\
\midrule
$L^2<12M^2$ & $\varepsilon^2\geq1$ & no turning point\\
 & $\varepsilon^2<1$ & horizon-connected branch only\\[2pt]
$12M^2<L^2<16M^2$
 & $\varepsilon^2\geq1$
 & no turning point\\
 & $2V_{\max}<\varepsilon^2<1$
 & horizon-connected branch only\\
 & $2V_{\min}<\varepsilon^2<2V_{\max}$
 & horizon-connected and bound branches\\
 & $\varepsilon^2<2V_{\min}$
 & horizon-connected branch only\\[2pt]
$L^2>16M^2$
 & $\varepsilon^2>2V_{\max}$ & no turning point\\
 & $1<\varepsilon^2<2V_{\max}$
 & horizon-connected and scattering branches\\
 & $2V_{\min}<\varepsilon^2<1$
 & horizon-connected and bound branches\\
 & $\varepsilon^2<2V_{\min}$
 & horizon-connected branch only\\
\bottomrule
\end{tabular}
\caption{Generic timelike motions in Schwarzschild spacetime. Here $V_{\max}$ and $V_{\min}$ are the values of the potential at its maximum and minimum. For $12M^2<L^2<16M^2$ the whole potential lies below its asymptotic value $1/2$: an energy level above the barrier but with $\varepsilon^2<1$ still meets the potential at large radius, so the particle has an apoapsis, which moves to infinity as $\varepsilon^2\to1$. The boundary cases $\varepsilon^2=2V_{\max}$ and $\varepsilon^2=2V_{\min}$ correspond to circular orbits and to orbits approaching them asymptotically.}
\label{tab:SchwPotentialClassification}
\end{table}

\begin{figure}[htbp]
\centering
\includegraphics{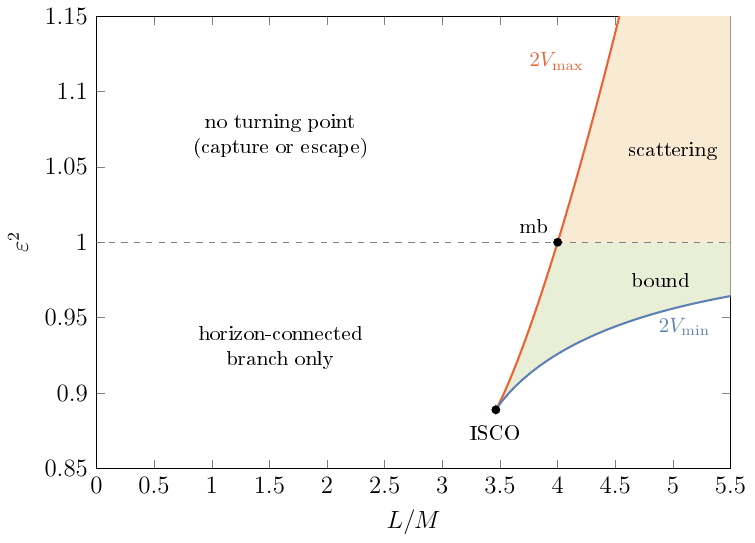}
 \caption{Table~\ref{tab:SchwPotentialClassification} in the plane of the constants of motion. The red and blue curves are the energies $2V_{\max}$ and $2V_{\min}$ of the unstable and stable circular orbits, parametrized by $r_c$ through Eq.~\eqref{SchwCircularConstants}; they meet in a cusp at the ISCO and cross the line $\varepsilon^2=1$ at the marginally bound orbit (mb). In the green region the particle may move on a bound branch, in the orange region on a scattering branch; in both cases a disconnected horizon-connected branch exists as well. There is no turning point when $\varepsilon^2$ lies above both the red curve and the line $\varepsilon^2=1$. In the rest of the diagram only a horizon-connected branch exists; for $\sqrt{12}\,M<L<4M$, this includes the strip between the red curve and the line $\varepsilon^2=1$, where the particle passes over the barrier but still has an apoapsis. Crossing either curve changes the number of real roots of the radial polynomial: the curves are the zero set of its discriminant, Eq.~\eqref{SchwDisc}.}
 \label{SchwParamSpace}
\end{figure}

The turning points bounding a branch are called the \emph{periapsis} $r_{\rm p}$, the smallest radius the particle reaches, and the \emph{apoapsis} $r_{\rm a}$, the largest. Three kinds of branches occur; examples are shown in Fig.~\ref{SchwOrbits}.
\begin{description}
 \item[Bound branch,] $r_{\rm p}\leq r\leq r_{\rm a}$. The particle oscillates radially between periapsis and apoapsis, while it orbits the black hole. This is the relativistic counterpart of a Keplerian ellipse, although the orbit does not close, as its periapsis precesses.
 \item[Scattering branch,] $r\geq r_{\rm p}$. The particle comes in from infinity, reaches the periapsis, and escapes back to infinity. This requires $\varepsilon\geq1$.
 \item[Horizon-connected branch,] $2M<r\leq r_{\rm a}$. There is no periapsis: the allowed region extends down to the horizon. A particle on this branch, moving outward, rises to the apoapsis $r_{\rm a}$, turns back, and falls through the horizon. When there is no turning point at all, the branch extends from the horizon to infinity, and the particle either plunges in from infinity or, when followed backward in time, escapes to it.
\end{description}
Falling through the horizon takes a finite proper time, but an infinite coordinate time. Near $r=2M$, Eq.~\eqref{SchwRadialFirst} gives $\dot r^2\to\varepsilon^2>0$, so the proper time $\int\dd r/|\dot r|$ needed to reach the horizon is finite. By contrast, for an infalling particle Eq.~\eqref{tschweq} gives
\begin{equation}
 \frac{\dd t}{\dd r}=\frac{\dot t}{\dot r}
 \simeq-\frac{1}{f(r)}\,,
\end{equation}
whose integral diverges logarithmically as $r\to2M$. A distant observer never sees the particle cross the horizon, while the particle itself crosses it at a finite value of its proper time. The divergence of $t$ reflects the failure of Schwarzschild coordinates at the horizon, not a physical singularity.

\begin{figure}[htbp]
\centering
\includegraphics{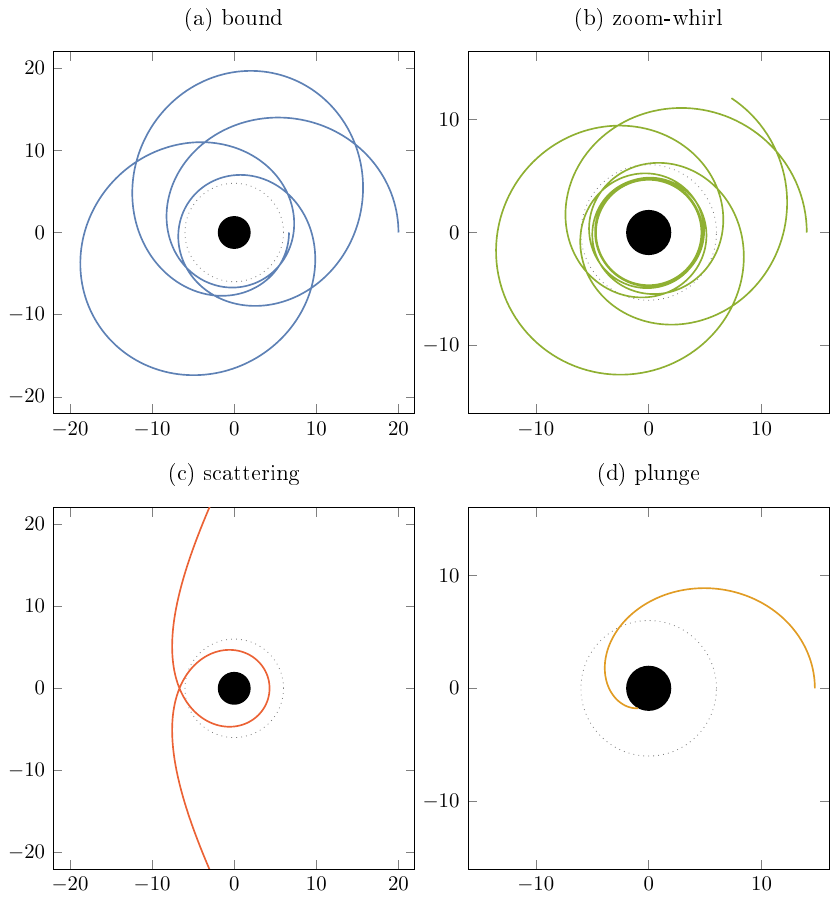}
 \caption{Timelike geodesics in the orbital plane, with coordinates $x=r\cos\varphi$, $y=r\sin\varphi$ in units of $M$. The black disk is the black hole, $r<2M$, and the dotted circle is the ISCO. (a)~Bound orbit with $r_{\rm p}=20M/3$ and $r_{\rm a}=20M$, the orbit of Example~\ref{Ex:SchwBound}; the periapsis advances by about $216^\circ$ per radial period. (b)~Zoom-whirl orbit with eccentricity $1/2$ just outside the separatrix: the particle whirls several times near $r_{\rm p}\simeq4.67M$ before zooming out to $r_{\rm a}=14M$. (c)~Scattering orbit with $\varepsilon=1.02$ and $r_{\rm p}=4.3M$, close to the capture threshold: the particle winds once around the black hole before escaping. (d)~Plunge from the apoapsis $r_{\rm a}\simeq14.7M$, where the radial velocity vanishes, on the horizon-connected branch with $L=3M$ and $\varepsilon^2=0.9$: since $L^2<12M^2$ the potential has no barrier, and the particle spirals into the black hole.}
 \label{SchwOrbits}
\end{figure}

\section{Timelike geodesics from Vieta's formulae}\label{Sec:SchwVieta}

We now derive the same classification by a different route, working with the roots of the radial polynomial instead of the shape of the potential. In Schwarzschild spacetime, the two methods are equivalent, and comparing them is the best way to see what Vieta's formulae achieve. In Kerr spacetime, by contrast, the effective-potential method largely loses its power. The radial function $R(r)$ depends on the energy in a way that cannot be separated from the rest, as $\varepsilon^2/2-V(r)$ can in Schwarzschild spacetime, and it also depends on Carter's constant, which ties the radial motion to the polar one. Apart from motion in the equatorial plane, there is no single one-dimensional potential whose graph can be compared with a horizontal energy line. The algebraic approach needs no such separation: it only requires the radial and polar functions to be polynomials, which they are. It is in Kerr spacetime, therefore, that Vieta's formulae will be used to their full power.

Multiplying Eq.~\eqref{SchwRadialFirst} by $r^3$ gives the radial polynomial
\begin{equation}\label{SchwTimelikePolynomial}
 P(r):=r^3\dot r^2
 =(\varepsilon^2-1)r^3+2Mr^2-L^2r+2ML^2\,.
\end{equation}
Since $r^3>0$, $P$ has the same sign as $\dot r^2$ and the same positive roots, with the same multiplicities, so the results of Section~\ref{Sec:TurningPoints} apply to it: its simple positive roots are radial turning points, and its double roots are circular orbits. Moreover, every positive root of $P$ lies outside the horizon. For $0<r<2M$ we have $f<0$, hence $V_1<0$ and $P=r^3(\varepsilon^2-2V_1)>0$; in particular, $P(2M)=8M^3\varepsilon^2>0$. In Schwarzschild spacetime, positive roots and exterior roots are therefore the same thing.

For $\varepsilon^2\neq1$, let $r_1,r_2,r_3$ denote the three roots, including complex and negative ones. Vieta's formulae give
\begin{align}
 r_1+r_2+r_3&=-\frac{2M}{\varepsilon^2-1}\,,\label{SchwVieta1}\\
 r_1r_2+r_1r_3+r_2r_3&=-\frac{L^2}{\varepsilon^2-1}\,,\label{SchwVieta2}\\
 r_1r_2r_3&=-\frac{2ML^2}{\varepsilon^2-1}\,.\label{SchwVieta3}
\end{align}
When the roots are real, we label them in decreasing order, $r_1\geq r_2\geq r_3$, so that $r_1$ is the outermost turning point.

\subsection{Unbound motion: \texorpdfstring{$\varepsilon^2>1$}{energy squared > 1}}

Here the product of the roots, Eq.~\eqref{SchwVieta3}, is negative, so the number of negative roots is odd (Section~\ref{Sec:RootPatterns}). It cannot be three: the pairwise products of three negative numbers are positive, whereas Eq.~\eqref{SchwVieta2} requires their sum to be negative. Hence exactly one root, $r_0<0$, is negative. \citet{Darwin:1959} reached the same conclusion from Descartes' rule of signs. The other two have positive product $r_1r_2r_3/r_0$, so they are either a complex-conjugate pair or two positive roots. The admissible generic patterns are therefore:
\begin{enumerate}
 \item one negative root and a complex-conjugate pair, giving no turning point: the particle is captured, or, followed backward in time, escapes;
 \item one negative root and two positive roots $r_{\rm a}<r_{\rm p}$, giving a horizon-connected branch $2M<r\leq r_{\rm a}$ and a scattering branch $r\geq r_{\rm p}$.
\end{enumerate}
At the capture threshold the two positive roots coalesce at a radius $r_c$. Using Vieta's formulae with the roots $r_0,r_c,r_c$ gives
\begin{equation}
 r_0=\frac{2Mr_c}{r_c-4M}\,,
 \qquad
 L^2=\frac{Mr_c^2}{r_c-3M}\,,
 \qquad
 \varepsilon^2=\frac{(r_c-2M)^2}{r_c(r_c-3M)}\,,
\end{equation}
the last two coinciding, as they must, with the circular-orbit constants~\eqref{SchwCircularConstants}. Because $r_0<0$, the critical radius obeys $3M<r_c<4M$. For fixed $\varepsilon>1$, a smaller angular momentum than the critical one gives capture, and a larger one gives scattering, together with the disconnected horizon-connected branch.

\subsection{Marginally bound motion: \texorpdfstring{$\varepsilon^2=1$}{energy squared = 1}}

Equation~\eqref{SchwTimelikePolynomial} becomes
\begin{equation}
 P(r)=r\left(2Mr^2-L^2r+2ML^2\right)\,.
\end{equation}
The discriminant of the quadratic factor is $L^2(L^2-16M^2)$. For $L<4M$ there is no exterior turning point, and a particle released from rest at infinity is captured. At $L=4M$ the quadratic has a double root at $r=4M$, the marginally bound circular orbit. For $L>4M$ it has two positive roots, producing a horizon-connected branch and a parabolic scattering branch.

\subsection{Bound-energy motion: \texorpdfstring{$\varepsilon^2<1$}{energy squared < 1}}

Equations~\eqref{SchwVieta1}--\eqref{SchwVieta3} now have positive right-hand sides. By the parity argument, the number of negative roots is even. Two negative roots and one positive root are impossible: if $r_1=-x$, $r_2=-y$, and $r_3=z$ with $x,y,z>0$, the first two Vieta relations would require simultaneously
\begin{equation}
 z>x+y\,,\qquad
 z<\frac{xy}{x+y}\,,
\end{equation}
which cannot hold, since $xy/(x+y)<x+y$. Hence there are only two generic possibilities:
\begin{enumerate}
 \item one positive root and a complex-conjugate pair, giving only a horizon-connected branch;
  \item three positive roots $r_1>r_2>r_3$, for which
 \begin{equation}
 P(r)=-(1-\varepsilon^2)(r-r_1)(r-r_2)(r-r_3)\,.
 \end{equation}
 The allowed exterior intervals are the horizon-connected branch $2M<r\leq r_3$ and the bound branch $r_{\rm p}=r_2\leq r\leq r_1=r_{\rm a}$.
\end{enumerate}
In the language of the potential, the second case occurs when $2V_{\min}<\varepsilon^2<2V_{\max}$: the three roots are the three intersections of the energy line with the graph in Fig.~\ref{SchwVschematic}.

A double root is a circular orbit. If $r_1=r_2=r_c>r_3$, the bound branch shrinks to a point and the circular orbit is stable; if $r_1>r_2=r_3=r_c$, the barrier between the two branches shrinks to a point and the circular orbit is unstable. When all three roots coincide, Vieta's formulae give
\begin{equation}
 3r_c=\frac{2M}{1-\varepsilon^2}\,,\qquad
 3r_c^2=\frac{L^2}{1-\varepsilon^2}\,,\qquad
 r_c^3=\frac{2ML^2}{1-\varepsilon^2}\,,
\end{equation}
whose unique solution is
\begin{equation}
 r_c=6M\,,\qquad L^2=12M^2\,,\qquad\varepsilon^2=\frac89\,.
\end{equation}
The triple root is the ISCO. Which pair of roots coincides, $r_1=r_2>r_3$ (stable) or $r_1>r_2=r_3$ (unstable), is the algebraic signature of the stability of circular motion, and the triple root is the transition between the two.

\subsection{The discriminant}\label{Sec:SchwDiscriminant}

Which of the generic patterns occurs is decided by the discriminant of the cubic. Applying Eq.~\eqref{CubicDisc} to Eq.~\eqref{SchwTimelikePolynomial} gives
\begin{equation}\label{SchwDisc}
 \Delta_3=4L^2\left[(\varepsilon^2-1)L^4
 -\left(27\varepsilon^4-36\varepsilon^2+8\right)M^2L^2-16M^4\right]\,.
\end{equation}
The zero set of the bracket is the curve of circular orbits in the $(L^2,\varepsilon^2)$ plane: substituting the parametrization~\eqref{SchwCircularConstants} makes it vanish identically. Viewed as a quadratic in $L^2$, the bracket has a double root precisely at $\varepsilon^2=8/9$, where it equals $-\tfrac19(L^2-12M^2)^2$. This is the cusp at which the stable and unstable branches of circular orbits meet, that is, the ISCO. At $\varepsilon^2=1$ the bracket reduces to $M^2(L^2-16M^2)$, giving the marginally bound orbit $L=4M$.

The classification of this section can then be restated compactly. Since every positive root lies outside the horizon, the discriminant alone decides the motion. For $\varepsilon^2<1$, $\Delta_3>0$ means three positive roots and hence a bound branch together with the horizon-connected one, whereas $\Delta_3<0$ means a single horizon-connected branch. For $\varepsilon^2>1$, $\Delta_3>0$ means scattering and $\Delta_3<0$ means capture.

\subsection{The separatrix}\label{Sec:SchwSeparatrix}

Dividing Eq.~\eqref{SchwVieta2} by Eq.~\eqref{SchwVieta3} gives
\begin{equation}\label{SchwReciprocal}
 \frac1{r_1}+\frac1{r_2}+\frac1{r_3}=\frac1{2M}
 \qquad(L\neq0)\,.
\end{equation}
The energy has dropped out, and the null relations of Section~\ref{Sec:SchwNull} give the same result. Every non-radial Schwarzschild geodesic, timelike or null, bound or unbound, therefore has roots whose reciprocals add up to the inverse Schwarzschild radius. In the variable $u=1/r$, this is the sum of the roots of the cubic for $u$, which appears in \citet{Hagihara:1931} and was used by \citet{Darwin:1961} to write the third root of a bound orbit, as we do below.

For a bound orbit, write the periapsis and apoapsis as
\begin{equation}
 r_{\rm p}=\frac{pM}{1+e}\,,\qquad r_{\rm a}=\frac{pM}{1-e}\,,\qquad 0\leq e<1\,,
\end{equation}
where $p$ and $e$ are the relativistic semi-latus rectum and eccentricity. Then $1/r_{\rm p}+1/r_{\rm a}=2/(pM)$, and Eq.~\eqref{SchwReciprocal} yields the third root without solving any equation, $r_3=2pM/(p-4)$. The other two Vieta relations then fix the constants of motion:
\begin{equation}\label{SchwSeparatrixConstants}
 L^2=\frac{p^2M^2}{p-3-e^2}\,,
 \qquad
 1-\varepsilon^2=\frac{(1-e^2)(p-4)}{p\,(p-3-e^2)}\,.
\end{equation}
The bound branch is separated from the horizon-connected branch as long as $r_3<r_{\rm p}$, that is,
\begin{equation}\label{SchwSeparatrix}
 p>6+2e\,.
\end{equation}
\citet{Darwin:1961} obtained this condition, in the form $l>2m(3+e)$ for the latus rectum $l$, from the same three roots. The curve $p=6+2e$ is now called the \emph{separatrix} of Schwarzschild bound motion \citep{CutlerKennefickPoisson:1994}. On it, $r_3$ and $r_{\rm p}$ merge into an unstable circular orbit, and the bound branch becomes the orbit of Section~\ref{Sec:SchwCircular} that approaches this circular orbit asymptotically. For $e=0$ one recovers the ISCO at $r=6M$; as $e\to1$, $p\to8$ and $r_{\rm p}\to4M$, the marginally bound circular orbit.

\section{Null geodesics}\label{Sec:SchwNull}

For $\kappa=0$ the effective potential is
\begin{equation}\label{Schwpotentialnull}
 V_0(r)=\frac{L^2}{2r^2}\left(1-\frac{2M}{r}\right)\,.
\end{equation}
It vanishes at the horizon and at infinity and has a unique exterior maximum at
\begin{equation}
 r_{\rm ph}=3M\,,\qquad
 V_0(r_{\rm ph})=\frac{L^2}{54M^2}\,.
\end{equation}
This is the Schwarzschild \emph{photon sphere}, on which photons can move on circular orbits. Since it sits at a maximum of the potential, the circular null orbit is unstable. As shown in Fig.~\ref{SchwVnull}, the maximum lies at $r=3M$ for every value of $L$, since $L$ only rescales the potential.

A photon coming from infinity is reflected by the barrier if its energy level lies below the maximum, $\varepsilon^2/2<L^2/(54M^2)$, that is, if $|b|>3\sqrt3\,M$. The corresponding critical value of the impact parameter is
\begin{equation}\label{SchwCriticalImpact}
 b_c:=3\sqrt{3}\,M\,.
\end{equation}
Null geodesics arriving from infinity are captured for $|b|<b_c$, scattered for $|b|>b_c$, and asymptotically approach the photon sphere at the critical value; examples are shown in Fig.~\ref{SchwPhotons}. These strong-deflection trajectories were analyzed in detail by \citet{Darwin:1959}. Photons with $|b|$ slightly larger than $b_c$ are the null counterpart of the zoom-whirl orbits of massive particles. They approach the photon sphere, wind around it, and escape again, as the orange trajectory in Fig.~\ref{SchwPhotons} does. The closer $b$ is to $b_c$, the more windings they perform: in the strong-deflection limit, the deflection angle diverges logarithmically,
\begin{equation}
 \hat\alpha(b)\simeq-\ln\!\left(\frac{b}{b_c}-1\right)+\ln\!\left[216\left(7-4\sqrt3\right)\right]-\pi\,,
\end{equation}
so that each additional winding requires $b-b_c$ to decrease by a factor $e^{2\pi}\simeq535$. \citet{Darwin:1959} obtained this result in the equivalent form $b-b_c\simeq3.48\,M\,e^{-\hat\alpha}$, and \citet{Bozza:2002} extended it to general static, spherically symmetric spacetimes. In gravitational lensing, these photons produce an infinite sequence of \emph{relativistic images} of a background source on both sides of the black hole, exponentially fainter and converging toward the angular radius of the photon sphere seen by the observer. \citet{Darwin:1959} already described these \enquote{ghost} images, each fainter than the previous one, and closer to the limiting radius, by the factor $e^{-2\pi}$. \citet{Ohanian:1987} expressed their deflection angles and intensities in terms of elliptic integrals; the name \emph{relativistic images} is due to \citet{VirbhadraEllis:2000}, and the strong-deflection treatment to \citet{Bozza:2002}. Their positions, magnifications, and time delays in Reissner--Nordstr\"om and Kerr spacetimes have been studied by \citet{GhazalePiattella:2025}.

\begin{figure}[htbp]
\centering
\includegraphics{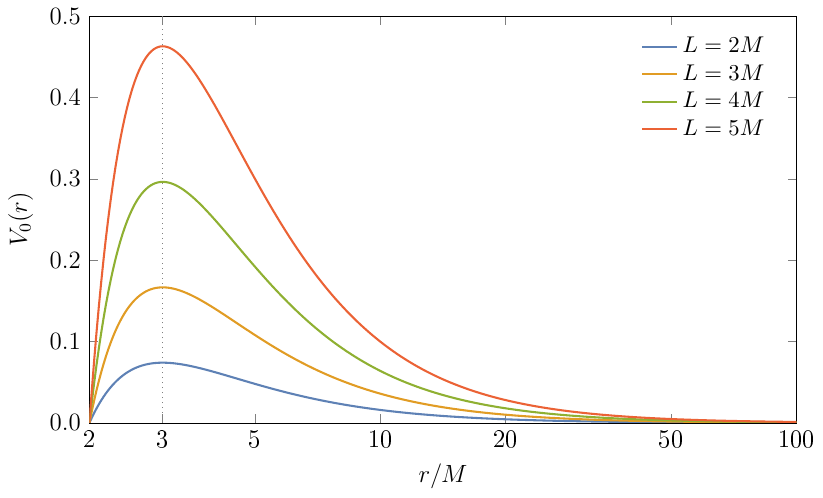}
 \caption{Null effective potential $V_0(r)$, Eq.~\eqref{Schwpotentialnull}, for four values of $L$, in the exterior region $r>2M$. All curves have their maximum at the photon sphere, $r_{\rm ph}=3M$, where $V_0=L^2/(54M^2)$. A photon with energy $\varepsilon$ is reflected by the barrier if $\varepsilon^2/2$ lies below the maximum, that is, if its impact parameter exceeds $b_c=3\sqrt3\,M$.}
 \label{SchwVnull}
\end{figure}

\begin{figure}[htbp]
\centering
\includegraphics{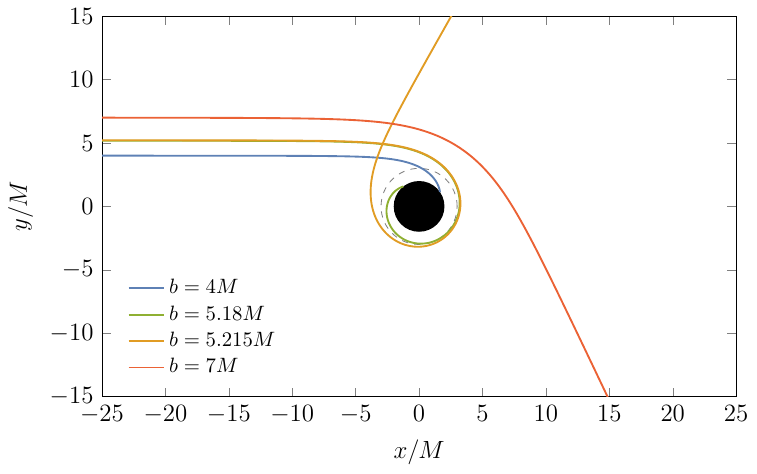}
 \caption{Photons coming from the left with different impact parameters. The black disk is the black hole and the dashed circle the photon sphere, $r=3M$. For $b=4M<b_c$ the photon is captured. For $b$ just below and just above $b_c\simeq5.196M$, the photon winds around the photon sphere before being captured or escaping, respectively. For $b=7M$ it is deflected by about $65^\circ$.}
 \label{SchwPhotons}
\end{figure}

The Vieta analysis makes the same classification especially transparent. Dividing by $\varepsilon^2$, the radial polynomial for null geodesics is
\begin{equation}\label{rschweq2null}
 P_0(r):=\frac{r^3\dot r^2}{\varepsilon^2}
 =r^3-b^2r+2Mb^2\,.
\end{equation}
If its roots are $r_1,r_2,r_3$, then
\begin{align}
 r_1+r_2+r_3&=0\,,\\
 r_1r_2+r_1r_3+r_2r_3&=-b^2\,,\\
 r_1r_2r_3&=-2Mb^2\,.
\end{align}
The product is negative, so the number of negative roots is odd; it cannot be three because the roots add up to zero. Exactly one root is negative, as \citet{Darwin:1959} noted, and in particular, a positive root accompanied by a complex-conjugate pair is impossible. As for timelike geodesics, every positive root lies outside the horizon. The generic alternatives are therefore:
\begin{enumerate}
 \item one negative root and a complex-conjugate pair: no turning point, and the photon is captured;
 \item one negative root and two positive roots: a horizon-connected branch and a scattering branch.
\end{enumerate}
Which alternative occurs is decided by the discriminant of the cubic, Eq.~\eqref{CubicDisc},
\begin{equation}\label{SchwNullDisc}
 \Delta_3=4b^4\left(b^2-27M^2\right)\,,
\end{equation}
which gives complex roots for $|b|<b_c$ and three real roots for $|b|>b_c$. At $|b|=b_c$, the roots are $-6M,3M,3M$, and the positive double root is the unstable circular photon orbit.

For $L=0$, Eq.~\eqref{SchwRadialFirst} reduces to $\dd r/\dd\lambda=\pm\varepsilon$, so radial null geodesics have no turning points in the exterior region.

\section{Worked examples}\label{Sec:SchwExamples}

We close the chapter with two examples. In each, the constants of motion are fixed first, the root pattern is predicted from Vieta's formulae and the discriminant, and the prediction is compared with the numerically computed roots. We set $M=1$, so that all lengths are in units of $M$, and numbers are rounded to the digits shown.

\begin{example}[A Schwarzschild bound orbit]\label{Ex:SchwBound}
Take $p=10$ and $e=1/2$, the orbit of Fig.~\ref{SchwOrbits}a. Then $r_{\rm p}=20/3$ and $r_{\rm a}=20$, and the reciprocal identity~\eqref{SchwReciprocal} gives the third root without solving any equation, $r_3=2p/(p-4)=10/3$. Equations~\eqref{SchwSeparatrixConstants} give $L^2=400/27\simeq14.815$ and $\varepsilon^2=14/15$. Since $p=10>6+2e=7$, the orbit lies on the bound side of the separatrix~\eqref{SchwSeparatrix}. Equation~\eqref{SchwDisc} gives $\Delta_3=(800/243)^2\simeq10.84>0$, so three real roots are expected, and the numerical roots are indeed $r_1=20.000$, $r_2=6.6667$ and $r_3=3.3333$. Since $2<r_3<r_2$, the particle either oscillates in $r_2\leq r\leq r_1$ or, on the disconnected horizon-connected branch, moves in $2<r\leq r_3$.
\end{example}

\begin{example}[A scattered photon]\label{Ex:SchwPhoton}
For $b=6>b_c\simeq5.196$, Eq.~\eqref{SchwNullDisc} gives $\Delta_3=4b^4(b^2-27)=46656>0$. The roots of Eq.~\eqref{rschweq2null} are $4.4534$, $2.3696$ and $-6.8229$: they add up to zero, and their reciprocals add up to $1/2$, as Eq.~\eqref{SchwReciprocal} requires. The photon comes in from infinity, reaches its periapsis at $r=4.4534$, and escapes again.
\end{example}

\chapter{Geodesics in Kerr spacetime}\label{Sec:Kerrgeod}

The Kerr solution \citep{Kerr:1963ud} describes the exterior gravitational field of an uncharged rotating black hole. Unlike the Schwarzschild geometry, it is axially symmetric, so generic geodesics are not planar. Their complete integrability follows from Carter's additional constant of motion \citep{Carter:1968rr,Carter:1968ks}, whose origin in a hidden symmetry, encoded in a rank-two Killing tensor, was identified by \citet{WalkerPenrose:1970}. \citet{Wilkins:1972rs} gave the first general conditions for bound motion and studied the spherical orbits of the extreme black hole; modern analytic formulations may be found in \citet{Schmidt:2002qk}, \citet{DrascoHughes:2004}, and \citet{FujitaHikida:2009}, while a detailed taxonomy of the exterior radial root structures of timelike geodesics was completed by \citet{CompereLiuLong:2022}.

\section{Metric and horizons}

We use the coordinates introduced by \citet{BoyerLindquist:1967}. The line element is
\begin{align}\label{KerrMetric}
 \dd s^2={}&-\left(1-\frac{2Mr}{\rho^2}\right)\dd t^2
 -\frac{4Mar\sin^2\theta}{\rho^2}\,\dd t\,\dd\varphi
 +\frac{\rho^2}{\Delta}\,\dd r^2+\rho^2\,\dd\theta^2\nonumber\\
 &+\left(r^2+a^2+\frac{2Ma^2r\sin^2\theta}{\rho^2}\right)
 \sin^2\theta\,\dd\varphi^2\,,
\end{align}
where
\begin{equation}
 \rho^2:=r^2+a^2\cos^2\theta\,,\qquad
 \Delta:=r^2-2Mr+a^2\,.
\end{equation}
The parameter $M$ is the mass, and $J=aM$ is the angular momentum of the black hole. Reversing the orientation of $\varphi$ reverses the sign of $a$, so we take $a\geq0$. The roots of $\Delta$ are
\begin{equation}
 r_{\rm H}=M+\sqrt{M^2-a^2}\,,\qquad r_{\rm C}=M-\sqrt{M^2-a^2}\,.
\end{equation}
For $a < M$, $r_{\rm H}$ is the event horizon, and $r_{\rm C}$ is the Cauchy horizon. We shall assume $0\leq a<M$ and work in the domain of outer communication, $r>r_{\rm H}$ \citep{Carter:1968rr}.

The Killing vector $\partial_t$ is normalized so that it coincides with the four-velocity of static observers at infinity. An observer at fixed $(r,\theta,\varphi)$ has four-velocity proportional to $\partial_t$, and is therefore possible only where $\partial_t$ is timelike,
\begin{equation}
 g_{tt}=-\left(1-\frac{2Mr}{\rho^2}\right)<0
 \quad\Longleftrightarrow\quad
 r^2-2Mr+a^2\cos^2\theta>0\,.
\end{equation}
Outside the horizon, this condition fails between $r_{\rm H}$ and the \emph{stationary-limit surface}
\begin{equation}
 r_{\rm erg}(\theta)=M+\sqrt{M^2-a^2\cos^2\theta}\,,
\end{equation}
the outer root of $g_{tt}=0$. The \emph{ergoregion} is the region
\begin{equation}\label{Ergoregion}
 r_{\rm H}<r<r_{\rm erg}(\theta)\,,
\end{equation}
in which $\partial_t$ is spacelike. Since $r_{\rm erg}=r_{\rm H}$ on the axis and $r_{\rm erg}=2M$ in the equatorial plane, the ergoregion is a shell that is thickest at the equator and pinches off at the poles (Fig.~\ref{KerrErgoregion}); it disappears for $a=0$. Two features distinguish the stationary-limit surface from the event horizon. First, it is not a one-way membrane: away from the axis, it is a timelike hypersurface, which particles and light can cross in both directions. Second, inside it, no observer can remain at rest with respect to infinity, and the conserved energy $\varepsilon=-p_\mu(\partial_t)^\mu$ of a particle need not be positive, since it is no longer the energy measured by any observer. The inner root of $g_{tt}=0$ lies inside the Cauchy horizon and plays no role here.

% Meridional section of a Kerr black hole (a = 0.9M), coordinates X = sqrt(r^2+a^2) sin(theta), Z = r cos(theta)
\begin{figure}[htbp]
\centering
\includegraphics{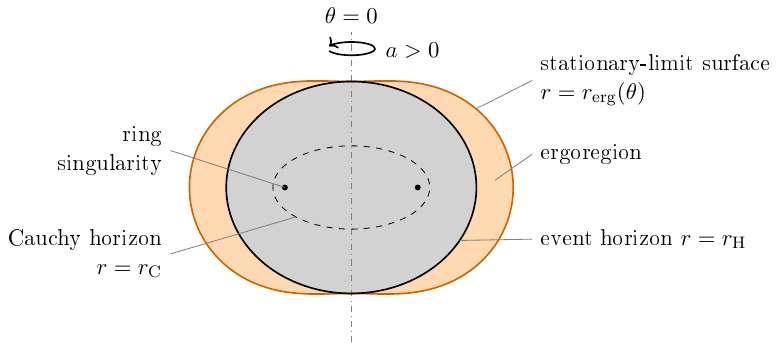}
\caption{Meridional section of a Kerr black hole with $a=0.9M$, drawn in the coordinates $X=\sqrt{r^2+a^2}\,\sin\theta$, $Z=r\cos\theta$, in which surfaces of constant $r$ are oblate spheroids. The black-hole interior $r<r_{\rm H}$ is gray. The ergoregion (orange) lies between the event horizon and the stationary-limit surface, which touch on the rotation axis. The ring singularity at $r=0$, $\theta=\pi/2$ appears as two dots; the inner stationary-limit surface, inside the Cauchy horizon, is not shown.}
\label{KerrErgoregion}
\end{figure}

The rotation of the black hole drags inertial frames. A convenient measure of this effect is the angular velocity, as seen from infinity, of the zero-angular-momentum observers (ZAMOs), whose worldlines are orthogonal to the hypersurfaces $t={\rm const}$; \citet{BardeenPressTeukolsky:1972} call them locally nonrotating observers:
\begin{equation}\label{ZAMOOmega}
 \omega:=-\frac{g_{t\varphi}}{g_{\varphi\varphi}}=\frac{2aMr}{\Sigma^2}\,,
 \qquad
 \Sigma^2:=(r^2+a^2)^2-a^2\Delta\sin^2\theta\,.
\end{equation}
It is positive for $a>0$ and falls off as $2J/r^3$ at large $r$. More generally, an observer moving on a circle with uniform angular velocity $\Omega=\dd\varphi/\dd t$ has four-velocity proportional to $\partial_t+\Omega\,\partial_\varphi$, which is timelike only if $\Omega_{\min}<\Omega<\Omega_{\max}$, where
\begin{equation}\label{OmegaRange}
 \Omega_{\min}=\omega-\frac{\rho^2\sqrt\Delta}{\Sigma^2\sin\theta}\,,
 \qquad
 \Omega_{\max}=\omega+\frac{\rho^2\sqrt\Delta}{\Sigma^2\sin\theta}\,.
\end{equation}
Their product is $g_{tt}/g_{\varphi\varphi}$, so inside the ergoregion $\Omega_{\min}>0$: every observer, and every future-directed causal trajectory, co-rotates with the black hole. As $r\to r_{\rm H}$, $\Delta\to0$ and the interval shrinks to the single value
\begin{equation}
 \Omega_{\rm H}=\omega(r_{\rm H},\theta)=\frac{a}{r_{\rm H}^2+a^2}\,,
\end{equation}
which does not depend on $\theta$. Thus $\Omega_{\rm H}$ is the angular velocity of the event horizon relative to infinity: observers just outside the horizon are forced to rotate with it, and the horizon rotates rigidly. Equivalently, the Killing vector
\begin{equation}
 \chi:=\partial_t+\Omega_{\rm H}\,\partial_\varphi
\end{equation}
is null on $r=r_{\rm H}$ and tangent to the horizon generators, so the event horizon is a Killing horizon of $\chi$. For a geodesic the conserved quantity $-p_\mu\chi^\mu=\varepsilon-\Omega_{\rm H}L_z$ is the one constrained by future orientation at the horizon, Eq.~\eqref{HorizonForwardCondition}. The angular velocity grows monotonically with the spin, from $\Omega_{\rm H}=0$ at $a=0$ toward $1/(2M)$ as $a\to M$.

\section{Separated geodesic equations}

Let $\lambda$ denote proper time for timelike geodesics and an affine parameter for null geodesics, and let $\kappa=1$ or $0$, respectively. The two Killing vectors give
\begin{equation}
 \varepsilon:=-p_t\,,\qquad L_z:=p_\varphi\,.
\end{equation}
For timelike motion these are the specific energy and the axial component of the specific angular momentum. Carter's separation introduces a third constant, $\mathcal Q$ \citep{Carter:1968rr}; see also the detailed study of Kerr geodesics by \citet{Bardeen:1973}. Define
\begin{equation}\label{KerrKDefinition}
 K:=\mathcal Q+(L_z-a\varepsilon)^2\,.
\end{equation}
The first-order geodesic equations are
\begin{align}
 \Delta\rho^2\frac{\dd t}{\dd\lambda}
 &=
 \left[(r^2+a^2)^2-a^2\Delta\sin^2\theta\right]\varepsilon
 -2aMrL_z\,,\label{dtdlambda}\\
 \Delta\rho^2\frac{\dd\varphi}{\dd\lambda}
 &=2aMr\varepsilon+
 \left(\rho^2-2Mr\right)L_z\csc^2\theta\,,\label{dphidlambda}\\
 \rho^4\left(\frac{\dd r}{\dd\lambda}\right)^2&=R(r)\,,
 \label{drdlambda}\\
 \rho^4\left(\frac{\dd\theta}{\dd\lambda}\right)^2&=\Theta(\theta)\,,
 \label{dthetadlambda}
\end{align}
where
\begin{align}
 R(r)
 &:=
 \left[\varepsilon(r^2+a^2)-aL_z\right]^2
 -\Delta\left(\kappa r^2+K\right)\,,\label{KerrRCompact}\\
 \Theta(\theta)
 &:=
 \mathcal Q-\cos^2\theta
 \left[a^2(\kappa-\varepsilon^2)
 +L_z^2\csc^2\theta\right]\,.\label{KerrTheta}
\end{align}
Expanding Eq.~\eqref{KerrRCompact} gives the quartic
\begin{equation}\label{KerrRExpanded}
 R=(\varepsilon^2-\kappa)r^4+2M\kappa r^3
 +\left[a^2(\varepsilon^2-\kappa)-L_z^2-\mathcal Q\right]r^2
 +2MKr-a^2\mathcal Q\,.
\end{equation}
These equations follow from Carter's Hamilton--Jacobi separation \citep{Carter:1968rr}, later extended to a wider class of spacetimes \citep{Carter:1968ks}. Exact solutions can be expressed using elliptic integrals and elliptic functions \citep{FujitaHikida:2009,GrallaLupsasca:2020}.

It is often convenient to introduce Mino time $\gamma$ by
\begin{equation}
 \dd\gamma=\frac{\dd\lambda}{\rho^2}\,.
\end{equation}
Then
\begin{equation}
 \left(\frac{\dd r}{\dd\gamma}\right)^2=R(r)\,,
 \qquad
 \left(\frac{\dd\theta}{\dd\gamma}\right)^2=\Theta(\theta)\,,
\end{equation}
so the radial and polar motions are dynamically decoupled. This parametrization is especially useful for frequencies and generic bound orbits \citep{Mino:2003,DrascoHughes:2004}.

The alternative Carter constant $K$ is manifestly non-negative for causal geodesics:
\begin{equation}\label{KNonnegative}
 K=p_\theta^2+
 \left(\frac{L_z}{\sin\theta}-a\varepsilon\sin\theta\right)^2
 +a^2\kappa\cos^2\theta\geq0\,.
\end{equation}
By contrast, $\mathcal Q$ itself can be negative for null geodesics and for unbound timelike geodesics.

\subsection{Time orientation and frame dragging}

For a future-directed geodesic outside the horizon, Eq.~\eqref{dtdlambda} must be positive along the allowed orbit. Near the horizon, this implies the invariant condition
\begin{equation}\label{HorizonForwardCondition}
 \varepsilon-\Omega_{\rm H} L_z\geq0\,.
\end{equation}
Away from the horizon, the same argument applied to the ZAMOs, whose four-velocity is proportional to $\partial_t+\omega\,\partial_\varphi$, gives the stronger condition
\begin{equation}\label{ZAMOForwardCondition}
 \varepsilon-\omega L_z>0\,,
\end{equation}
of which Eq.~\eqref{HorizonForwardCondition} is the limit $r\to r_{\rm H}$. Outside the ergoregion $\partial_t$ is itself future-directed and timelike, so $\varepsilon>0$. Negative-energy geodesics can therefore occur only inside the ergoregion and, since $\omega>0$ for $a>0$, Eq.~\eqref{ZAMOForwardCondition} requires $L_z<\varepsilon/\omega<0$.

This is the kinematic basis of the Penrose process \citep{Penrose:1969,PenroseFloyd:1971}. A particle that enters the ergoregion and splits into two fragments can place one of them on a negative-energy geodesic, which must eventually cross the horizon; by conservation of $\varepsilon$, the other fragment escapes to infinity with more energy than the incoming particle. The surplus is drawn from the rotation of the black hole. When the negative-energy fragment is captured, the mass and angular momentum of the hole change by $\delta M=\varepsilon$ and $\delta J=L_z$ of that fragment, and Eq.~\eqref{HorizonForwardCondition} gives
\begin{equation}
 \delta M\geq\Omega_{\rm H}\,\delta J\,,
\end{equation}
so energy can be extracted only together with angular momentum. This inequality was derived by \citet{PenroseFloyd:1971}, who noted that it expresses the increase of the area of the horizon. For a single decay, the gain is appreciable only if the fragments separate at relativistic relative speeds \citep{Wald:1974}; for matter coming from stable orbits, \citet{BardeenPressTeukolsky:1972} estimated that more than about half the speed of light is needed.

Collisions offer a second route to high energies. For two particles with four-momenta $p_1$ and $p_2$ meeting at the same event, the center-of-mass energy is
\begin{equation}
 E_{\rm cm}^2:=-(p_1+p_2)_\mu(p_1+p_2)^\mu\,.
\end{equation}
\citet{BanadosSilkWest:2009} found that for an extremal black hole $E_{\rm cm}$ grows without bound as the collision point approaches the horizon, provided one of the particles has the critical angular momentum $L_z=2M\varepsilon$. In the present notation, the critical condition is $\varepsilon-\Omega_{\rm H}L_z=0$, the equality case of Eq.~\eqref{HorizonForwardCondition}, which for $a=M$, where $\Omega_{\rm H}=1/(2M)$, reduces to the value above. For generic particles, $E_{\rm cm}$ remains finite at the horizon; it diverges only when one of the colliding particles approaches criticality. This \emph{BSW effect} is subject to strong kinematic restrictions: in the extremal case, the critical particle needs an infinite proper time to reach the horizon, while for $a<M$, critical particles arriving from infinity have a radial turning point outside the horizon, so that $E_{\rm cm}$ stays bounded \citep{BertiEtAl:2009,JacobsonSotiriou:2010}. \citet{GribPavlov:2011} showed that arbitrarily large center-of-mass energies are nevertheless possible for non-extremal black holes when the near-critical angular momentum is acquired in a previous collision close to the horizon \citep[see also][]{GribPavlovPiattella:2011,GribPavlovPiattella:2012}. They also found that $E_{\rm cm}$ can grow without bound anywhere in the ergoregion, if one of the particles has a large negative angular momentum \citep{GribPavlov:2013}, and studied collisions involving particles of zero or negative energy \citep{GribPavlov:2017}. Whether a large $E_{\rm cm}$ translates into energetic particles at infinity is a separate question: the escaping products of a single collision carry only a modest multiple of the incoming energy \citep{BejgerEtAl:2012,LeiderschneiderPiran:2016}. A review is given by \citet{HaradaKimura:2014}.

The sign of $L_z$ alone does not determine the sign of $\dd\varphi/\dd\lambda$ everywhere. A particle with $L_z<0$ can reverse its Boyer--Lindquist azimuthal motion; inside the ergoregion, every future-directed causal trajectory is nevertheless forced to co-rotate in the sense $\dd\varphi/\dd t>0$ for $a>0$. Zeros of Eq.~\eqref{dphidlambda} are therefore turning points of the Boyer--Lindquist azimuth only: they carry no information beyond the constants of motion and play no role in the classification.

\section{Polar motion}

Let us set
\begin{equation}
 z:=\cos^2\theta\,,\qquad 0\leq z\leq1\,,
 \qquad A:=a^2(\varepsilon^2-\kappa)\,.
\end{equation}
For $a\neq0$, the three signs of $A$ correspond, for timelike geodesics, to $\varepsilon^2<1$, $\varepsilon^2=1$, and $\varepsilon^2>1$; for null geodesics, $A=a^2\varepsilon^2>0$ always. Multiplying Eq.~\eqref{KerrTheta} by $1-z=\sin^2\theta$ gives
\begin{equation}\label{PolarPolynomial}
 F(z):=(1-z)\Theta=-Az^2+Bz+\mathcal Q\,,
\end{equation}
with
\begin{equation}\label{VorticalB}
 B:=A-L_z^2-\mathcal Q\,.
\end{equation}
The allowed polar region is $F(z)\geq0$. Its endpoint values are
\begin{equation}\label{PolarEndpoints}
 F(0)=\mathcal Q\,,\qquad F(1)=-L_z^2\,.
\end{equation}
When $L_z\neq0$, the poles $\theta=0,\pi$ ($z=1$) cannot be reached. When $L_z=0$, instead, $F(z)=(1-z)(\mathcal Q+Az)$, and the root $z=1$ is an artifact of the multiplication by $1-z$: the polar motion is then governed directly by $\Theta=\mathcal Q+Az$, and the poles are crossed if and only if $\mathcal Q+A>0$; for $\mathcal Q+A=0$ the axis is a double root of $\Theta$.

Let $A \neq 0$ and $z_{1,2}$ be the roots of $F(z)$. From Vieta's formulae we have: 
\begin{equation}\label{PolarVieta}
 z_1+z_2=\frac{B}{A}\,,
 \qquad
 z_1z_2=-\frac{\mathcal Q}{A}\,,
 \qquad
 (1-z_1)(1-z_2)=\frac{L_z^2}{A}\,,
\end{equation}
where the last relation is a combination of the first two, with the definition of $B$ in Eq.~\eqref{VorticalB} employed. The roots are real if and only if the discriminant of $F$ (Section~\ref{Sec:Discriminant}) is non-negative,
\begin{equation}\label{PolarDiscriminant}
 D_\theta := B^2 + 4A\mathcal Q \geq 0\,.
\end{equation}
Moreover, since only $z\in[0,1]$ is physical, the signs of the second and third relations tell us where real roots lie: $z_1z_2>0$ means that they are on the same side of $z=0$, and $(1-z_1)(1-z_2)>0$ indicates that they are on the same side of $z=1$.

If $D_\theta<0$, then $z_1$ and $z_2$ are complex conjugate roots, $z_2=\bar z_1$, so that $z_1z_2=|z_1|^2>0$ and $(1-z_1)(1-z_2)=|1-z_1|^2>0$. Equation~\eqref{PolarVieta} then demands
\begin{equation}
    A > 0\,, \qquad \mathcal{Q} < 0\,, \qquad L_z \neq 0\,.
\end{equation}
Then $F(z)$ has the sign of $-A<0$ for every real $z$, and no polar motion is possible. The same conclusion holds whenever neither root lies in $[0,1]$: $F$ then has a constant sign on $[0,1]$, namely that of $F(1)=-L_z^2<0$.

This leaves us with two kinds of non-equatorial motion, distinguished by the sign of $\mathcal Q$. If $\mathcal Q>0$, then $F(0)>0>F(1)$, so exactly one root lies in $(0,1)$; we call it $z_{\max}$. By the second and third relations of Eq.~\eqref{PolarVieta}, the other root is negative for $A>0$ and greater than unity for $A<0$. The orbit crosses the equatorial plane and oscillates between
\begin{equation}\label{PolarCrossing}
 \theta_{\min}=\arccos\sqrt{z_{\max}}\,,
 \qquad
 \theta_{\max}=\pi-\theta_{\min}\,.
\end{equation}
The particle oscillates symmetrically about the equatorial plane, between two cones that are mirror-symmetric with respect to it, and crosses the plane twice in each polar period. It is the analog of an inclined orbit in Newtonian gravity. For $a=0$ one finds $z_{\max}=\mathcal Q/(\mathcal Q+L_z^2)=\sin^2 i$, with $i$ the inclination of the orbital plane, so that $\mathcal Q=0$ is the equatorial limit and the band widens as $\mathcal Q$ grows, approaching the poles only as $\mathcal Q/L_z^2\to\infty$. Unlike the Newtonian case, for $a\neq0$ the orbit does not lie in a fixed plane: the dragging of inertial frames makes it precess about the spin axis and, for generic constants of motion, gradually fill the whole band (Fig.~\ref{Fig:PolarCrossing}).
% Generated by checks/polar_figures.py -- do not edit coordinates by hand
\begin{figure}[t]
\centering
\colorlet{bandcol}{blue!12}\colorlet{edgecol}{blue!70!black}\colorlet{conecol}{gray!70}\colorlet{orbcol}{red!75!black}
\raisebox{-141.45503pt}{\includegraphics{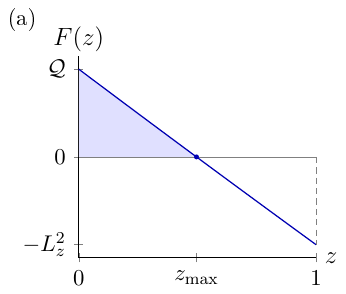}}%
\hfill
\raisebox{-167.63028pt}{\includegraphics{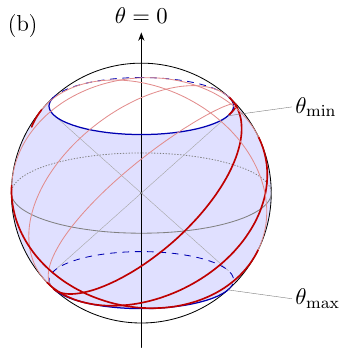}}
\caption{Polar motion of crossing type. (a)~The polynomial $F(z)$ of Eq.~\eqref{PolarPolynomial} for a bound timelike geodesic with $a=0.9M$, $\varepsilon=0.92$, $L_z=2M$ and $\mathcal Q=4M^2$ ($4.54M\leq r\leq6.32M$). Since $F(0)=\mathcal Q>0$ and $F(1)=-L_z^2<0$, exactly one root, $z_{\max}$, lies in $(0,1)$, and the motion is allowed for $0\leq z\leq z_{\max}$ (shaded). (b)~The same geodesic over three polar periods, projected radially onto a sphere. It oscillates between the cones $\theta_{\min}$ and $\theta_{\max}=\pi-\theta_{\min}$, crossing the equatorial plane twice per period, while the dragging of inertial frames makes it precess about the spin axis. Parts on the far side of the sphere are faded.}
\label{Fig:PolarCrossing}
\end{figure}

If instead $\mathcal Q<0$, then $F(0)<0$ and $F(1)<0$, so the roots in $[0,1]$, if any, come in pairs. By Eq.~\eqref{PolarVieta}, both lie in $(0,1)$ if and only if
\begin{equation}\label{VorticalConditions}
 A>0\,,\qquad 0<B<2A\,,\qquad D_\theta\geq0\,,
\end{equation}
since the first condition makes $z_1z_2$ positive, the second places $z_1+z_2$ in $(0,2)$, and the third makes the roots real. Calling them $z_{\min}\leq z_{\max}$, the motion is allowed for $z_{\min}\leq z\leq z_{\max}$, that is, in two bands symmetric with respect to the equatorial plane. In this \emph{vortical} motion, so named by \citet{DeFelice:1972ad}, who studied it for the extreme black hole, $a=M$, the particle never reaches the equatorial plane. It stays in one hemisphere, confined between two cones about the spin axis; the mirror orbit in the other hemisphere is equally allowed. In Mino time, the polar motion oscillates between the two cones, but a vortical geodesic crosses the exterior region too quickly to complete an oscillation: numerically, its Mino time outside the horizon is always less than half the time needed to go from one cone to the other, so it has at most one polar turning point before reaching the horizon (Fig.~\ref{Fig:PolarVortical}).% Generated by checks/polar_figures.py -- do not edit coordinates by hand
\begin{figure}[t]
\centering
\colorlet{bandcol}{blue!12}\colorlet{edgecol}{blue!70!black}\colorlet{conecol}{gray!70}\colorlet{orbcol}{red!75!black}
\raisebox{-141.46175pt}{\includegraphics{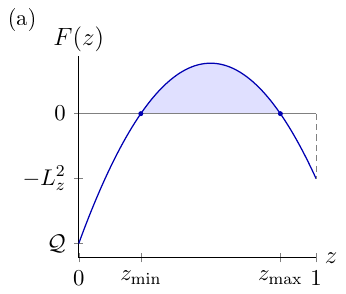}}%
\hfill
\raisebox{-151.18451pt}{\includegraphics{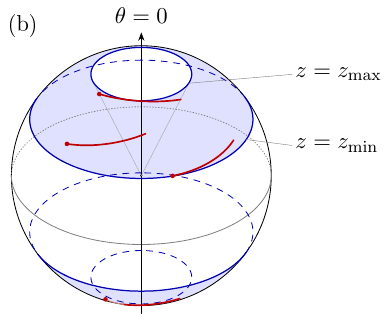}}
\caption{Polar motion of vortical type. (a)~The polynomial $F(z)$ of Eq.~\eqref{PolarPolynomial} for counter-rotating photons with $a=0.9M$, $L_z/\varepsilon=-0.3M$ and $\mathcal Q/\varepsilon^2=-0.18M^2$. Both endpoint values are negative, and the two roots $z_{\min}$ and $z_{\max}$ lie in $(0,1)$: the motion is allowed for $z_{\min}\leq z\leq z_{\max}$ (shaded). (b)~Four such photons arriving from large distance, projected radially onto a sphere, until they cross the event horizon (dots). Each stays in the band of one hemisphere, between the cones $z=z_{\max}$ and $z=z_{\min}$ or their mirror images, and never reaches the equatorial plane. The azimuth is that of ingoing Kerr coordinates, which is regular at the horizon. Parts on the far side of the sphere are faded.}
\label{Fig:PolarVortical}
\end{figure}

The marginal case $\mathcal Q=0$, which separates the two kinds of motion, is discussed in the classification below. Vortical motion has no Newtonian or Schwarzschild counterpart. It requires $A>0$, which is impossible for $a=0$, where $A=0$, and for timelike geodesics with $\varepsilon^2\leq1$; hence, it exists only for unbound timelike geodesics and for photons \citep[cf.][]{DeFelice:1972ad,Wilkins:1972rs}. Conversely, crossing motion is the only non-equatorial polar motion available to bound timelike geodesics.

\subsection{The effective potential for the polar motion}

The classification has a transparent effective-potential reading. Writing
\begin{equation}\label{PolarPotential}
 \Theta=\mathcal Q-U(z)\,,\qquad U(z):=\frac{L_z^2\,z}{1-z}-A\,z\,,
\end{equation}
the polar motion is a one-dimensional motion in $\theta$, since $\Theta=p_\theta^2$: a particle with \enquote{energy} $\mathcal Q$ in the potential $U(\cos^2\theta)$, as in Section~\ref{Sec:EffPot}. One has $U(0)=0$ and, for $L_z\neq0$, $U\to+\infty$ at the poles. If $L_z^2\geq A$, which includes every $A\leq0$, $U$ increases monotonically and $\mathcal Q\geq0$ is required. If $L_z^2<A$, the equatorial plane is a local maximum of $U$, and $U$ has a minimum at
\begin{equation}\label{PolarPotentialMin}
 z_* = 1-\frac{|L_z|}{\sqrt A}\,,\qquad U(z_*)=-\bigl(\sqrt A-|L_z|\bigr)^2\,.
\end{equation}
Orbits with $\mathcal Q>0$ pass over the maximum and cross the equatorial plane; $\mathcal Q=0$ gives the unstable equatorial motion and the separatrix; $U(z_*)<\mathcal Q<0$ gives vortical motion trapped in one of the two mirror-symmetric wells; $\mathcal Q=U(z_*)$ gives the stable cones $z=z_*$; and $\mathcal Q<U(z_*)$ is forbidden. For $L_z=0$, the minimum lies at the poles, $z_*=1$, and the two wells become polar caps. These ranges coincide with those obtained from Vieta's formulae in Eq.~\eqref{VorticalConditions}. Real roots outside $[0,1]$ and complex roots all describe the same situation, $\mathcal Q$ below the minimum of $U$: no geodesic has these constants of motion. See Fig.~\ref{Fig:PolarPotential}.

% Polar effective potential U(cos^2 theta), Eq. (PolarPotential); plotted against cos(theta)
\begin{figure}[htbp]
 \centering
 \includegraphics{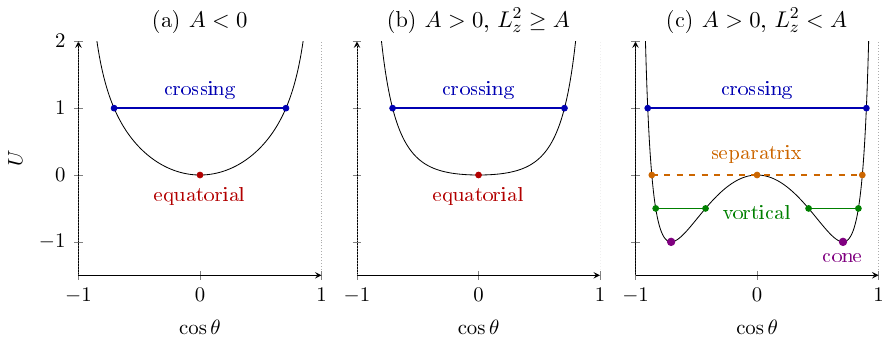}
 \caption{Polar potential $U$ of Eq.~\eqref{PolarPotential} as a function of $\cos\theta$, which shows both hemispheres; the equatorial plane is $\cos\theta=0$ and the poles are $\cos\theta=\pm1$ (dotted). Motion requires $\mathcal Q\geq U$; horizontal segments show the allowed range for representative values of $\mathcal Q$, and dots mark turning points or equilibria. Values are in arbitrary units. (a)~$A=-1$, $L_z^2=1/2$; (b)~$A=1$, $L_z^2=3/2$. In both, $L_z^2\geq A$ and $U$ has its only minimum at the equator: $\mathcal Q>0$ gives equator-crossing motion, $\mathcal Q=0$ stable equatorial motion, and $\mathcal Q<0$ is forbidden. (c)~$A=4$, $L_z^2=1$: the equator is a local maximum of $U$, flanked by two mirror-symmetric wells with minima $U=-(\sqrt A-|L_z|)^2=-1$ at $\cos^2\theta=1-|L_z|/\sqrt A=1/2$. From top: $\mathcal Q=1$, equator-crossing; $\mathcal Q=0$, unstable equatorial motion and the separatrix (dashed), which approaches the equatorial plane asymptotically; $\mathcal Q=-1/2$, vortical motion in either well; $\mathcal Q=-1$, the cones. Values $\mathcal Q<-1$ are forbidden.}
 \label{Fig:PolarPotential}
\end{figure}

\subsection{Classification of the polar motion based on Vieta's formulae}

The complete classification is as follows:
\begin{enumerate}
  \item \textbf{$A<0$} ($a\neq0$, timelike with $\varepsilon^2<1$). $F$ is convex and, for $L_z\neq0$, the third of Eqs.~\eqref{PolarVieta} gives $(1-z_1)(1-z_2)<0$: the roots are real and lie on opposite sides of $z=1$. Taking $z_1<1<z_2$, the allowed part of $[0,1]$ is $0\leq z\leq z_1$. The sign of $z_1$ follows from $z_1z_2=\mathcal Q/|A|$:
 \begin{itemize}
  \item $\mathcal Q>0$: $z_1\in(0,1)$ and the orbit crosses the equatorial plane, with $z_{\max}=z_1$;
  \item $\mathcal Q=0$: $z_1=0$ and the orbit is equatorial;
  \item $\mathcal Q<0$: $z_1<0$, so $F<0$ on $[0,1]$ and no motion is possible.
 \end{itemize}
 For $L_z=0$, one root is $z=1$ and the other is $\mathcal Q/|A|$; the conclusions are unchanged, except that for $\mathcal Q>0$, one has $z_{\max}=\min(1,\mathcal Q/|A|)$, and the orbit passes over the poles when $\mathcal Q>|A|$; for $\mathcal Q=|A|$, $\Theta=|A|\sin^2\theta$ and the orbit approaches the axis only asymptotically. Hence, bound timelike geodesics have $\mathcal Q\geq0$, with $\mathcal Q=0$ if and only if they are equatorial.
 \item \textbf{$A=0$} ($a=0$, or timelike with $\varepsilon^2=1$). $F(z)=\mathcal Q-(L_z^2+\mathcal Q)z$ is linear:
 \begin{itemize}
  \item $\mathcal Q>0$: the orbit crosses the equatorial plane, with
  \begin{equation}
   z_{\max}=\frac{\mathcal Q}{L_z^2+\mathcal Q}\,;
  \end{equation}
  \item $\mathcal Q=0$: $F(z)=-L_z^2z$, so the orbit is equatorial if $L_z\neq0$; if $L_z=0$, then $\Theta\equiv0$ and every constant $\theta$ is a solution;
  \item $\mathcal Q<0$: $F<0$ on $[0,1)$ and no motion is possible.
 \end{itemize}
 \item  \textbf{$A>0$} ($a\neq0$, timelike with $\varepsilon^2>1$, or null). $F$ is concave, and by the third of Eqs.~\eqref{PolarVieta}, $(1-z_1)(1-z_2)\geq0$. This means that real roots lie on the same side of $z=1$, one of them being $z=1$ itself when $L_z=0$.
 \begin{itemize}
  \item $\mathcal Q>0$: $z_1z_2=-\mathcal Q/A<0$, so the roots are real, with one negative; the other is then at most $1$, with equality only for $L_z=0$, when the orbit passes over the poles. The orbit crosses the equatorial plane, with $z_{\max}$ as the positive root.
  \item $\mathcal Q=0$: $F(z)=z\,(B-Az)$ with $B=A-L_z^2$, and the equatorial plane is an invariant solution. If $L_z^2\geq A$, it is the only one. If $L_z^2<A$, a non-equatorial branch also exists, with
  \begin{equation}
   0<z\leq\frac{B}{A}=1-\frac{L_z^2}{A}\,.
  \end{equation}
 Writing $x=\cos\theta$, one has exactly $(\dd x/\dd\gamma)^2=\sin^2\theta\,\Theta=F(x^2)$, so near the equator $(\dd x/\dd\gamma)^2\simeq Bx^2$ and $|\cos\theta|$ varies exponentially in Mino time:
 the orbit remains in one hemisphere and approaches the equatorial plane only as $\gamma\to\pm\infty$, without ever reaching it. This branch is the common limit of the equator-crossing orbits as $\mathcal Q\to0^+$ and of the vortical orbits below as $\mathcal Q\to0^-$ (for which $z_{\min}\to0$); it is the separatrix of the polar motion. Therefore, $\mathcal Q=0$ is necessary, but not by itself sufficient, for a geodesic to be equatorial.
  \item $\mathcal Q<0$: $z_1z_2>0$ and $(1-z_1)(1-z_2)\geq0$ hold automatically, so real roots lie in $(0,1]$ exactly when both $z_1+z_2$ and $(1-z_1)+(1-z_2)$ are positive. Since $F(0)<0$ and $F(1)\leq0$, motion is possible if and only if the conditions \eqref{VorticalConditions} hold, the boundary value $B=2A$ being allowed only for $L_z=0$, $\mathcal Q=-A$, the motion along the axis discussed below.
  The condition $B<2A$ is not redundant: for $A=1$, $L_z^2=0.1$, $\mathcal Q=-5$ one has $B>0$ and $D_\theta>0$, but both roots exceed $1$. If $D_\theta=0$, the double root $z_*=B/(2A)=1-|L_z|/\sqrt A$ of Eq.~\eqref{PolarPotentialMin} gives motion at constant $\theta$ on the cone $\cos^2\theta=z_*$. If $D_\theta>0$, the motion is vortical,
  \begin{equation}
   0<z_{\min}\leq z\leq z_{\max}\leq1\,,
  \end{equation}
  with $z_{\max}=1$ only for $L_z=0$, in which case the band is a polar cap crossed through the axis. The principal null geodesics, with $L_z=a\varepsilon\sin^2\theta_0$ and $\mathcal Q=-a^2\varepsilon^2\cos^4\theta_0$, are an example of the constant-$\theta$ case $D_\theta=0$, with $z_\ast=\cos^2\theta_0$. The limiting value $z_*=1$ ($L_z=0$, $\mathcal Q=-A$) is motion along the axis.
 \end{itemize}
\end{enumerate}
The results are summarized in Table~\ref{Tab:PolarClassification}.

\begin{table}[htbp]
\centering
\small
\begin{tabular}{@{}lll@{}}
\toprule
Sign of $A$ & Carter constant & Polar motion\\
\midrule
$A<0$ & $\mathcal Q>0$ & crossing, $0\leq z\leq z_{\max}$; second root $>1$\\
 & $\mathcal Q=0$ & equatorial\\
 & $\mathcal Q<0$ & none\\[2pt]
$A=0$ & $\mathcal Q>0$ & crossing, $z_{\max}=\mathcal Q/(L_z^2+\mathcal Q)$\\
 & $\mathcal Q=0$ & equatorial$^{\ast}$\\
 & $\mathcal Q<0$ & none\\[2pt]
$A>0$ & $\mathcal Q>0$ & crossing, $0\leq z\leq z_{\max}$; second root $<0$\\
 & $\mathcal Q=0$ & equatorial; separatrix $0<z\leq1-L_z^2/A$ if $L_z^2<A$\\
 & $\mathcal Q<0$ & vortical, $z_{\min}\leq z\leq z_{\max}$, if $0<B<2A$ and $D_\theta\geq0$\\
 & & none otherwise\\
\bottomrule
\end{tabular}
\caption{Polar motion in Kerr spacetime by the signs of $A=a^2(\varepsilon^2-\kappa)$ and $\mathcal Q$. For timelike geodesics with $a\neq0$, the three groups of rows correspond to $\varepsilon^2<1$, $\varepsilon^2=1$ and $\varepsilon^2>1$; null geodesics always have $A>0$. Crossing orbits oscillate between the angles of Eq.~\eqref{PolarCrossing}; the position of the second root of $F$ is given for $A\neq0$ and $L_z\neq0$. $^\ast$For $A=L_z=\mathcal Q=0$, every constant $\theta$ is a solution.}
\label{Tab:PolarClassification}
\end{table}

\section{Radial motion and Vieta's formulae}

The radial motion is governed by Eq.~\eqref{drdlambda}, $\rho^4(\dd r/\dd\lambda)^2=R(r)$, or, in Mino time, by $(\dd r/\dd\gamma)^2=R(r)$. Since $\rho^4>0$, a geodesic can move only where $R(r)\geq0$, and its radial velocity vanishes only at the roots of $R$. As in Schwarzschild spacetime (Section~\ref{Sec:SchwVieta}), the simple roots outside the horizon are the turning points (Section~\ref{Sec:TurningPoints}), and multiple roots correspond to orbits of constant radius, which we treat in Section~\ref{Sec:KerrSpherical}. The radial problem thus reduces to locating the roots of the quartic~\eqref{KerrRExpanded} relative to the event horizon. In Mino time, this problem is decoupled from the polar one: the two motions interact only through the constants of motion $(\varepsilon,L_z,\mathcal Q)$, which must also satisfy the polar admissibility conditions of the previous section.

Three values of the radial polynomial are especially useful:
\begin{align}
 R(0)&=-a^2\mathcal Q\,,\label{RAtZero}\\
 R(r_{\rm H})&=
 \left[\varepsilon(r_{\rm H}^2+a^2)-aL_z\right]^2\geq0\,,\label{RAtHorizon}\\
 R(r)&\sim(\varepsilon^2-\kappa)r^4
 \qquad(r\to\infty)\,.\label{RAtInfinity}
\end{align}
Since $\Omega_{\rm H}=a/(r_{\rm H}^2+a^2)$, the horizon value vanishes if and only if $\varepsilon=\Omega_{\rm H}L_z$, the equality case of Eq.~\eqref{HorizonForwardCondition}. Moreover, $\Delta<0$ between the two horizons, so Eq.~\eqref{KerrRCompact} gives
\begin{equation}\label{RBetweenHorizons}
 R(r)>0\qquad(r_{\rm C}<r<r_{\rm H})\,,
\end{equation}
except for null geodesics with $K=0$, for which $R$ is a perfect square. In the notation of the previous section, the coefficient of $r^2$ in Eq.~\eqref{KerrRExpanded} is the polar coefficient $B$ of Eq.~\eqref{VorticalB}:
\begin{equation}\label{KerrRWithB}
 R=(\varepsilon^2-\kappa)r^4+2M\kappa r^3+Br^2+2MKr-a^2\mathcal Q\,.
\end{equation}
For $\varepsilon^2\neq\kappa$, let $r_1,\ldots,r_4$ denote all roots of $R$, including complex roots, negative roots, and roots at or inside the event horizon. Equation~\eqref{KerrRWithB} and Vieta's formulae give
\begin{align}
 \sum_i r_i&=-\frac{2M\kappa}{\varepsilon^2-\kappa}\,,
 \label{KerrVieta1}\\
 \sum_{i<j}r_ir_j&=\frac{B}{\varepsilon^2-\kappa}
 =a^2-\frac{L_z^2+\mathcal Q}{\varepsilon^2-\kappa}\,,
 \label{KerrVieta2}\\
 \sum_{i<j<k}r_ir_jr_k&=
 -\frac{2MK}{\varepsilon^2-\kappa}\,,
 \label{KerrVieta3}\\
 \prod_i r_i&=
 -\frac{a^2\mathcal Q}{\varepsilon^2-\kappa}\,.
 \label{KerrVieta4}
\end{align}
Two features are immediately informative. For null geodesics, the roots sum to zero, so that when they are all real and not all zero, at least one of them is negative. For $a\neq0$, the product of the roots has the sign of $-\mathcal Q$ when $\varepsilon^2>\kappa$ and that of $\mathcal Q$ when $\varepsilon^2<\kappa$; together with Eq.~\eqref{RAtZero}, it controls whether a root lies between $r=0$ and the horizon. However, the Kerr classification cannot be inferred from these relations alone: one must also impose $r>r_{\rm H}$, $R\geq0$, the polar admissibility conditions of the previous section, and future orientation, Eqs.~\eqref{HorizonForwardCondition} and~\eqref{ZAMOForwardCondition} \citep{CompereLiuLong:2022}.

\subsection{Timelike geodesics with \texorpdfstring{$\varepsilon^2>1$}{energy squared > 1}}\label{Sec:KerrUnbound}

Here $\varepsilon^2-1>0$, and $R(r_{\rm H})\geq0$. Writing $R=(\varepsilon^2-1)\prod_i(r-r_i)$, every complex-conjugate pair and every real root below the horizon contributes a positive factor at $r=r_{\rm H}$, and every real root above it contributes a negative one. Hence, if $R(r_{\rm H})>0$, the number of exterior roots, counted with multiplicity, is even: zero, two, or four. The last possibility is excluded by Eq.~\eqref{KerrVieta1}: four exterior roots would all be positive, whereas their sum is $-2M/(\varepsilon^2-1)<0$. At most two roots, therefore, lie outside the horizon, and the sign of $\mathcal Q$ determines where the others are.

\begin{itemize}
 \item If $\mathcal Q>0$, Eq.~\eqref{KerrVieta4} gives $\prod_ir_i<0$, so at least one root is real and negative. When two roots lie outside the horizon, the remaining two have a negative product, so they are real and of opposite sign; by Eq.~\eqref{RBetweenHorizons}, the positive one lies in $(0,r_{\rm C}]$.
 \item If $\mathcal Q=0$, one root is $r=0$, and the other three are the roots of $(\varepsilon^2-1)r^3+2Mr^2+Br+2MK$. For $K>0$, their product, $-2MK/(\varepsilon^2-1)$, is negative, so when two of them lie outside the horizon, the third is negative. For $K=0$, that is, $L_z=a\varepsilon$, the root $r=0$ is double, and the remaining quadratic, $(\varepsilon^2-1)r^2+2Mr-a^2$, has roots of opposite sign; by parity, the positive one lies inside the horizon, and there is no exterior root.
 \item If $\mathcal Q<0$, the polar motion is vortical, and the condition $B>0$ of Eq.~\eqref{VorticalConditions} makes the coefficient of $r^2$ in Eq.~\eqref{KerrRWithB} positive. Since $K\geq0$ by Eq.~\eqref{KNonnegative}, every coefficient of
 \begin{equation}
 R=(\varepsilon^2-1)r^4+2Mr^3+Br^2+2MKr+a^2|\mathcal Q|
 \end{equation}
 is non-negative, and the constant term is positive. Hence $R>0$ for every $r>0$: $R$ has no positive root at all, as Descartes' rule of signs (Section~\ref{Sec:Descartes}) also states, since the coefficients show no sign change.
\end{itemize}

The generic exterior patterns are, therefore, only two:
\begin{enumerate}
 \item no exterior root: a horizon-to-infinity transit orbit, interpreted as capture or escape according to its time orientation;
 \item two exterior roots $r_1>r_2$: a horizon-connected branch $r_{\rm H}<r\leq r_2$ and a scattering branch $r\geq r_1$.
\end{enumerate}
Two conclusions follow. First, \emph{no radially bound timelike geodesic with $\varepsilon^2>1$ exists in the domain of outer communication}; a double root of $R$ in this regime is necessarily an unstable spherical orbit separating capture from scattering. Second, \emph{every vortical timelike geodesic is a transit orbit}: it has no radial turning point outside the horizon and connects the horizon with infinity. Both statements follow from the signs of the coefficients alone, without computing a single root. They are not new. \citet{Wilkins:1972rs} proved the first, also for $\varepsilon^2=1$, by applying Descartes' rule to $R$ expanded about $r=M$, and showed that bound motion requires $\mathcal Q\geq0$; \citet{CompereLiuLong:2022} found that every radial root structure with $\mathcal Q<0$ has no exterior root. The argument given here only makes the role of Vieta's formulae explicit. The complete radial taxonomy, including the non-generic cases, is given by \citet{CompereLiuLong:2022}. Figure~\ref{Fig:KerrScattering} shows the two generic patterns, together with the double root that separates them, for three orbits that differ only in $\mathcal Q$.

Since a transit orbit reaches infinity, where Eq.~\eqref{ZAMOForwardCondition} requires $\varepsilon>0$, vortical timelike geodesics, in fact, have $\varepsilon>1$. Energies $\varepsilon<-1$ are possible only inside the ergoregion. For them, the transit and scattering branches are past-directed, and the only future-directed motion is the horizon-connected branch, with its turning point inside the ergoregion.

% Fig_KerrScattering.tex -- generated by checks/fig_kerr_scattering.py
% a=0.9M, epsilon=1.05, L_z=2M; Q = 0.92 Q_c, (1+2e-5) Q_c, 1.08 Q_c, Q_c=7.70855 M^2
\begin{figure}[htbp]
\centering
\colorlet{KScat}{blue!70!black}
\colorlet{KCapt}{red!75!black}
\colorlet{KCrit}{black}
\includegraphics{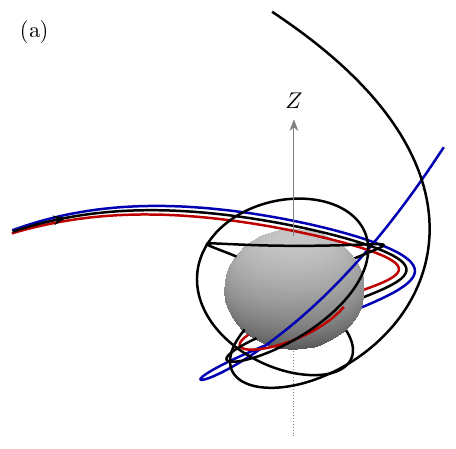}

\medskip
\includegraphics{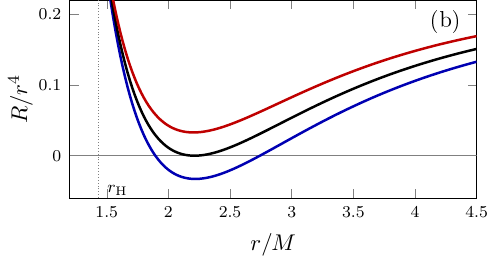}
\caption{Timelike geodesics with $\varepsilon^2>1$ for $a=0.9M$, $\varepsilon=1.05$ and $L_z=2M$, coming in from large $r$ with $\mathcal Q>0$. (a)~Trajectories in Kerr--Schild Cartesian coordinates, in which the horizon (gray) is regular; the arrow marks the common incoming direction. The three orbits differ only in $\mathcal Q$: for $\mathcal Q=1.08\,\mathcal Q_{\rm c}$, with $\mathcal Q_{\rm c}\simeq7.71M^2$ the critical value, the particle is scattered (blue); for $\mathcal Q=0.92\,\mathcal Q_{\rm c}$ it is captured (red); for $\mathcal Q=(1+2\times10^{-5})\,\mathcal Q_{\rm c}$ it winds about three times around the unstable spherical orbit $r_s\simeq2.21M$ before escaping (black). None of the orbits stays in a plane. (b)~The radial function $R/r^4$ for the same three values of $\mathcal Q$: two exterior roots (scattering), two nearly coincident roots close to $r_s$ (near-critical), and no exterior root (capture).}
\label{Fig:KerrScattering}
\end{figure}

\subsection{Marginally bound timelike geodesics}\label{Sec:KerrMarginal}

For $\varepsilon^2=1$, the quartic term vanishes and $A=0$, so that $B=-(L_z^2+\mathcal Q)$, and Eq.~\eqref{KerrRWithB} becomes the cubic
\begin{equation}\label{KerrMarginalR}
 R=2Mr^3-(L_z^2+\mathcal Q)r^2+2MKr-a^2\mathcal Q\,.
\end{equation}
By the polar classification (case $A=0$), motion requires $\mathcal Q\geq0$, and for $L_z\neq0$, $\mathcal Q=0$ means equatorial motion. The leading coefficient is positive and $R(r_{\rm H})\geq0$, so the parity argument of the previous subsection applies: for $R(r_{\rm H})>0$, the number of exterior roots, counted with multiplicity, is even; hence, zero or two for a cubic. The generic exterior topology is therefore either a transit orbit, with no turning point, or a pair of turning points $r_1>r_2$, with a horizon-connected branch $r_{\rm H}<r\leq r_2$ and a parabolic scattering branch $r\geq r_1$. A double root marks a marginally bound spherical orbit, which is unstable since $R\geq0$ on both sides of it. Near the horizon, these orbits look like those of Fig.~\ref{Fig:KerrScattering}; only the escaping branch changes character, reaching infinity with zero radial velocity.

For $\varepsilon=-1$, which is possible only inside the ergoregion, only the horizon-connected branch is future-directed. Together with the previous subsection, this shows that radially bound timelike motion in the Kerr exterior requires $\varepsilon^2<1$, exactly as in Schwarzschild spacetime.

Vieta's formulae for the three roots of the cubic are
\begin{align}
 r_1+r_2+r_3&=\frac{L_z^2+\mathcal Q}{2M}\,,\\
 r_1r_2+r_1r_3+r_2r_3&=K\,,\\
 r_1r_2r_3&=\frac{a^2\mathcal Q}{2M}\,.
\end{align}
Since $\mathcal Q\geq0$, every coefficient of $R(-r)$ is non-positive, so $R$ has no negative roots, and the product of the roots is non-negative. When two roots lie outside the horizon, the third is therefore real and, by Eq.~\eqref{RBetweenHorizons}, lies in $[0,r_{\rm C}]$; it vanishes if and only if $a\mathcal Q=0$.

The equatorial case shows these relations at work. For $\mathcal Q=0$ and $\varepsilon=1$, one root is $r=0$, and the other two are the roots of $2Mr^2-L_z^2r+2M(L_z-a)^2$, with a sum of $L_z^2/(2M)$ and a product of $(L_z-a)^2$. They coincide at $r_c$ if and only if
\begin{equation}
 r_c=|L_z-a|\,,\qquad L_z^2=4Mr_c\,.
\end{equation}
For prograde orbits, $L_z>a$, and for retrograde orbits, $L_z<0$, the solutions outside the horizon are
\begin{align}
 r_{\rm mb}^{\rm pro}&=2M-a+2\sqrt{M(M-a)}, & L_{\rm mb}^{\rm pro}&=r_{\rm mb}^{\rm pro}+a\,,\\
 r_{\rm mb}^{\rm ret}&=2M+a+2\sqrt{M(M+a)}, & L_{\rm mb}^{\rm ret}&=-\bigl(r_{\rm mb}^{\rm ret}-a\bigr)\,,
\end{align}
the radii and angular momenta of the marginally bound circular orbits. For $a=0$ both reduce to $r_{\rm mb}=4M$ and $|L_z|=4M$, the Schwarzschild values. Their dependence on the spin is shown in Fig.~\ref{KerrRadii}.

\subsection{Bound timelike geodesics: \texorpdfstring{$\varepsilon^2<1$}{energy squared < 1}}\label{Sec:KerrBound}

Now $\varepsilon^2-1<0$, so $R\to-\infty$ as $r\to\infty$, whereas $R(r_{\rm H})\geq0$. The parity argument of Section~\ref{Sec:KerrUnbound}, with the leading coefficient now negative, shows that for $R(r_{\rm H})>0$ the number of exterior roots, counted with multiplicity, is odd:
\begin{enumerate}
 \item one exterior root $r_1$: only the horizon-connected branch $r_{\rm H}<r\leq r_1$ is allowed;
 \item three exterior roots $r_1>r_2>r_3$: the allowed intervals are the horizon-connected branch $r_{\rm H}<r\leq r_3$ and the bound branch
 \begin{equation}
 r_2\leq r\leq r_1\,.
 \end{equation}
\end{enumerate}
For the bound branch, $r_{\rm p}=r_2$ is the periapsis, and $r_{\rm a}=r_1$ is the apoapsis.

By the polar classification (case $A<0$), $\mathcal Q\geq0$ in this energy regime, with $\mathcal Q=0$ only for equatorial orbits. Since also $B=A-L_z^2-\mathcal Q<0$, every coefficient of $R(-r)$ is non-positive, so $R$ has no negative root. Moreover, $R(0)=-a^2\mathcal Q\leq0\leq R(r_{\rm H})$, and $R>0$ between the horizons by Eq.~\eqref{RBetweenHorizons}, so at least one root lies in $[0,r_{\rm C}]$. In the three-root case, the fourth root $r_4$ therefore lies in $[0,r_{\rm C}]$; by Eq.~\eqref{KerrVieta4}, it is positive when $a\mathcal Q\neq0$ and vanishes when $a\mathcal Q=0$. In the one-root case, the remaining three roots are either a real root in $[0,r_{\rm C}]$ and a complex-conjugate pair, or three real roots in $[0,r_{\rm C}]$.

The sign of the discriminant $\Delta_4$ of $R$ (Section~\ref{Sec:Discriminants}) distinguishes these possibilities only in part. If $\Delta_4<0$, there is one complex-conjugate pair, and the orbit belongs to the one-root case. If $\Delta_4>0$, all four roots are real and non-negative, but either three or only one of them lie outside the horizon. The second possibility does occur, as Example~\ref{Ex:KerrDiscriminant} shows: three roots lie between $0$ and $r_{\rm C}$, and only one outside the horizon (Fig.~\ref{Fig:KerrBound}c, compared with a genuinely bound case in panel~b). The two cases are separated by Descartes' rule applied at the horizon, which is exact when all roots are real (Section~\ref{Sec:Descartes}). Writing $R(r_{\rm H}+h)=\sum_{k=0}^{4}R^{(k)}(r_{\rm H})\,h^k/k!$, the number of roots with $r>r_{\rm H}$ equals the number of sign changes in the sequence
\begin{equation}\label{KerrHorizonDescartes}
 \varepsilon^2-1\,,\quad \tfrac16R'''(r_{\rm H})\,,\quad \tfrac12R''(r_{\rm H})\,,\quad R'(r_{\rm H})\,,\quad R(r_{\rm H})\,.
\end{equation}
A bound branch therefore exists if and only if $\Delta_4>0$ and this sequence has three sign changes.

A standard parametrization of the bound radial turning points is
\begin{equation}
 r_{\rm p}=\frac{pM}{1+e}\,,
 \qquad
 r_{\rm a}=\frac{pM}{1-e}\,,
 \qquad 0\leq e<1\,,
\end{equation}
where $p$ and $e$ are the relativistic semi-latus rectum and eccentricity. These may be supplemented by the inclination angle $i$, defined by
\begin{equation}\label{KerrInclination}
 \cos i:=\frac{L_z}{\sqrt{L_z^2+\mathcal Q}}\,,
\end{equation}
which is smaller than $\pi/2$ for prograde and larger than $\pi/2$ for retrograde orbits and, for $a=0$, reduces to the inclination of the orbital plane introduced in the polar section. The triple $(p,e,i)$ may be used in place of $(\varepsilon,L_z,\mathcal Q)$; the inclination $i$ is the one used by \citet{DrascoHughes:2004}, while \citet{Schmidt:2002qk} and \citet{FujitaHikida:2009} use related inclination parameters built from the polar turning point.

The boundary between bound and plunging orbits is the \emph{separatrix}, the Kerr counterpart of the curve $p=6+2e$ of Section~\ref{Sec:SchwVieta}. On it, the periapsis $r_2$ merges with the turning point $r_3$ of the horizon-connected branch, and
\begin{equation}
 R=(\varepsilon^2-1)(r-r_1)(r-r_2)^2(r-r_4)\,.
\end{equation}
The double root $r_2$ is an unstable spherical orbit, and the bound branch becomes a \emph{homoclinic} orbit: it leaves the spherical orbit in the asymptotic past, reaches the apoapsis $r_1$, and returns to the spherical orbit in the asymptotic future. Figure~\ref{Fig:KerrBoundOrbits} shows a generic bound orbit and, for slightly different values of $\mathcal Q$, an orbit close to homoclinic and one that has crossed the separatrix and plunges.

The separatrix lies on the zero set of $\Delta_4$, which, however, contains more. Double roots outside the horizon also occur as $r_1=r_2$, the stable spherical orbits at which bound orbits reach $e=0$ (Section~\ref{Sec:KerrSpherical}), and double roots inside the Cauchy horizon make $\Delta_4$ vanish without affecting the exterior motion. \citet{SteinWarburton:2020} showed that, in terms of $p$, $e$ and an inclination parameter, the separatrix is the zero set of a single polynomial, of degree twelve in $p$; for equatorial orbits, it reduces to a quartic in $p$ \citep{OShaughnessy:2003}. An explicit parametric solution, in terms of the radius of the unstable spherical orbit, was found by \citet{LevinPerezGiz:2009} in the equatorial case and extended to inclined orbits by \citet{NgTeo:2025}. Their construction is itself an application of Vieta's formulae: matching the constant and linear coefficients of $R$ with those of the factorized form above gives $r_4$, and then the apoapsis $r_1$, in terms of $r_2$. At $e=0$, the three roots $r_1$, $r_2$, $r_3$ coalesce, and the separatrix ends on the innermost stable spherical orbit, where $R''=0$; as $e\to1$, the apoapsis recedes to infinity, $\varepsilon\to1$, and the separatrix ends on the marginally bound spherical orbit of Section~\ref{Sec:KerrMarginal}. For $a=0$ it reduces to $p=6+2e$ for every inclination. Figure~\ref{Fig:KerrBound}a shows the separatrix for $a=0.9M$: frame dragging moves it inward for prograde orbits and outward for retrograde ones, the polar orbits lying inside the Schwarzschild line, and starting at $p\simeq5.44$ for $e=0$.

% Fig_KerrBound.tex -- generated by checks/fig_kerr_bound.py
% (a) separatrix for a=0.9M from the unstable spherical orbits (Ng & Teo 2025, Eqs. 4.1), checked against R(r_u)=R'(r_u)=0
% (b),(c) a=0.86M, epsilon=0.9: (L_z,Q)=(2.3M,1.0M^2) and (1.13M,0.0066M^2); both Delta_4>0
\begin{figure}[htbp]
\centering
\includegraphics{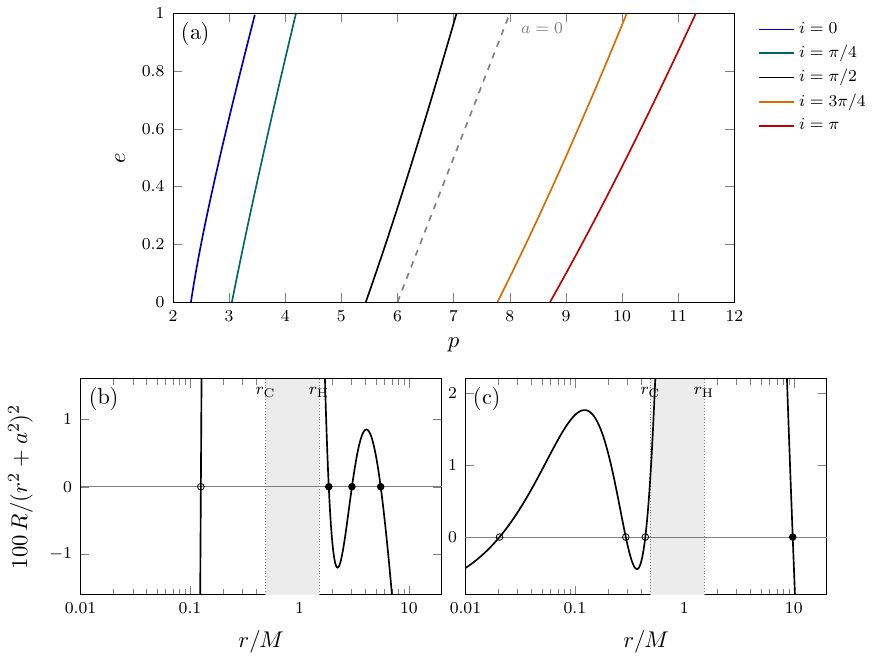}
\caption{(a)~Separatrix between bound and plunging timelike orbits in the $(p,e)$ plane for $a=0.9M$ and several inclinations $i$, Eq.~\eqref{KerrInclination}, from prograde equatorial ($i=0$) to retrograde equatorial ($i=\pi$); bound orbits lie to the right of each curve, plunging ones to the left. The dashed line is the Schwarzschild separatrix $p=6+2e$. Each curve is traced by the homoclinic orbits of the unstable spherical orbits; at $e=0$ it starts at the innermost stable spherical orbit, $p=r_{\rm ISCO}/M$ in the equatorial cases. (b), (c)~Two radial functions with the same $a=0.86M$ and $\varepsilon=0.9$, both with $\Delta_4>0$: (b)~$L_z=2.3M$, $\mathcal Q=M^2$, three exterior roots (filled dots) and a bound branch $r_2\leq r\leq r_1$; (c)~$L_z=1.13M$, $\mathcal Q=0.0066M^2$, three roots inside the Cauchy horizon (open dots) and a single exterior root, so that no bound branch exists. In (b) and (c) the radial function is plotted as $100\,R/(r^2+a^2)^2$, which, unlike the $R/r^4$ of Figs.~\ref{Fig:KerrScattering} and~\ref{Fig:KerrBoundOrbits}, stays finite as $r\to0$, where the roots inside the Cauchy horizon lie. The shaded band is the region between the horizons, where $R>0$. The sequence~\eqref{KerrHorizonDescartes} has three sign changes in~(b) and one in~(c).}
\label{Fig:KerrBound}
\end{figure}

% Fig_KerrBoundOrbits.tex -- generated by checks/fig_kerr_bound_orbits.py
% a=0.9M, epsilon=0.92249, L_z=1.82900M; Q = 4.75116 M^2 (p=5, e=0.3, i=50 deg), (1+2e-4) Q_c, 0.92 Q_c; Q_c=4.24900 M^2
\begin{figure}[htbp]
\centering
\colorlet{KScat}{blue!70!black}
\colorlet{KCapt}{red!75!black}
\colorlet{KCrit}{black}
\includegraphics{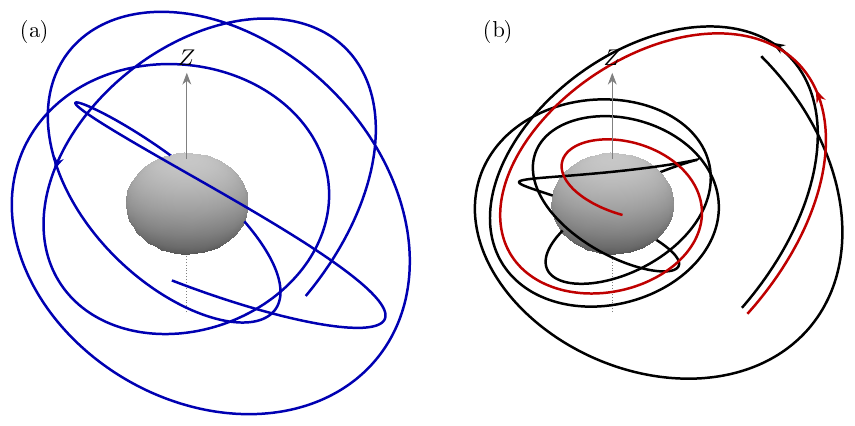}

\medskip
\includegraphics{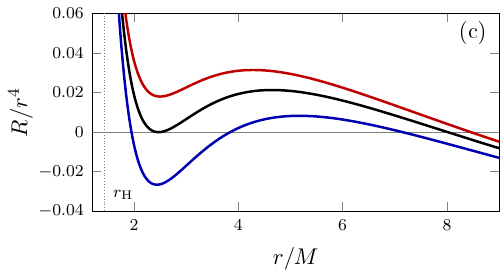}
\caption{Timelike geodesics with $\varepsilon^2<1$ for $a=0.9M$, $\varepsilon\simeq0.922$ and $L_z\simeq1.83M$, all starting at their apoapsis $r_1$ in the equatorial plane, in Kerr--Schild Cartesian coordinates; the arrows mark the direction of motion. The orbits differ only in $\mathcal Q$. (a)~A bound orbit, $\mathcal Q\simeq4.75M^2$, with $p=5$, $e=0.3$ and $i=50^\circ$, over two radial periods: it oscillates in latitude between $\theta\simeq40^\circ$ and $140^\circ$ while its orbit precesses about the spin axis. (b)~Near the separatrix value $\mathcal Q_{\rm c}\simeq4.25M^2$: for $\mathcal Q=(1+2\times10^{-4})\,\mathcal Q_{\rm c}$ the orbit is close to homoclinic and winds about three times around the unstable spherical orbit $r\simeq2.47M$ before returning to its apoapsis (black); for $\mathcal Q=0.92\,\mathcal Q_{\rm c}$ the periapsis has disappeared and the particle plunges into the horizon (red). (c)~The radial function $R/r^4$ for the three values of $\mathcal Q$: three exterior roots (blue), two nearly coincident roots close to the spherical orbit (black), and a single exterior root (red).}
\label{Fig:KerrBoundOrbits}
\end{figure}

\section{Multiple roots, spherical orbits, and stability}\label{Sec:KerrSpherical}

The root patterns of the previous sections can change in only two ways: two roots meet, or a root crosses the horizon. The first event is the vanishing of the discriminant, $\Delta_4=0$ (Section~\ref{Sec:Discriminants}); the second is $R(r_{\rm H})=0$, that is, $\varepsilon=\Omega_{\rm H}L_z$ by Eq.~\eqref{RAtHorizon}. Multiple roots are therefore the boundaries of the classification in the space of constants of motion, and they also carry geodesics of their own. Differentiating $(\dd r/\dd\gamma)^2=R(r)$ gives $\dd^2r/\dd\gamma^2=R'(r)/2$, so $r=r_s$ is a solution if and only if
\begin{equation}\label{SphericalConditions}
 R(r_s)=0\,,\qquad R'(r_s)=0\,.
\end{equation}
Such orbits are called \emph{spherical}; they need not remain at constant $\theta$. Equatorial circular orbits are the special case $\mathcal Q=0$, $\theta=\pi/2$. Spherical timelike orbits and their parameter domain were analyzed by \citet{Wilkins:1972rs}, for the extreme black hole, and, in a modern unified treatment, by \citet{Teo:2021}.

The radial stability of a spherical orbit at fixed constants of motion is determined by the sign of $R''(r_s)$: for $R''(r_s)<0$, $R<0$ on both sides of $r_s$, the orbit is stable; for $R''(r_s)>0$, $R>0$ on both sides, it is unstable. This sign can be read off from the position of the other two roots, $\tilde r_1$ and $\tilde r_2$. Writing $R=(\varepsilon^2-\kappa)(r-r_s)^2(r-\tilde r_1)(r-\tilde r_2)$, one finds
\begin{equation}\label{KerrDoubleRootStability}
 R''(r_s)=2(\varepsilon^2-\kappa)(r_s-\tilde r_1)(r_s-\tilde r_2)\,.
\end{equation}
We take $R(r_{\rm H})>0$ and $r_s>r_{\rm H}$, and use the parity argument of Section~\ref{Sec:KerrUnbound}.
\begin{itemize}
 \item $\varepsilon^2>\kappa$ (unbound timelike geodesics and all null geodesics). At most two exterior roots exist, counted with multiplicity: for timelike geodesics because of Eq.~\eqref{KerrVieta1}, for null geodesics because the roots sum to zero. The double root uses both, so $\tilde r_1$ and $\tilde r_2$ are either real and below the horizon or complex conjugate. In both cases $(r_s-\tilde r_1)(r_s-\tilde r_2)>0$, and the orbit is unstable. In particular, \emph{every spherical photon orbit in the domain of outer communication is unstable}.
 \item $\varepsilon^2=1$. Now $R=2M(r-r_s)^2(r-\tilde r)$ and $R''(r_s)=4M(r_s-\tilde r)$. By parity, $\tilde r$ is not exterior, so the marginally bound spherical orbits are unstable.
 \item $\varepsilon^2<1$. Since $R(r_{\rm H})>0$ while $\varepsilon^2-1<0$, the factorized form of $R$ evaluated at the horizon requires $(r_{\rm H}-\tilde r_1)(r_{\rm H}-\tilde r_2)<0$: the two roots are real, one inside and one outside the horizon, say $\tilde r_2<r_{\rm H}<\tilde r_1$. The orbit is stable if $\tilde r_1<r_s$, when the double root is the outermost root, $r_1=r_2$; it is the limit $e\to0$ of the bound orbits of Section~\ref{Sec:KerrBound}. It is unstable if $\tilde r_1>r_s$, when $r_2=r_3$; this is the separatrix, and $\tilde r_1$ is the apoapsis of the homoclinic orbit.
 \item The two cases with $\varepsilon^2<1$ meet at a triple root, $\tilde r_1=r_s$, where $R''(r_s)=0$. This marginally stable orbit is the \emph{innermost stable spherical orbit} for the given inclination, the end point $e=0$ of the separatrix.
\end{itemize}
Stability is thus a statement about root patterns: a stable spherical orbit is a bound orbit whose two turning points have merged, an unstable one is the meeting point of a bound or scattering branch with the horizon-connected branch.

\subsection{Equatorial circular timelike orbits}

For $\mathcal Q=0$, Eq.~\eqref{KerrVieta4} shows that $r=0$ is a root of $R$. A circular orbit in the equatorial plane at $r=r_c$ therefore corresponds to
\begin{equation}
 R=(\varepsilon^2-1)\,r\,(r-r_c)^2(r-\tilde r)\,,
\end{equation}
and Vieta's formulae for the three non-trivial coefficients give
\begin{equation}\label{KerrCircVieta}
 2r_c+\tilde r=\frac{2M}{1-\varepsilon^2}\,,
 \qquad
 r_c^2+2r_c\tilde r=a^2+\frac{L_z^2}{1-\varepsilon^2}\,,
 \qquad
 r_c^2\tilde r=\frac{2M(L_z-a\varepsilon)^2}{1-\varepsilon^2}\,.
\end{equation}
For given $r_c$, these three relations determine $\varepsilon$, $L_z$, and the remaining root $\tilde r$. Their solutions are the energy and angular momentum found by \citet{BardeenPressTeukolsky:1972},
\begin{align}
 \varepsilon_{\rm pro}
 &=\frac{r^{3/2}-2Mr^{1/2}+aM^{1/2}}
 {r^{3/4}\left(r^{3/2}-3Mr^{1/2}+2aM^{1/2}\right)^{1/2}}\,,
 \label{KerrEPro}\\
 L_{\rm pro}
 &=\frac{M^{1/2}\left(r^2-2a\sqrt{Mr}+a^2\right)}
 {r^{3/4}\left(r^{3/2}-3Mr^{1/2}+2aM^{1/2}\right)^{1/2}}\,,
 \label{KerrLPro}\\
 \varepsilon_{\rm ret}
 &=\frac{r^{3/2}-2Mr^{1/2}-aM^{1/2}}
 {r^{3/4}\left(r^{3/2}-3Mr^{1/2}-2aM^{1/2}\right)^{1/2}}\,,
 \label{KerrERet}\\
 L_{\rm ret}
 &=-\frac{M^{1/2}\left(r^2+2a\sqrt{Mr}+a^2\right)}
 {r^{3/4}\left(r^{3/2}-3Mr^{1/2}-2aM^{1/2}\right)^{1/2}}\,,
 \label{KerrLRet}
\end{align}
where \enquote{pro} and \enquote{ret} denote prograde and retrograde motion for $a>0$, and $r=r_c$.

The first of Eqs.~\eqref{KerrCircVieta} gives the remaining root explicitly, $\tilde r=2M/(1-\varepsilon^2)-2r_c$. For $\varepsilon^2<1$, the stability condition $\tilde r<r_c$ of the previous section becomes $1-\varepsilon^2>2M/(3r_c)$; for $\varepsilon^2\geq1$, circular orbits are unstable. Both cases are covered by the single criterion
\begin{equation}\label{KerrCircStability}
 \varepsilon^2<1-\frac{2M}{3r_c}\qquad\hbox{(stable circular orbit)}\,.
\end{equation}
The \emph{innermost stable circular orbit} (ISCO) is the triple root $\tilde r=r_c$, where $R=(\varepsilon^2-1)\,r\,(r-r_c)^3$. Equations~\eqref{KerrCircVieta} then reduce to
\begin{equation}\label{KerrISCOVieta}
 1-\varepsilon^2=\frac{2M}{3r_c}\,,
 \qquad
 L_z^2=\frac{2M\left(3r_c^2-a^2\right)}{3r_c}\,,
 \qquad
 (L_z-a\varepsilon)^2=\frac{r_c^2}{3}\,.
\end{equation}
Substituting the first two relations into the third and eliminating the square roots yields
\begin{equation}\label{KerrISCOEquation}
 r_c^2-6Mr_c\pm8a\sqrt{Mr_c}-3a^2=0\,,
\end{equation}
with the upper sign for prograde and the lower sign for retrograde orbits. The relevant roots are
\begin{align}
 r_{\rm ISCO}^{\rm pro}
 &=M\left[3+Z_2-\sqrt{(3-Z_1)(3+Z_1+2Z_2)}\right]\,,\\
 r_{\rm ISCO}^{\rm ret}
 &=M\left[3+Z_2+\sqrt{(3-Z_1)(3+Z_1+2Z_2)}\right]\,,
\end{align}
where
\begin{align}
 Z_1&:=1+\left(1-\frac{a^2}{M^2}\right)^{1/3}
 \left[\left(1+\frac aM\right)^{1/3}
 +\left(1-\frac aM\right)^{1/3}\right]\,,\\
 Z_2&:=\left(3\frac{a^2}{M^2}+Z_1^2\right)^{1/2}\,,
\end{align}
and the first two of Eqs.~\eqref{KerrISCOVieta} give the corresponding energy and angular momentum, $\varepsilon_{\rm ISCO}$ and $L_{\rm ISCO}$, directly. For $a=0$, one recovers $r_{\rm ISCO}=6M$, $\varepsilon^2_{\rm ISCO}=8/9$, and $L_{\rm ISCO}^2=12M^2$. In the opposite direction, as $\varepsilon\to1$, the first of Eqs.~\eqref{KerrCircVieta} sends $\tilde r$ to infinity: the quartic degenerates into the cubic of Section~\ref{Sec:KerrMarginal}, whose double root is the marginally bound circular orbit $r_{\rm mb}$ derived there. The dependence of these radii on the spin is shown in Fig.~\ref{KerrRadii}.

% Fig_KerrRadii.tex -- generated by checks/fig_kerr_radii.py
\begin{figure}[htbp]
\centering
\colorlet{KHor}{black}
\colorlet{KPh}{blue!70!black}
\colorlet{KMb}{red!75!black}
\colorlet{KIsco}{teal!80!black}
\includegraphics{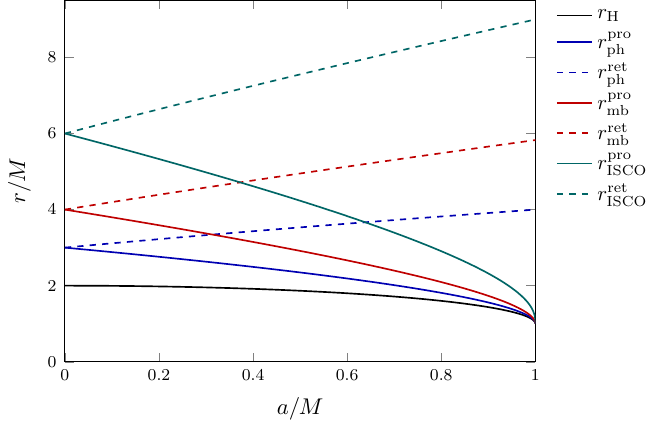}
\caption{Characteristic equatorial radii as functions of the spin: event horizon $r_{\rm H}$, light ring $r_{\rm ph}$, marginally bound orbit $r_{\rm mb}$, and ISCO. Solid lines are prograde and dashed lines retrograde. At $a=0$ they reduce to $2M$, $3M$, $4M$, and $6M$; at $a=M$ the prograde radii all tend to $M$, while the retrograde ones tend to $4M$, $(3+2\sqrt2)M$, and $9M$.}
\label{KerrRadii}
\end{figure}

In the equatorial plane, the same information can be displayed as an effective potential, as in Section~\ref{Sec:EffPot}, although it must be read off from $R$ rather than added to it. For $\mathcal Q=0$, Eq.~\eqref{KerrRCompact} is a quadratic in $\varepsilon$,
\begin{equation}
 R=\alpha(r)\,\varepsilon^2-4aMrL_z\,\varepsilon
 +a^2L_z^2-\Delta\left(r^2+L_z^2\right)\,,
 \qquad
 \alpha(r):=(r^2+a^2)^2-a^2\Delta\,,
\end{equation}
so that $R=\alpha(\varepsilon-V_{\rm eq})(\varepsilon-\tilde V_{\rm eq})$, with
\begin{equation}\label{KerrVeqDef}
 V_{\rm eq}(r)=\frac{2aMrL_z+r\sqrt{\Delta\left[r^4+L_z^2r^2+a^2(r^2+2Mr)\right]}}{\alpha(r)}\,,
\end{equation}
and $\tilde V_{\rm eq}$ given by the same expression with the opposite sign in front of the square root. Future-directed motion outside the horizon is allowed where $V_{\rm eq}(r)\leq\varepsilon$; the branch $\varepsilon\leq\tilde V_{\rm eq}$ is its time reverse. At the horizon, $V_{\rm eq}=\tilde V_{\rm eq}=\Omega_{\rm H}L_z$, so the condition $\varepsilon\geq V_{\rm eq}$ reduces there to Eq.~\eqref{HorizonForwardCondition}; for $L_z<0$, the potential is negative near the horizon, which is the equatorial picture of the negative-energy orbits of the Penrose process. The circular orbits above are the extrema of $V_{\rm eq}$, the ISCO is an inflection point, and the marginally bound orbit is a maximum with $V_{\rm eq}=1$. For $\varepsilon>0$, the root topologies of Table~\ref{tab:SchwPotentialClassification} carry over with $\varepsilon$ in place of $\varepsilon^2$, the local maximum and minimum of $V_{\rm eq}$ in place of $2V_{\max}$ and $2V_{\min}$, and the thresholds $|L_z|=|L_{\rm ISCO}|$ and $|L_z|=|L_{\rm mb}|$, taken separately for prograde and retrograde orbits, in place of $L^2=12M^2$ and $L^2=16M^2$. Off the equatorial plane, $R$ is still quadratic in $\varepsilon$, but the potential then depends on both $L_z$ and $\mathcal Q$, and its shape would have to be controlled over a two-parameter family; the root arguments of this chapter apply directly instead.

\section{Null geodesics and the Kerr photon shell}\label{Sec:KerrNull}

Outside the ergoregion, future orientation requires $\varepsilon>0$; null geodesics with $\varepsilon=0$ exist only inside it, and we do not consider them further. For $\varepsilon\neq0$, introduce the impact parameters
\begin{equation}
 \xi:=\frac{L_z}{\varepsilon}\,,
 \qquad
 \eta:=\frac{\mathcal Q}{\varepsilon^2}\,,
\end{equation}
which for $a=0$ satisfy $\xi^2+\eta=b^2$, with $b$ the impact parameter of Chapter~\ref{Sec:Schgeod}. After division by $\varepsilon^2$, the radial and polar polynomials are
\begin{align}
 \mathcal R(r)
 &:=
 \frac{R(r)}{\varepsilon^2}
 =r^4+(a^2-\xi^2-\eta)r^2
 +2Mr\left[\eta+(\xi-a)^2\right]-a^2\eta\,,
 \label{NullKerrR}\\
 \frac{F(z)}{\varepsilon^2}
 &=-a^2z^2+(a^2-\xi^2-\eta)z+\eta\,.
 \label{NullKerrF}
\end{align}
The coefficient $a^2-\xi^2-\eta=B/\varepsilon^2$ appears in both. For the four radial roots,
\begin{align}
 \sum_i r_i&=0\,,\\
 \sum_{i<j}r_ir_j&=a^2-\xi^2-\eta\,,\\
 \sum_{i<j<k}r_ir_jr_k&=-2M\left[\eta+(\xi-a)^2\right]\,,\\
 \prod_i r_i&=-a^2\eta\,.
\end{align}
Since $\mathcal R(r_{\rm H})\geq0$ and $\mathcal R>0$ at large $r$, the parity argument of Section~\ref{Sec:KerrUnbound} gives an even number of exterior roots \citep[cf.][]{GrallaLupsasca:2020}, and four are excluded because the roots sum to zero. Away from multiple roots, a photon therefore has either no exterior turning point, in which case it is captured (or, time-reversed, emerges from the white-hole region), or two, in which case it is scattered at the outer one. In particular, there is no finite-width band of radially bound null motion \citep{Teo:2003,GrallaLupsasca:2020}.

For $\eta<0$, the photon is vortical, and the condition $B>0$ of Eq.~\eqref{VorticalConditions} states that the coefficient of $r^2$ in Eq.~\eqref{NullKerrR} is positive. Every coefficient of $\mathcal R$ is then non-negative, with a positive constant term, so $\mathcal R>0$ for $r>0$: \emph{vortical photons have no radial turning point}, in agreement with the classification of \citet{GrallaLupsasca:2020}. They fall into the black hole or emerge from the white-hole region.

For $\eta>0$, the product of the roots is negative, so an odd number of roots are real and negative. Suppose three were negative, $-x,-y,-z$, and one positive, $w$. The vanishing sum gives $w=x+y+z$, and the third Vieta relation becomes
\begin{equation}
 w(xy+yz+zx)-xyz=-2M\left[\eta+(\xi-a)^2\right]<0\,.
\end{equation}
But $(x+y+z)(xy+yz+zx)\geq9xyz$ by the arithmetic--geometric mean inequality applied to each factor, a contradiction. Hence, when all roots are real, exactly one is negative and three are positive. For $\eta=0$, one root is $r=0$, and the remaining cubic, $r^3+(a^2-\xi^2)r+2M(\xi-a)^2$, has roots summing to zero and a non-positive product, hence a negative root: the possible outcomes are the same as for $\eta>0$.

The discriminant of the null quartic~\eqref{NullKerrR} factorizes; a direct computation gives
\begin{equation}\label{KerrNullDisc}
 \Delta_4=16\left[\eta+(\xi-a)^2\right]^2\mathcal S(\xi,\eta)\,,
\end{equation}
where, with $u:=\eta+\xi^2$,
\begin{align}
 \mathcal S={}&M^2u^2\left(u-27M^2\right)+108M^4a\xi u
 \nonumber\\
 &+a^2\left(-u^3+u^2\xi^2+33M^2u^2-36M^2u\xi^2-54M^4u-108M^4\xi^2\right)
 \nonumber\\
 &-4a^3\xi\left(u^2-u\xi^2-27M^4\right)
 +a^4\left(4\xi^4-2u\xi^2-2u^2+36M^2\xi^2-33M^2u-27M^4\right)
 \nonumber\\
 &-4a^5\xi\left(u-\xi^2\right)
 +a^6\left(\xi^2-u-M^2\right)\,.
\end{align}
For $a=0$ this reduces to $M^2u^2(u-27M^2)$, and $u=\xi^2+\eta$ is the square of the Schwarzschild impact parameter, so that $\mathcal S=0$ reproduces $b=b_c$. The prefactor vanishes for the principal null geodesics, $K=0$, for which $\mathcal R$ is a perfect square.

The sign of $\Delta_4$ does not by itself separate capture from scattering. If $\Delta_4<0$, $\mathcal R$ has two real roots of opposite sign (their product has the sign of $\prod_ir_i<0$, because the complex pair contributes a positive factor); since the number of exterior roots is even, the positive root lies inside the horizon, and the photon is captured. If $\Delta_4>0$, the three positive roots lie either two outside the horizon and one inside, giving scattering, or all three in $(0,r_{\rm C}]$, giving capture. The latter, which is region II in the classification of \citet{GrallaLupsasca:2020}, occurs near extremality at small $\eta$: for $a=0.99M$, $\xi=1.5M$, $\eta=0.004M^2$, the positive roots are $0.0076M$, $0.5174M$ and $0.7724M$, all below $r_{\rm C}\simeq0.8589M$. Descartes' rule at the horizon, Eq.~\eqref{KerrHorizonDescartes}, settles the matter: since $\mathcal R'''(r_{\rm H})=24r_{\rm H}>0$, a photon with $\eta\geq0$ is scattered if and only if $\Delta_4>0$ and at least one of $\mathcal R'(r_{\rm H})$ and $\mathcal R''(r_{\rm H})$ is negative.

The boundary between capture and scattering is formed by the spherical photon orbits. For a double root at $r=r_s$, the conditions
\begin{equation}
 \mathcal R(r_s)=0\,,\qquad \mathcal R'(r_s)=0
\end{equation}
give \citep{Teo:2003}
\begin{align}
 \xi(r_s)
 &=\frac{r_s^2(r_s-3M)+a^2(r_s+M)}
 {a(M-r_s)}\,,\label{PhotonXi}\\
 \eta(r_s)
 &=\frac{r_s^3\left[4a^2M-r_s(r_s-3M)^2\right]}
 {a^2(r_s-M)^2}\,.\label{PhotonEta}
\end{align}
The apparent singularity of these expressions at $a=0$ reflects the restoration of spherical symmetry; the Schwarzschild photon sphere must be recovered before dividing by $a$. By Section~\ref{Sec:KerrSpherical}, every spherical photon orbit in the domain of outer communication is radially unstable \citep[see also][]{Teo:2003}.

Since vortical photons have no turning point, spherical photon orbits require $\eta\geq0$. By Eq.~\eqref{PhotonEta}, this restricts $r_s$ to the \emph{photon shell} bounded by the two equatorial light rings, where $\eta=0$ \citep{Teo:2003}. The radii of the light rings, Eqs.~\eqref{PhotonRingPro}--\eqref{PhotonRingRet} below, were found by \citet{BardeenPressTeukolsky:1972}; here they are obtained directly from Vieta's formulae. For $\eta=0$, $\mathcal R=r\,(r^3+(a^2-\xi^2)r+2M(\xi-a)^2)$; since the cubic has no $r^2$ term, a double root $r_{\rm ph}$ forces the third root to be $-2r_{\rm ph}$, so that the cubic equals $(r-r_{\rm ph})^2(r+2r_{\rm ph})=r^3-3r_{\rm ph}^2r+2r_{\rm ph}^3$. Comparing coefficients,
\begin{equation}
 \xi^2=a^2+3r_{\rm ph}^2\,,\qquad (\xi-a)^2=\frac{r_{\rm ph}^3}{M}\,,
\end{equation}
and eliminating $\xi$ gives $\sqrt{r_{\rm ph}}\,(r_{\rm ph}-3M)=\mp2a\sqrt M$, that is, $r_{\rm ph}(r_{\rm ph}-3M)^2=4a^2M$, the vanishing of the bracket in Eq.~\eqref{PhotonEta}. The upper sign corresponds to prograde ($\xi>0$) and the lower to retrograde ($\xi<0$) orbits, and the roots outside the horizon are
\begin{align}
 r_{\rm ph}^{\rm pro}
 &=2M\left[1+\cos\left(\frac23\arccos\left(-\frac aM\right)\right)\right]\,,\label{PhotonRingPro}\\
 r_{\rm ph}^{\rm ret}
 &=2M\left[1+\cos\left(\frac23\arccos\left(\frac aM\right)\right)\right]\,.\label{PhotonRingRet}
\end{align}
For $a=0$, both reduce to $r_{\rm ph}=3M$ with $\xi^2=27M^2$, the critical impact parameter $b_c=3\sqrt3\,M$ of Chapter~\ref{Sec:Schgeod}. Interior values of $r_s$ describe constant-$r$ photons oscillating in latitude (Fig.~\ref{Fig:KerrPhotonOrbits}). Near the prograde light ring they stay close to the equator and co-rotate; the band widens inward from both light rings until, at the radius where $\xi=0$, the orbit passes over the poles; near the retrograde light ring the orbits again hug the equator, now counter-rotating. Together, these orbits form the critical set separating capture from escape, and they underlie the critical curve of a Kerr black hole. Figure~\ref{KerrShell} shows the curve $(\xi(r_s),\eta(r_s))$ traced as $r_s$ spans the photon shell. Photons with $\eta>0$ whose impact parameters lie in the region enclosed by the curve and the $\xi$ axis are captured; those outside it are scattered. Since a double root of $\mathcal R$ is a zero of its discriminant, the curve is the part of the zero set of $\mathcal S$, Eq.~\eqref{KerrNullDisc}, on which the double root lies outside the horizon \citep[see also][]{GrallaLupsasca:2020}. The zero set can have a second branch with $\eta>0$, made of spherical photon orbits inside the Cauchy horizon (for $a=0.99M$ and $\xi=1.5M$, at $\eta\simeq0.106M^2$, with a double root at $r\simeq0.31M$); it lies in the capture region and plays no role in the exterior. The orbit $r_s=3M$ survives for every spin: Eqs.~\eqref{PhotonXi}--\eqref{PhotonEta} give $\xi=-2a$ and $\eta=27M^2$ there, and since
\begin{equation}
 27M^2-\eta(r_s)=\frac{(r_s-3M)^2\left(r_s^4-4Ma^2r_s+3M^2a^2\right)}{a^2(r_s-M)^2}\,,
\end{equation}
whose quartic factor is bounded below by $3M^{4/3}a^2\bigl(M^{2/3}-a^{2/3}\bigr)\geq0$, this is the maximum of $\eta$ over the shell, as \citet{Teo:2003} observed.

% Fig_KerrShell.tex -- generated by checks/fig_kerr_shell.py
\begin{figure}[htbp]
\centering
\colorlet{KZero}{blue!70!black}
\colorlet{KHalf}{orange!85!black}
\colorlet{KNine}{teal!80!black}
\colorlet{KMax}{red!75!black}
\includegraphics{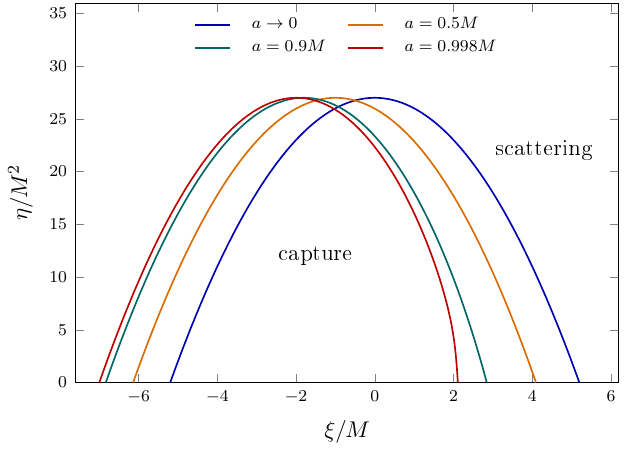}
\caption{Impact parameters $(\xi,\eta)$ of spherical photon orbits, Eqs.~\eqref{PhotonXi}--\eqref{PhotonEta}, for several spins. For $a\to0$ the curve tends to the Schwarzschild parabola $\xi^2+\eta=27M^2$. For every spin, the maximum $\eta=27M^2$ is reached at $\xi=-2a$, by the spherical orbit $r_s=3M$. The endpoints on the $\xi$ axis are the prograde ($\xi>0$) and retrograde ($\xi<0$) equatorial light rings.}
\label{KerrShell}
\end{figure}

% Fig_KerrPhotonOrbits.tex -- generated by checks/fig_kerr_photon_orbits.py
% a=0.9M; spherical photon orbits at r_s=1.75, 2.53, 3.00, 3.75 M (photon shell 1.558M--3.910M)
\begin{figure}[htbp]
\centering
\colorlet{KPhot}{teal!80!black}
\includegraphics{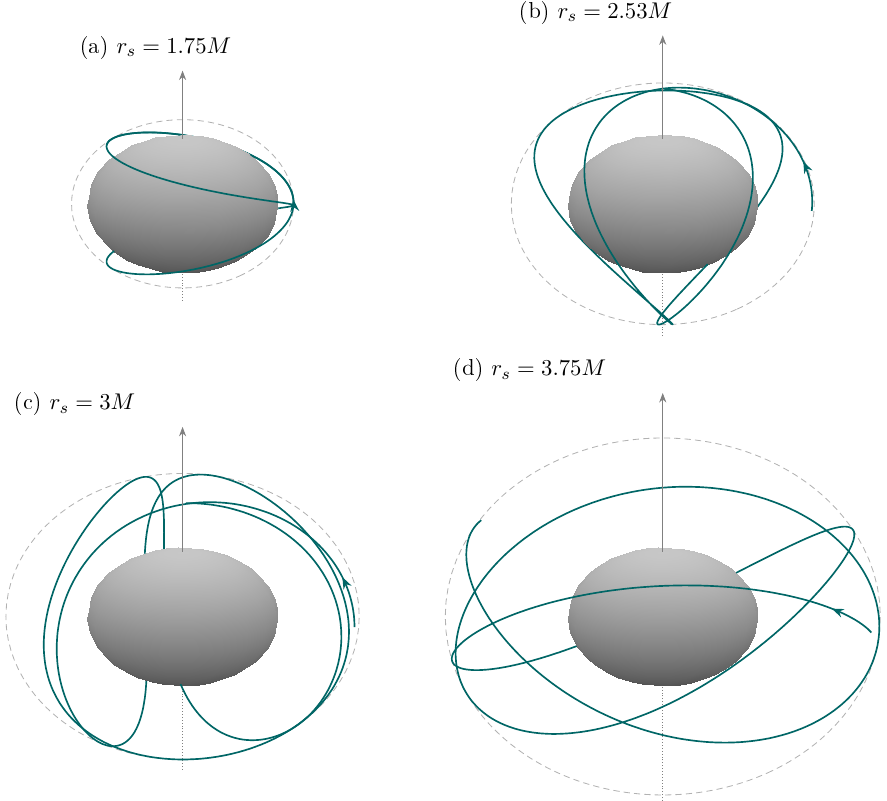}
\caption{Spherical photon orbits for $a=0.9M$, at four radii across the photon shell $r_{\rm ph}^{\rm pro}\simeq1.56M\leq r_s\leq r_{\rm ph}^{\rm ret}\simeq3.91M$, in Kerr--Schild Cartesian coordinates. Each orbit lies on the surface $r=r_s$ (dashed outline) and oscillates in latitude while it moves in azimuth; the arrows mark the direction of motion. (a)~Near the prograde light ring, $\xi\simeq2.37M$: the orbit stays within $47^\circ$ of the equator and co-rotates with the hole. (b)~$\xi\simeq0.11M$: the orbit passes within $1.3^\circ$ of the poles and is dragged in the prograde direction; at $r_s\simeq2.56M$, where $\xi=0$, it would cross the spin axis. (c)~$r_s=3M$, $\xi=-2a$, $\eta=27M^2$, the maximum of $\eta$ over the shell: the orbit counter-rotates and reaches $\theta\simeq19^\circ$. (d)~Near the retrograde light ring, $\xi\simeq-5.82M$: the orbit stays within $28^\circ$ of the equator and counter-rotates. Each orbit corresponds to one point of the $a=0.9M$ curve in Fig.~\ref{KerrShell}.}
\label{Fig:KerrPhotonOrbits}
\end{figure}

\section{Worked examples}\label{Sec:KerrExamples}

We close the chapter with four examples, in which the root pattern predicted from the signs of the coefficients, Vieta's formulae, the discriminant, and Descartes' rule at the horizon is compared with the numerically computed roots. We set $M=1$, so that all lengths are in units of $M$, and numbers are rounded to the digits shown.

For a given set of constants $(\varepsilon,L_z,\mathcal Q)$, the analysis of this chapter reduces to the following steps:
\begin{enumerate}
 \item determine the allowed polar interval from $F(z)$, Eqs.~\eqref{PolarVieta}--\eqref{VorticalConditions};
 \item restrict the radial root pattern with Vieta's formulae, Eqs.~\eqref{KerrVieta1}--\eqref{KerrVieta4}, the parity argument, and the signs of the coefficients of Eq.~\eqref{KerrRWithB};
 \item decide whether the roots are real from the sign of $\Delta_4$ (Section~\ref{Sec:Discriminants}) and, if they are, count those outside the horizon with Descartes' rule at $r_{\rm H}$, Eq.~\eqref{KerrHorizonDescartes};
 \item at a multiple root, read off the stability from the position of the remaining roots, Eq.~\eqref{KerrDoubleRootStability};
 \item impose future orientation, Eqs.~\eqref{HorizonForwardCondition} and~\eqref{ZAMOForwardCondition}.
\end{enumerate}
The examples below follow these steps.

\begin{example}[A generic bound Kerr orbit]\label{Ex:KerrBound}
Take $a=0.9$, so that $r_{\rm H}\simeq1.4359$ and $r_{\rm C}\simeq0.5641$, and ask for the orbit with $p=5$, $e=0.3$ and inclination $i=50^\circ$, Eq.~\eqref{KerrInclination}; this is the bound orbit of Fig.~\ref{Fig:KerrBoundOrbits}a. Its turning points are $r_{\rm p}=50/13\simeq3.8462$ and $r_{\rm a}=50/7\simeq7.1429$, and the three conditions $R(r_{\rm p})=R(r_{\rm a})=0$ and $\cos i=L_z/\sqrt{L_z^2+\mathcal Q}$ fix
\begin{equation}
 \varepsilon=0.922488\,,\qquad L_z=1.828999\,,\qquad \mathcal Q=4.751162\,.
\end{equation}
\begin{enumerate}
 \item \emph{Polar motion.} Here $A=a^2(\varepsilon^2-1)=-0.12070<0$ and $\mathcal Q>0$, so the orbit crosses the equatorial plane. By Eq.~\eqref{PolarVieta}, the roots of $F$ have product $-\mathcal Q/A=39.363$ and sum $B/A=68.077$; they are $z_{\max}=0.58320$ and $67.494$, the latter outside $[0,1]$. The particle oscillates between $\theta_{\min}\simeq40.2^\circ$ and $139.8^\circ$.
 \item \emph{Signs and parity.} The coefficients of $R$ are $(-0.14902,\allowbreak\,2,\allowbreak\,-8.21710,\allowbreak\,11.49737,\allowbreak\,-3.84844)$. Every coefficient of $R(-r)$ is negative, so no root is negative, and since $\varepsilon^2<1$, the number of exterior roots is odd: one or three.
 \item \emph{Discriminant and Descartes' rule at the horizon.} The discriminant is $\Delta_4\simeq12.48>0$, so all four roots are real. The sequence~\eqref{KerrHorizonDescartes} is $(-0.149,\allowbreak\,1.144,\allowbreak\,-1.445,\allowbreak\,-1.494,\allowbreak\,1.006)$, with three sign changes: exactly three roots lie outside the horizon, and a bound branch exists.
 \item \emph{Multiple roots.} None.
 \item \emph{Future orientation.} The orbit stays outside the ergoregion, where $\varepsilon>0$ is all that is required.
\end{enumerate}
The numerical roots are $r_1=7.1429$, $r_2=3.8462$, $r_3=1.9505$ and $r_4=0.4820$. They satisfy Eqs.~\eqref{KerrVieta1}--\eqref{KerrVieta4} to the digits shown, $\sum r_i=13.4215$, $\sum_{i<j}r_ir_j=55.1427$, $\sum_{i<j<k}r_ir_jr_k=77.1557$, and $\prod r_i=25.8258$; and, as predicted in Section~\ref{Sec:KerrBound}, the fourth root lies in $(0,r_{\rm C}]$. The turning point $r_3$ of the horizon-connected branch lies well below the periapsis $r_2$; since the separatrix is reached when these two roots merge, this orbit is far from plunging.
\end{example}

\begin{example}[A vortical photon]\label{Ex:KerrVortical}
Take $a=0.9$, $\xi=0.5$ and $\eta=-0.1$. Then $B/\varepsilon^2=a^2-\xi^2-\eta=0.66$ lies between $0$ and $2A/\varepsilon^2=2a^2=1.62$, and $D_\theta/\varepsilon^4=(a^2-\xi^2-\eta)^2+4a^2\eta=0.1116>0$, so the conditions~\eqref{VorticalConditions} hold. The polar roots are $z_{\min}=0.2012$ and $z_{\max}=0.6136$: the photon is confined to the band $38.4^\circ\leq\theta\leq63.3^\circ$, or its mirror image, and never crosses the equatorial plane. The radial polynomial is
\begin{equation}
 \mathcal R=r^4+0.66\,r^2+0.12\,r+0.081\,,
\end{equation}
with only positive coefficients, so there is no positive root; indeed, the four roots are $0.1316\pm0.7585\,i$ and $-0.1316\pm0.3455\,i$. The photon has no radial turning point: followed forward in time, an ingoing photon falls into the black hole, and an outgoing one, emerging from the white-hole region, escapes to infinity.
\end{example}

\begin{example}[A spherical photon orbit]\label{Ex:KerrSpherical}
Take $a=0.9$ and $r_s=2.5$, which lies inside the photon shell $1.5579<r<3.9103$ for this spin. Equations~\eqref{PhotonXi}--\eqref{PhotonEta} give $\xi=0.21481$ and $\eta=22.4194$. The radial roots are the double root $2.5$, together with $s_1=0.5258$ and $s_2=-5.5258$; their sum is zero, as required. By Eq.~\eqref{KerrDoubleRootStability}, $\mathcal R''(r_s)$ has the sign of $(r_s-s_1)(r_s-s_2)=1.974\times8.026>0$, so the orbit is unstable, as is every spherical photon orbit outside the horizon: nearby photons either escape or are captured. The polar roots are $z_{\max}=0.99802$ and $-27.733$, so the orbit reaches $\theta_{\min}\simeq2.6^\circ$; since $\xi$ is small, it passes almost over the poles while remaining at constant radius, like the orbit of Fig.~\ref{Fig:KerrPhotonOrbits}b.
\end{example}

\begin{example}[When the discriminant is not enough]\label{Ex:KerrDiscriminant}
Take $a=0.86$, so that $r_{\rm H}\simeq1.5103$ and $r_{\rm C}\simeq0.4897$, and $\varepsilon=0.9$, $L_z=1.13$, $\mathcal Q=0.0066$.
\begin{enumerate}
 \item \emph{Polar motion.} As in Example~\ref{Ex:KerrBound}, $A<0$ and $\mathcal Q>0$: the orbit crosses the equatorial plane.
 \item \emph{Signs and parity.} The coefficients of $R$ are $(-0.19,\allowbreak\,2,\allowbreak\,-1.42402,\allowbreak\,0.26667,\allowbreak\,-0.00488)$. Again, no root is negative, and the number of exterior roots is odd.
 \item \emph{Discriminant and Descartes' rule at the horizon.} The discriminant is positive, so all four roots are real and non-negative, exactly as in Example~\ref{Ex:KerrBound}. A bound branch might therefore be expected. However, the sequence~\eqref{KerrHorizonDescartes} is $(-0.190,\allowbreak\,0.852,\allowbreak\,5.037,\allowbreak\,7.033,\allowbreak\,3.051)$, with a single sign change: only one root lies outside the horizon, and there is no bound branch.
\end{enumerate}
The numerical roots confirm this: $r_1=9.7742$ is the only exterior root, while $0.4399$, $0.2917$ and $0.0205$ all lie inside the Cauchy horizon, in agreement with Eq.~\eqref{RBetweenHorizons}. The only exterior motion is the horizon-connected branch $r_{\rm H}<r\leq9.7742$: a particle at its apoapsis $r_1$ falls into the black hole. The sign of $\Delta_4$ tells us that the roots are real, but only Descartes' rule at the horizon tells us where they are.
\end{example}

\chapter{Geodesics in Schwarzschild--de~Sitter spacetime}\label{Sec:SdS}

The Schwarzschild--de~Sitter, or Kottler, solution \citep{Kottler:1918} describes a spherically symmetric black hole in a universe with a positive cosmological constant $\Lambda$. Far from the hole, it approaches the de~Sitter spacetime, and it has two horizons: the event horizon of the black hole and a cosmological horizon, beyond which the repulsion due to $\Lambda$ dominates. Its horizons, its static radius, and its circular orbits with their stability were analyzed by \citet{StuchlikHledik:1999}, building on earlier work by Stuchl\'\i k; most of the results of Sections~\ref{Sec:SdSMetric}--\ref{Sec:SdSTimelikePot} can be found there. The complete analytic solution of the geodesic equation, in terms of hyperelliptic functions, was given by \citet{HackmannLaemmerzahl:2008}, who also classified the orbits by the number of positive zeros of the radial polynomial.

We return to spherical symmetry with the tools of the previous chapters. The treatment follows Chapter~\ref{Sec:Schgeod}, and most results reduce to the Schwarzschild ones as $\Lambda\to0$. What changes is the algebra: the radial polynomial of timelike geodesics becomes a quintic, for which no general formula in radicals exists, and the region of physical interest is bounded by two horizons instead of one horizon and infinity. Section~\ref{Sec:SdSMetric} locates the horizons and the static radius, Section~\ref{Sec:SdSGeodesics} adapts the geodesic equations and the effective potential, and Section~\ref{Sec:SdSTimelikePot} classifies the timelike motion with the potential. Section~\ref{Sec:SdSVieta} obtains the same classification from Vieta's formulae and Descartes' rule of signs, Section~\ref{Sec:SdSNull} treats photons, and Section~\ref{Sec:SdSExamples} closes with worked examples.

\section{Metric, horizons, and the static radius}\label{Sec:SdSMetric}

In Schwarzschild coordinates, the line element has the form~\eqref{Schwmetric} with
\begin{equation}\label{SdSf}
 f(r)=1-\frac{2M}{r}-\frac{\Lambda r^2}{3}\,,
 \qquad
 M>0\,,\quad\Lambda>0\,.
\end{equation}
The horizons are the positive zeros of $f$. Multiplying by $-3r$ turns the condition $f=0$ into the cubic equation
\begin{equation}\label{SdSHorizonCubic}
 h(r):=-3rf(r)=\Lambda r^3-3r+6M=0\,,
\end{equation}
whose roots obey
\begin{equation}
 r_1+r_2+r_3=0\,,\qquad
 r_1r_2+r_1r_3+r_2r_3=-\frac3\Lambda\,,\qquad
 r_1r_2r_3=-\frac{6M}{\Lambda}\,.
\end{equation}
The product is negative, so the number of negative roots is odd, and it cannot be three, because the roots add up to zero. Exactly one root is negative, and the other two are either positive or a complex-conjugate pair. The discriminant of the cubic, Eq.~\eqref{CubicDisc}, is
\begin{equation}
 \Delta_3=108\,\Lambda\left(1-9\Lambda M^2\right)\,,
\end{equation}
so the two remaining roots are real and distinct if and only if $9\Lambda M^2<1$ \citep{StuchlikHledik:1999}. We assume this from now on, and call them $r_{\rm H}<r_{\rm C}$: the event horizon and the \emph{cosmological horizon}. Since the roots add up to zero, the negative root is $-(r_{\rm H}+r_{\rm C})$, and the other two Vieta relations become
\begin{equation}
 r_{\rm H}^2+r_{\rm H}r_{\rm C}+r_{\rm C}^2=\frac3\Lambda\,,
 \qquad
 r_{\rm H}r_{\rm C}\left(r_{\rm H}+r_{\rm C}\right)=\frac{6M}{\Lambda}\,.
\end{equation}
The values $h(0)=6M$, $h(2M)=8\Lambda M^3$, $h(3M)=-3M(1-9\Lambda M^2)$, and $h(\sqrt{3/\Lambda})=6M$ locate the horizons without solving the cubic:
\begin{equation}
 2M<r_{\rm H}<3M<r_{\rm C}<\sqrt{\frac3\Lambda}\,.
\end{equation}
The cosmological constant thus moves the event horizon slightly outward, and the cosmological horizon lies inside the horizon $\sqrt{3/\Lambda}$ of the empty de~Sitter spacetime. As $9\Lambda M^2\to1$ the two horizons merge at $r=3M$; for $9\Lambda M^2>1$ there is no horizon at all, and the spacetime is no longer a black hole.

Between the horizons $f>0$ and $\partial_t$ is timelike: this \emph{static region}, $r_{\rm H}<r<r_{\rm C}$, plays the role of the domain of outer communication, and it is the region we consider. Outside it, $f<0$: inside the event horizon, as in Schwarzschild spacetime, and beyond the cosmological horizon in the expanding region of the de~Sitter universe.

\paragraph{The static radius.} The function $f$ has a single maximum in the static region, where
\begin{equation}
 f'(r)=\frac{2M}{r^2}-\frac{2\Lambda r}{3}=0\,,
\end{equation}
that is, at the \emph{static radius}
\begin{equation}\label{SdSStaticRadius}
 r_{\rm st}=\left(\frac{3M}{\Lambda}\right)^{1/3}\,,
 \qquad
 f(r_{\rm st})=1-\frac{3M}{r_{\rm st}}=1-\left(9\Lambda M^2\right)^{1/3}\,.
\end{equation}
For radial timelike motion, $L=0$ and $V_1=f/2$, so the radial equation of the Remark in Section~\ref{Sec:SchwEqs} becomes $\ddot r=-f'/2$. At $r_{\rm st}$, therefore, the attraction of the black hole is balanced by the cosmic repulsion. A particle at rest there stays at rest: the static observer at $r_{\rm st}$ is a geodesic, although an unstable one, since $f$ has a maximum \citep{StuchlikHledik:1999}. Inside $r_{\rm st}$, gravity attracts, and outside it, $\Lambda$ repels.

The observed value $\Lambda\simeq1.1\times10^{-52}\,{\rm m^{-2}}$ makes these scales astronomical. For one solar mass, $\Lambda M^2\simeq2\times10^{-46}$ and $r_{\rm st}\simeq110\,{\rm pc}$; for $10^{12}$ and $10^{15}$ solar masses, which are typical of a large galaxy and a rich cluster, $r_{\rm st}\simeq1.1$ and $11\,{\rm Mpc}$, while $r_{\rm C}\simeq\sqrt{3/\Lambda}\simeq5\,{\rm Gpc}$. Near a black hole, $\Lambda$ is utterly negligible, and its effects concern the motion at the scale of $r_{\rm st}$. In the figures of this chapter, we take the much larger value $\Lambda M^2=10^{-4}$, for which all the characteristic radii fit in one picture: $r_{\rm H}\simeq2.0003M$, $r_{\rm st}\simeq31.07M$, and $r_{\rm C}\simeq172.2M$. Figure~\ref{SdSRadii} shows how these radii, and those of the special circular orbits of Section~\ref{Sec:SdSCircular}, depend on $\Lambda$.

\begin{figure}[htbp]
\centering
\includegraphics{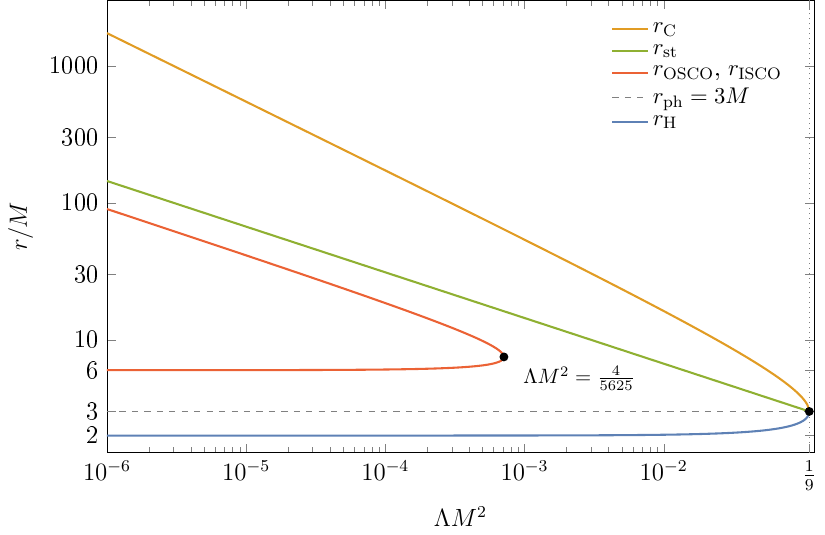}
 \caption{Characteristic radii of the Schwarzschild--de~Sitter spacetime as functions of $\Lambda M^2$: the event horizon $r_{\rm H}$, the photon sphere $r_{\rm ph}=3M$, the innermost and outermost stable circular orbits, the static radius $r_{\rm st}$, and the cosmological horizon $r_{\rm C}$. As $\Lambda\to0$, $r_{\rm H}\to2M$ and $r_{\rm ISCO}\to6M$, while $r_{\rm OSCO}\simeq4^{-1/3}r_{\rm st}$, $r_{\rm st}$ and $r_{\rm C}$ grow without bound. Stable circular orbits exist only for $\Lambda M^2<4/5625$, where the ISCO and the OSCO merge at $r=15M/2$. At $9\Lambda M^2=1$ the two horizons, the static radius and the photon sphere coalesce at $r=3M$.}
 \label{SdSRadii}
\end{figure}

\section{Geodesic equations and effective potential}\label{Sec:SdSGeodesics}

As noted in the Remark of Section~\ref{Sec:SchwEqs}, nothing in the derivation of the Schwarzschild geodesic equations depends on the form of $f$. The motion is planar, $\varepsilon=f\,\dot t$ and $L=r^2\dot\varphi$ are conserved, and the radial motion obeys
\begin{equation}\label{SdSRadialFirst}
 \dot r^2=\varepsilon^2-f(r)\left(\kappa+\frac{L^2}{r^2}\right)\,,
 \qquad
 V_\kappa(r)=\frac12\,f(r)\left(\kappa+\frac{L^2}{r^2}\right)\,,
\end{equation}
with $\ddot r=-V_\kappa'(r)$ and $\kappa=1$ or $0$ for timelike or null geodesics. The effective-potential method of Section~\ref{Sec:EffPot} applies without change.

Two interpretations change, however, because the spacetime has no asymptotically flat region. The normalization of the Killing vector $\partial_t$ is inherited from the Schwarzschild metric, but there is no asymptotic region in which static observers are inertial. A static observer at radius $r$ has four-velocity $\partial_t/\sqrt{f}$ and measures the energy $\varepsilon/\sqrt{f(r)}$ per unit rest mass: $\varepsilon$ is a conserved label of the geodesic, but not the energy at infinity. The particle at rest at the static radius, for instance, has $\varepsilon^2=f(r_{\rm st})<1$. In the same way, the ratio $b=L/\varepsilon$ still labels null geodesics, but it is no longer the impact parameter of a straight line at infinity. In the static region, $f>0$, and $\varepsilon=f\,\dot t>0$ for every future-directed geodesic.

The potential vanishes at both horizons and is positive between them. For any $\varepsilon>0$, therefore, a neighborhood of each horizon is always allowed, and a particle can always cross either horizon. As at the event horizon, the crossing of the cosmological horizon takes a finite proper time, since $\dot r^2\to\varepsilon^2$ there, but an infinite coordinate time $t$.

\section{Timelike geodesics: effective-potential analysis}\label{Sec:SdSTimelikePot}

For $\kappa=1$, the effective potential is
\begin{equation}\label{SdSPotential}
 V_1(r)=\frac12\left(1+\frac{L^2}{r^2}\right)\left(1-\frac{2M}{r}-\frac{\Lambda r^2}{3}\right)\,.
\end{equation}
It vanishes at $r_{\rm H}$ and $r_{\rm C}$, and in between it has the shape shown in Fig.~\ref{SdSV}.

\subsection{Stationary points and circular orbits}\label{Sec:SdSCircular}

The stationary points of $V_1$ are the zeros of
\begin{equation}\label{SdSStationary}
 -r^4\,V_1'(r)=\frac\Lambda3\,r^5-Mr^2+L^2r-3ML^2\,,
\end{equation}
which reduces to the Schwarzschild condition $Mr^2-L^2r+3ML^2=0$ for $\Lambda=0$. The coefficients show three changes of sign, so by Descartes' rule, there are three positive stationary points or one; since $V_1$ vanishes at both horizons and is positive between them, there is at least one maximum. The potential, therefore, has either a single maximum or, from the inside out, a maximum, a minimum, and a second maximum. The inner maximum is the centrifugal barrier of Schwarzschild spacetime; the outer one is new, and it is produced by the cosmic repulsion, which pushes particles outward beyond the static radius. We denote the positions of the two maxima by $r_{\max}^{\rm in}<r_{\max}^{\rm out}$, with $r_{\min}$ between them, and the heights of the barriers by $V_{\max}^{\rm in}$ and $V_{\max}^{\rm out}$.

As in Section~\ref{Sec:SchwCircular}, the conditions $V_1(r_c)=\varepsilon^2/2$ and $V_1'(r_c)=0$ provide the constants of motion for a circular orbit of radius $r_c$:
\begin{equation}\label{SdSCircularConstants}
 L^2=\frac{r_c^2\left(M-\Lambda r_c^3/3\right)}{r_c-3M}\,,
 \qquad
 \varepsilon^2=\frac{r_c\,f(r_c)^2}{r_c-3M}
 =\frac{\left(r_c-2M-\Lambda r_c^3/3\right)^2}{r_c\,(r_c-3M)}\,.
\end{equation}
They reduce to Eq.~\eqref{SchwCircularConstants} for $\Lambda=0$ \citep[see][and references therein]{StuchlikHledik:1999}. Since $L^2\geq0$, circular orbits exist only for
\begin{equation}
 3M<r_c\leq r_{\rm st}\,.
\end{equation}
The lower limit is the photon sphere, which does not depend on $\Lambda$. At the upper limit $L=0$ and $\varepsilon^2=f(r_{\rm st})$: the circular orbit degenerates into the particle at rest at the static radius. Beyond $r_{\rm st}$, the repulsion exceeds the attraction, and no centripetal force is available.

\paragraph{Stability.} Substituting Eq.~\eqref{SdSCircularConstants} into $V_1''$ gives
\begin{equation}\label{SdSStability}
 V_1''(r_c)=\frac{S(r_c)}{3r_c^3\,(r_c-3M)}\,,
 \qquad
 S(r):=3M(r-6M)-\Lambda r^3(4r-15M)\,,
\end{equation}
so a circular orbit is stable when $S(r_c)>0$. For $\Lambda=0$, $S=3M(r-6M)$, and we recover the Schwarzschild ISCO. The polynomial $-S(r)=4\Lambda r^4-15\Lambda Mr^3-3Mr+18M^2$ shows two changes of sign, so it has two positive roots or none. When they exist, they are the \emph{innermost} and the \emph{outermost stable circular orbits}, $r_{\rm ISCO}<r_{\rm OSCO}$: stable circular orbits exist only between them. The derivatives of both constants in Eq.~\eqref{SdSCircularConstants} with respect to $r_c$ are proportional to $S(r_c)$, so $L^2$ and $\varepsilon^2$ have a minimum at the ISCO and a maximum at the OSCO. The two roots merge when $S$ and $S'$ vanish together, which happens at $r=15M/2$ for \citep{StuchlikHledik:1999}
\begin{equation}
 \Lambda M^2=\frac{4}{5625}\simeq7.1\times10^{-4}\,.
\end{equation}
For larger $\Lambda$, every circular orbit is unstable. Since $S(6M)=-1944\,\Lambda M^4<0$, the circular orbit at $6M$ is unstable for every $\Lambda>0$: the cosmological constant pushes the ISCO outward. For small $\Lambda$, the OSCO is where $3Mr\simeq4\Lambda r^4$, that is, $r_{\rm OSCO}\simeq4^{-1/3}r_{\rm st}\simeq0.63\,r_{\rm st}$. For $\Lambda M^2=10^{-4}$,
\begin{gather}
 r_{\rm ISCO}\simeq6.069M\,,\qquad L_{\rm ISCO}\simeq3.451M\,,\qquad \varepsilon^2_{\rm ISCO}\simeq0.8857\,,\nonumber\\
 r_{\rm OSCO}\simeq18.53M\,,\qquad L_{\rm OSCO}\simeq4.174M\,,\qquad \varepsilon^2_{\rm OSCO}\simeq0.9253\,.
\end{gather}

The circular orbits thus form three families, shown in Fig.~\ref{SdSCircularFig}: unstable orbits at the inner maximum of the potential, for $3M<r_c<r_{\rm ISCO}$; stable orbits at the minimum, for $r_{\rm ISCO}<r_c<r_{\rm OSCO}$; and unstable orbits at the outer maximum, for $r_{\rm OSCO}<r_c<r_{\rm st}$. Along these families, $L^2$ decreases, increases, and decreases again. Hence, for $L<L_{\rm ISCO}$ or $L>L_{\rm OSCO}$, the potential has a single maximum, the outer or the inner one respectively, while for $L_{\rm ISCO}<L<L_{\rm OSCO}$ it has two maxima and a minimum.

\begin{figure}[htbp]
\centering
\includegraphics{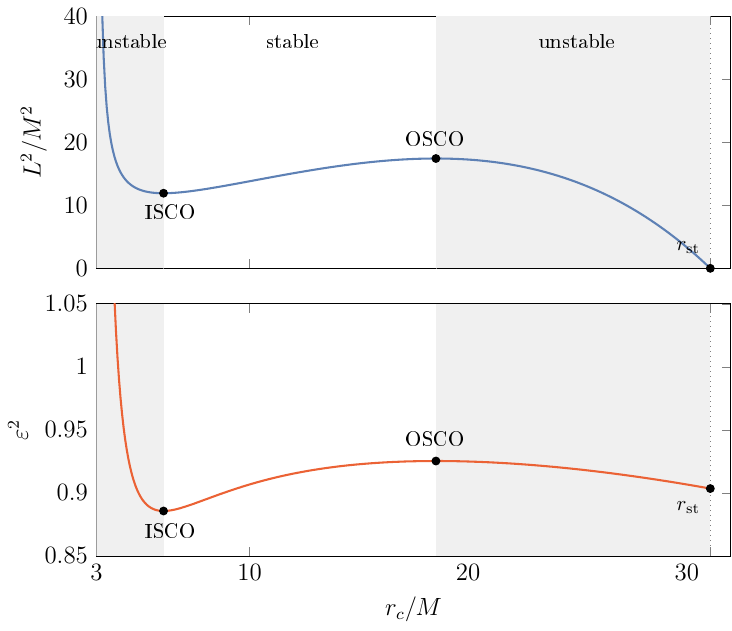}
 \caption{Angular momentum (top) and energy (bottom) of circular timelike orbits as functions of their radius, Eq.~\eqref{SdSCircularConstants}, for $\Lambda M^2=10^{-4}$. Circular orbits exist between the photon sphere, $r_c=3M$, where both constants diverge, and the static radius, $r_{\rm st}\simeq31.07M$, where $L=0$ and $\varepsilon^2=f(r_{\rm st})$. Orbits in the shaded regions are unstable; stable orbits lie between the ISCO, where both constants are minimal, and the OSCO, where both are maximal.}
 \label{SdSCircularFig}
\end{figure}

\begin{figure}[htbp]
\centering
 \includegraphics{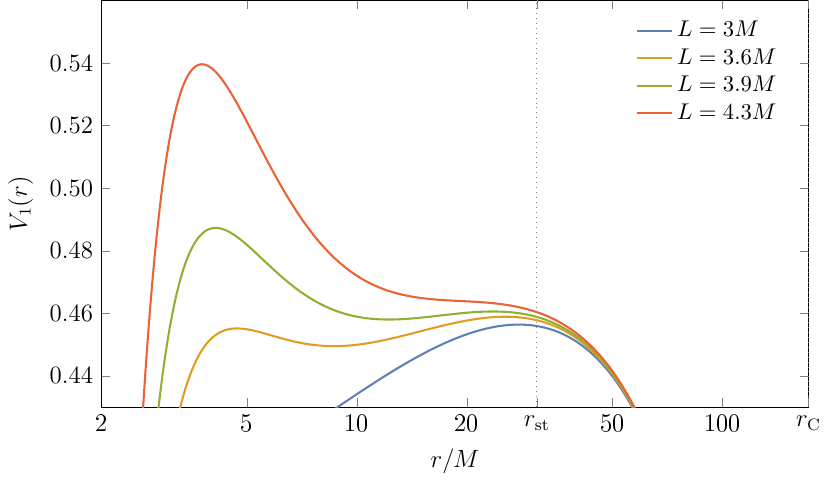}
 \caption{Effective potential $V_1(r)$ for timelike geodesics, Eq.~\eqref{SdSPotential}, with $\Lambda M^2=10^{-4}$, in the static region $r_{\rm H}<r<r_{\rm C}$. The potential vanishes at both horizons. For $L=3M<L_{\rm ISCO}$ it has a single maximum, near the static radius. For $L=3.6M$ and $L=3.9M$, between $L_{\rm ISCO}\simeq3.451M$ and $L_{\rm OSCO}\simeq4.174M$, a centrifugal barrier near the photon sphere, a well, and an outer barrier coexist: the inner barrier is lower than the outer one for $L=3.6M$ and higher for $L=3.9M$. For $L=4.3M>L_{\rm OSCO}$ only the inner barrier is left.}
 \label{SdSV}
\end{figure}

\subsection{Classification of the motion}\label{Sec:SdSClassification}

When there are two barriers, let $V_{\rm high}$ and $V_{\rm low}$ denote the higher and the lower of $V_{\max}^{\rm in}$ and $V_{\max}^{\rm out}$. Comparing the energy level $\varepsilon^2/2$ with the heights of the barriers and the bottom of the well gives the generic motions of Table~\ref{tab:SdSPotentialClassification}, displayed in the plane of the constants of motion in Fig.~\ref{SdSParamSpace}. The branches are those of Section~\ref{Sec:SchwClassification}, with the cosmological horizon in place of infinity.
\begin{description}
 \item[Bound branch,] $r_{\rm p}\leq r\leq r_{\rm a}$. The particle oscillates radially in the well between the two barriers.
 \item[Cosmological branch,] $r_{\rm p}\leq r<r_{\rm C}$. The particle enters the static region through the cosmological horizon, reaches its periapsis, and leaves again. This branch replaces the scattering branch of Schwarzschild spacetime, but it exists for every energy: since $V_1$ vanishes at $r_{\rm C}$, every energy level meets the potential near the cosmological horizon.
 \item[Horizon-connected branch,] $r_{\rm H}<r\leq r_{\rm a}$, as in Schwarzschild spacetime. When there is no turning point at all, the branch connects the two horizons: the particle falls into the black hole from the cosmological horizon, or, followed backward in time, travels from the black hole to the cosmological region.
\end{description}
The cosmological branch is the qualitative novelty. In Schwarzschild spacetime, every particle with $\varepsilon^2<1$ that is not on a bound orbit eventually falls into the black hole. Here, a particle beyond the outer barrier is pushed outward by the cosmic repulsion instead, and leaves through the cosmological horizon. For the same reason, a bound orbit can now be lost in two ways: across the inner barrier, toward the black hole, or across the outer one, toward the cosmological horizon. It is bounded by two separatrices, families of orbits asymptotic to the inner and outer unstable circular orbits respectively.

\begin{table}[htbp]
\centering
\small
\begin{tabular}{@{}ll>{\raggedright\arraybackslash}p{0.4\textwidth}@{}}
\toprule
Angular momentum & Energy range & Radial motion in the static region\\
\midrule
$L<L_{\rm ISCO}$ or $L>L_{\rm OSCO}$
 & $\varepsilon^2>2V_{\max}$ & no turning point\\
 & $\varepsilon^2<2V_{\max}$ & horizon-connected and cosmological branches\\[2pt]
$L_{\rm ISCO}<L<L_{\rm OSCO}$
 & $\varepsilon^2>2V_{\rm high}$ & no turning point\\
 & $2V_{\rm low}<\varepsilon^2<2V_{\rm high}$
 & horizon-connected and cosmological branches\\
 & $2V_{\min}<\varepsilon^2<2V_{\rm low}$
 & horizon-connected, bound and cosmological branches\\
 & $\varepsilon^2<2V_{\min}$
 & horizon-connected and cosmological branches\\
\bottomrule
\end{tabular}
\caption{Generic timelike motions in Schwarzschild--de~Sitter spacetime. For $L<L_{\rm ISCO}$ or $L>L_{\rm OSCO}$ the potential has a single maximum $V_{\max}$, near the static radius or near the photon sphere, respectively. For $L_{\rm ISCO}<L<L_{\rm OSCO}$ it has an inner barrier $V_{\max}^{\rm in}$, a well $V_{\min}$ and an outer barrier $V_{\max}^{\rm out}$; $V_{\rm high}$ and $V_{\rm low}$ are the higher and the lower of the two barriers. When the energy lies between the two barrier heights, the particle passes over the lower barrier but not over the higher one: if the inner barrier is the lower, the horizon-connected branch extends to the outer barrier; if the outer barrier is the lower, the cosmological branch extends down to the inner one. The boundary cases correspond to circular orbits and to orbits approaching them asymptotically.}
\label{tab:SdSPotentialClassification}
\end{table}

\begin{figure}[htbp]
\centering
\includegraphics{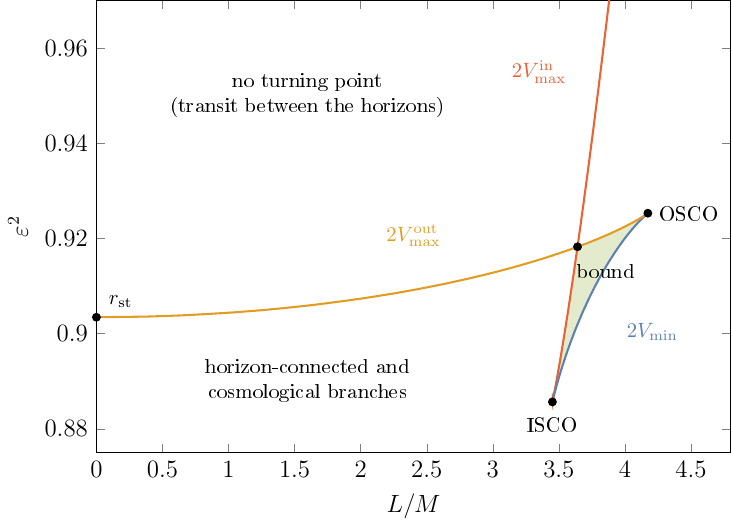}
 \caption{Table~\ref{tab:SdSPotentialClassification} in the plane of the constants of motion, for $\Lambda M^2=10^{-4}$. The red, blue and orange curves are the energies of the inner unstable, stable, and outer unstable circular orbits, parametrized by $r_c$ through Eq.~\eqref{SdSCircularConstants}. The stable family meets the two unstable ones in cusps at the ISCO and the OSCO; the outer family ends at $L=0$ on the particle at rest at the static radius. The two unstable families cross where the two barriers have the same height, Eq.~\eqref{SdSEqualBarriers}. In the green region a bound branch exists, together with a horizon-connected and a cosmological branch. Above the upper envelope of the curves there is no turning point; everywhere else, including the region below the frame, there are a horizon-connected and a cosmological branch. As $\Lambda\to0$, the orange curve tends to the line $\varepsilon^2=1$ and the diagram tends to Fig.~\ref{SchwParamSpace}.}
 \label{SdSParamSpace}
\end{figure}

\begin{figure}[htbp]
\centering
\includegraphics{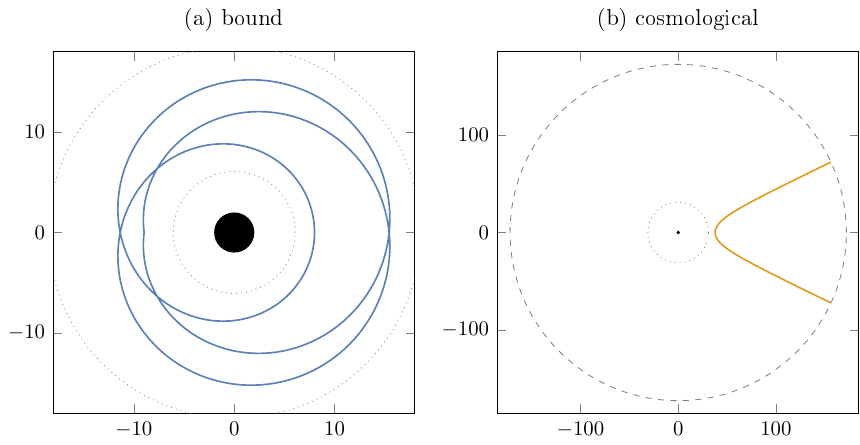}
 \caption{Timelike geodesics in the orbital plane for $\Lambda M^2=10^{-4}$, with $x=r\cos\varphi$ and $y=r\sin\varphi$ in units of $M$. (a)~Bound orbit with $r_{\rm p}=8M$ and $r_{\rm a}=16M$, the orbit of Example~\ref{Ex:SdSBound}; the dotted circles are the ISCO and the OSCO. (b)~Cosmological branch with $L=3M$ and $\varepsilon^2=0.905$: the particle enters through the cosmological horizon (dashed), is turned back at $r_{\rm p}\simeq37.9M$, outside the static radius (dotted), and leaves again. The black hole is the small disk at the center.}
 \label{SdSOrbits}
\end{figure}

\section{Timelike geodesics from Vieta's formulae}\label{Sec:SdSVieta}

\citet{HackmannLaemmerzahl:2008} classified the orbits by the number of positive zeros of the radial polynomial, written in the variable $1/r$, which they mapped numerically in the plane of the constants of motion. Here we show which numbers are possible, and when, from the coefficients alone. Multiplying Eq.~\eqref{SdSRadialFirst} by $r^3$ gives the radial polynomial
\begin{equation}\label{SdSPolynomial}
 P(r):=r^3\dot r^2=\frac\Lambda3\,r^5+\left(\tilde\varepsilon^2-1\right)r^3+2Mr^2-L^2r+2ML^2\,,
 \qquad
 \tilde\varepsilon^2:=\varepsilon^2+\frac{\Lambda L^2}{3}\,.
\end{equation}
It is the Schwarzschild polynomial~\eqref{SchwTimelikePolynomial}, with $\varepsilon^2$ replaced by $\tilde\varepsilon^2$, plus the term $\Lambda r^5/3$. As in Schwarzschild spacetime, positive roots and physical roots are the same thing. For $0<r<r_{\rm H}$ and for $r>r_{\rm C}$ we have $f<0$, hence $V_1<0$ and $P=r^3(\varepsilon^2-2V_1)>0$; at the horizons, $P=r^3\varepsilon^2>0$. Every positive root of $P$ therefore lies in the static region, and since $P$ is positive at both ends of it, the number of these roots, counted with multiplicity, is even: zero, two, or four.

For $L\neq0$, let $r_0,\dots,r_4$ be the five roots. Vieta's formulae give
\begin{align}
 \sum_ir_i&=0\,,\label{SdSVieta1}\\
 \sum_{i<j}r_ir_j&=\frac{3(\tilde\varepsilon^2-1)}{\Lambda}\,,\label{SdSVieta2}\\
 \sum_{i<j<k}r_ir_jr_k&=-\frac{6M}{\Lambda}\,,\label{SdSVieta3}\\
 \sum_{i<j<k<l}r_ir_jr_kr_l&=-\frac{3L^2}{\Lambda}\,,\label{SdSVieta4}\\
 r_0r_1r_2r_3r_4&=-\frac{6ML^2}{\Lambda}\,.\label{SdSVieta5}
\end{align}

\paragraph{Signs.} The coefficients of $P(-r)$, in decreasing order of degree, are $-\Lambda/3$, $0$, $-(\tilde\varepsilon^2-1)$, $2M$, $L^2$, $2ML^2$. Whatever the sign of $\tilde\varepsilon^2-1$, the sequence of non-zero coefficients has exactly one change of sign, so by Descartes' rule $P$ has \emph{exactly one negative root}, which we call $r_0$. The coefficients of $P$ itself have the signs $+$, $\sgn(\tilde\varepsilon^2-1)$, $+$, $-$, $+$. If $\tilde\varepsilon^2\geq1$, they change sign twice, and there are at most two positive roots; if $\tilde\varepsilon^2<1$, they change sign four times. A bound branch, which needs four positive roots, therefore requires
\begin{equation}\label{SdSBoundNecessary}
 \varepsilon^2<1-\frac{\Lambda L^2}{3}\,,
\end{equation}
the counterpart of the Schwarzschild condition $\varepsilon^2<1$. The generic root patterns are therefore only three:
\begin{enumerate}
 \item one negative root and two complex-conjugate pairs: no turning point, and the particle travels between the two horizons;
 \item one negative root, a complex-conjugate pair, and two positive roots $r_1>r_2$: a horizon-connected branch $r_{\rm H}<r\leq r_2$ and a cosmological branch $r_1\leq r<r_{\rm C}$;
 \item one negative root and four positive roots $r_1>r_2>r_3>r_4$: a horizon-connected branch $r_{\rm H}<r\leq r_4$, a bound branch $r_{\rm p}=r_3\leq r\leq r_2=r_{\rm a}$, and a cosmological branch $r_1\leq r<r_{\rm C}$.
\end{enumerate}
Compared with Schwarzschild spacetime, the pattern with a single positive root, which gave a horizon-connected branch only, has disappeared: the outermost positive root is now the periapsis of a cosmological branch. As $\Lambda\to0$, this root and the negative one run off to $+\infty$ and $-\infty$, and the Schwarzschild patterns are recovered. In the third pattern, Eq.~\eqref{SdSVieta1} gives the negative root directly, $r_0=-(r_1+r_2+r_3+r_4)$.

\paragraph{Two identities.} Dividing Eq.~\eqref{SdSVieta4} by Eq.~\eqref{SdSVieta5} gives
\begin{equation}\label{SdSReciprocal}
 \sum_{i=0}^4\frac1{r_i}=\frac1{2M}\,,
\end{equation}
which is Eq.~\eqref{SchwReciprocal} with two more roots. The energy and the cosmological constant have both dropped out. As $\Lambda\to0$, the reciprocals of the two runaway roots vanish, and the Schwarzschild identity is recovered. Combining Eqs.~\eqref{SdSVieta3}--\eqref{SdSVieta5} in the same way gives the sum of the squared reciprocals,
\begin{equation}\label{SdSReciprocalSquares}
 \sum_{i=0}^4\frac1{r_i^2}=\left(\sum_i\frac1{r_i}\right)^2-2\,\frac{\sum_{i<j<k}r_ir_jr_k}{r_0r_1r_2r_3r_4}=\frac1{4M^2}-\frac2{L^2}\,,
\end{equation}
again independent of $\varepsilon$ and $\Lambda$. If $L^2<8M^2$, the right-hand side is negative, and the roots cannot all be real (Section~\ref{Sec:RootPatterns}): there is no bound branch. The same identity holds for the Schwarzschild cubic, where it gives a weaker bound than $L^2>12M^2$.

\paragraph{Multiple roots.} A double root of $P$ in the static region is a circular orbit, as in Section~\ref{Sec:SchwVieta}. A triple root $r_c$ requires
\begin{equation}
 P(r)=\frac\Lambda3\,(r-r_c)^3\left(r^2+3r_cr+\beta\right)\,,
\end{equation}
where the coefficient $3r_c$ is fixed by the vanishing of the $r^4$ term. Comparing the other coefficients with Eq.~\eqref{SdSPolynomial} gives
\begin{equation}
 \beta=\frac{6Mr_c^2}{6M-r_c}\,,\qquad
 L^2=\Lambda r_c^2\left(r_c^2-\beta\right)\,,\qquad
 3M(r_c-6M)=\Lambda r_c^3\left(4r_c-15M\right)\,,
\end{equation}
and the last condition is $S(r_c)=0$: the triple roots are the ISCO and the OSCO. Vieta's formulae also tell where they lie. The two remaining roots have the negative sum $-3r_c$ and the product $\beta$. Since $P$ has exactly one negative root, they must be real and of opposite signs, so $\beta<0$, that is, $r_c>6M$. Both marginally stable orbits lie beyond the Schwarzschild ISCO.

The two unstable families of circular orbits cross, in Fig.~\ref{SdSParamSpace}, where the potential has two maxima at the same height, and $P$ has two double roots. Writing $u$ and $v$ for their sum and product, the vanishing $r^4$ term fixes the negative root at $-2u$, and
\begin{equation}
 P(r)=\frac\Lambda3\,(r+2u)\left(r^2-ur+v\right)^2\,.
\end{equation}
Comparing coefficients gives
\begin{equation}\label{SdSEqualBarriers}
 v=\frac{4Mu^2}{M+u}\,,\qquad
 \Lambda u^3(u+5M)=3M(M+u)\,,\qquad
 L^2=\frac{\Lambda uv^2}{3M}\,.
\end{equation}
For $\Lambda M^2=10^{-4}$, the two orbits have radii $4.570M$ and $25.26M$, with $L\simeq3.641M$ and $\varepsilon^2\simeq0.9182$. As $\Lambda\to0$, $u^3\simeq3M/\Lambda=r_{\rm st}^3$ and $v\simeq4Mu$, so the inner orbit tends to $v/u=4M$, the marginally bound orbit of Schwarzschild spacetime, while the outer one recedes with the static radius.

\paragraph{The discriminant.} Since exactly one root is always real and negative, the three patterns have two, one, and no complex-conjugate pairs, and Eq.~\eqref{DiscSign} gives:
\begin{itemize}
 \item $\Delta_5<0$: horizon-connected and cosmological branches;
 \item $\Delta_5>0$: either no turning point or, in addition, a bound branch.
\end{itemize}
The discriminant is a polynomial in the constants of motion, too long to be useful in explicit form, which vanishes on the curves of circular orbits in Fig.~\ref{SdSParamSpace}. As for the Kerr quartic, a positive discriminant does not decide between the two alternatives that remain. A bound branch is excluded if $\tilde\varepsilon^2\geq1$, by Eq.~\eqref{SdSBoundNecessary}, or if $L^2<8M^2$, by Eq.~\eqref{SdSReciprocalSquares}; otherwise, a single radius in the static region where $P<0$ shows that the roots exist.

\section{Null geodesics}\label{Sec:SdSNull}

For $\kappa=0$, multiplying Eq.~\eqref{SdSRadialFirst} by $r^3/\varepsilon^2$ gives
\begin{equation}\label{SdSNullPolynomial}
 \frac{r^3\dot r^2}{\varepsilon^2}
 =\left(1+\frac{\Lambda b^2}{3}\right)\left(r^3-b_\Lambda^2\,r+2Mb_\Lambda^2\right)\,,
 \qquad
 \frac1{b_\Lambda^2}:=\frac1{b^2}+\frac\Lambda3\,.
\end{equation}
Up to a positive factor, this is the Schwarzschild null polynomial~\eqref{rschweq2null} with $b$ replaced by $b_\Lambda$: the cosmological constant enters only through this effective impact parameter. Everything in Section~\ref{Sec:SchwNull} carries over. Exactly one root is negative; the positive roots, when they exist, lie in the static region, by the argument of Section~\ref{Sec:SdSVieta}; the photon sphere is at $r_{\rm ph}=3M$ for every $\Lambda$; and the discriminant $4b_\Lambda^4(b_\Lambda^2-27M^2)$ separates capture, for $b_\Lambda<3\sqrt3\,M$, from a horizon-connected branch together with a cosmological one, for $b_\Lambda>3\sqrt3\,M$. In terms of $b$, the critical value is \citep{StuchlikHledik:1999}
\begin{equation}
 b_c=\frac{3\sqrt3\,M}{\sqrt{1-9\Lambda M^2}}\,.
\end{equation}
As $b\to\infty$, $b_\Lambda^2\to3/\Lambda$ and the cubic in parentheses tends to $h(r)/\Lambda$, whose positive roots are the two horizons. For finite $b$ with $b_\Lambda>3\sqrt3\,M$, the cubic is negative at $r=3M$, where it equals $M(27M^2-b_\Lambda^2)$: the periapsis of the cosmological branch lies between $3M$ and $r_{\rm C}$, and the turning point of the horizon-connected branch between $r_{\rm H}$ and $3M$.

The same substitution works for the orbit equation. With $u=1/r$, $(\dd u/\dd\varphi)^2=1/b_\Lambda^2-u^2+2Mu^3$, which is the Schwarzschild equation with $b_\Lambda$ in place of $b$. The shape of a light ray, as a curve $r(\varphi)$, therefore does not depend on $\Lambda$ beyond this relabeling \citep{Islam:1983}. This does not mean that lensing is blind to $\Lambda$: the angles that observers actually measure between light rays depend on the metric at the observer's position, and therefore on $\Lambda$ \citep{RindlerIshak:2007}.

\section{Worked examples}\label{Sec:SdSExamples}

We set $M=1$ and $\Lambda=10^{-4}$, as in the figures, and round numbers to the digits shown.

\begin{example}[A Schwarzschild--de~Sitter bound orbit]\label{Ex:SdSBound}
Take $r_{\rm p}=8$ and $r_{\rm a}=16$, the orbit of Fig.~\ref{SdSOrbits}a. Since $P$ is linear in $\varepsilon^2$ and $L^2$, the two conditions $P(8)=P(16)=0$ give $\varepsilon^2\simeq0.91483$ and $L^2\simeq14.288$. Then $\tilde\varepsilon^2-1\simeq-0.0847<0$, the coefficients of $P$ change sign four times, and $1/4-2/L^2\simeq0.110>0$: neither test excludes a bound branch. The numerical roots are $r_1=33.107$, $r_2=16$, $r_3=8$, $r_4=3.3463$ and $r_0=-60.453$, so $r_{\rm H}<r_4<r_3$ and $r_2<r_1<r_{\rm C}$, and the discriminant is positive, as no complex pair is present. The roots add up to zero, and their reciprocals add up to $1/2$, Eq.~\eqref{SdSReciprocal}. Besides the bound branch, the same constants of motion allow a horizon-connected branch, $r_{\rm H}<r\leq3.3463$, and a cosmological branch, $r\geq33.107$, which turns back just outside the static radius $r_{\rm st}\simeq31.07$.
\end{example}

\begin{example}[A transit orbit that the rule of signs does not exclude]\label{Ex:SdSTransit}
Take $L=2$ and $\varepsilon^2=0.95$. Then $\tilde\varepsilon^2-1\simeq-0.0499<0$, so Descartes' rule allows zero, two, or four positive roots, and the discriminant, $\Delta_5\simeq1.4\times10^{-5}$, is positive, which leaves zero or four. Equation~\eqref{SdSReciprocalSquares} settles the question: $1/4-2/L^2=-1/4<0$, so the roots are not all real, and there must be two complex-conjugate pairs. Indeed, the numerical roots are $-51.937$, $0.9972\pm1.7932i$, and $24.971\pm21.773i$. There is no turning point: since $L<L_{\rm ISCO}$ and $\varepsilon^2$ exceeds the height $2V_{\max}\simeq0.907$ of the single barrier, at $r\simeq29.8$, the particle travels between the horizons.
\end{example}

\begin{example}[A scattered photon]\label{Ex:SdSPhoton}
For $b=6$, the effective impact parameter is $b_\Lambda=(1/36+\Lambda/3)^{-1/2}\simeq5.9964$, larger than $3\sqrt3\simeq5.196$; equivalently, $b>b_c\simeq5.1985$. The roots of the cubic in Eq.~\eqref{SdSNullPolynomial} are $4.4489$, $2.3704$ and $-6.8193$, close to those of Example~\ref{Ex:SchwPhoton} for the same $b$. The photon comes in through the cosmological horizon, reaches its periapsis at $r=4.4489$, and leaves again, while the root $2.3704$ bounds a disconnected horizon-connected branch.
\end{example}

\chapter{Conclusions}\label{Sec:Conclusions}

The aim of this monograph was to show how far one can go in classifying black-hole geodesics without solving for the turning points. The answer is: very far. Vieta's formulae, supplemented by the parity of the number of exterior roots, the rule of signs, the discriminant, and the values of the radial polynomial at $r=0$, at the horizon, and at infinity, determine the qualitative behavior of every generic geodesic in the domain of outer communication of the Schwarzschild, Kerr, and Schwarzschild--de~Sitter black holes. Several of these results were known: some go back to \citet{Hagihara:1931}, \citet{Darwin:1959,Darwin:1961} and \citet{Wilkins:1972rs}, the complete Kerr radial taxonomy is due to \citet{CompereLiuLong:2022}, and the Schwarzschild--de~Sitter horizons and circular orbits were analyzed by \citet{StuchlikHledik:1999}. What the method adds is a single, transparent route to all of them.

The main results can be summarized as follows.
\begin{enumerate}
 \item \textbf{Schwarzschild spacetime.} The radial cubic admits only two generic root patterns in each energy regime. Two negative roots and one positive root are excluded for $\varepsilon^2<1$, and a positive root with a complex pair is excluded for null geodesics. Every positive root lies outside the horizon, so the discriminant alone decides the motion. The reciprocals of the three roots always add up to $1/(2M)$, Eq.~\eqref{SchwReciprocal}. This gives the third root of a bound orbit in closed form and, with it, the separatrix $p=6+2e$, as in \citet{Darwin:1961}. The ISCO is the triple root, and the critical photon impact parameter $b_c=3\sqrt3\,M$ is the zero of the discriminant.
 \item \textbf{Kerr spacetime: polar motion.} The endpoint values $F(0)=\mathcal Q$ and $F(1)=-L_z^2$, together with the Vieta relations of the polar quadratic, give the three regimes $\mathcal Q>0$, $\mathcal Q=0$ and $\mathcal Q<0$. Bound timelike geodesics have $\mathcal Q\geq0$ \citep{Wilkins:1972rs}. Vortical motion requires $A>0$, $0<B<2A$ and $D_\theta\geq0$, Eq.~\eqref{VorticalConditions}; for massive particles, $A>0$ means $\varepsilon^2>1$. Because the polar function is linear in its \enquote{energy} $\mathcal Q$, the same classification has a transparent effective-potential reading.
 \item \textbf{Kerr spacetime: radial motion.} For $\varepsilon^2\geq1$ there are at most two exterior turning points: for $\varepsilon^2>1$ because four positive roots would have a negative sum, and for $\varepsilon^2=1$ because the radial function is a cubic with an even number of exterior roots. Hence no radially bound orbit exists, as \citet{Wilkins:1972rs} first showed with the rule of signs. For $\varepsilon^2<1$ the number of exterior roots is odd, and at least one root lies inside the Cauchy horizon, because $R>0$ between the two horizons. A positive discriminant is necessary for bound motion but not sufficient: three of the four real roots may lie inside the Cauchy horizon. The exact criterion is Descartes' rule applied at the horizon, which counts the exterior roots once the discriminant has shown them to be real. The separatrix of bound motion is the locus where the periapsis meets the turning point of the horizon-connected branch.
 \item \textbf{Multiple roots and stability.} Spherical orbits are the double roots of $R$, and their stability is decided by the position of the other two roots. Every spherical orbit with $\varepsilon^2\geq\kappa$, including every spherical photon orbit, is unstable. In the equatorial plane, Vieta's formulae alone give the circular-orbit constants, the stability criterion $\varepsilon^2<1-2M/(3r_c)$, the ISCO equation as a triple root, and the marginally bound radii.
 \item \textbf{Vortical geodesics.} The coefficient $B$ that makes vortical polar motion possible is also the coefficient of $r^2$ in $R$. Every coefficient of the radial polynomial is then non-negative, with a positive constant term, so vortical geodesics, whether timelike or null, have no radial turning point outside the horizon, in agreement with \citet{CompereLiuLong:2022}. Since they reach infinity, vortical timelike geodesics have $\varepsilon>1$.
 \item \textbf{Photons.} The equatorial light rings follow from Vieta's formulae, and spherical photon orbits fill the shell between them. Their impact parameters form the part of the discriminant locus of $\mathcal R$ on which the double root lies outside the horizon, the boundary between capture and scattering; for every spin, $\eta\leq27M^2$, with equality only for the orbit at $r_s=3M$ \citep{Teo:2003}.
 \item \textbf{Schwarzschild--de~Sitter spacetime.} The horizon cubic has roots that add up to zero, and its discriminant shows that a black hole requires $9\Lambda M^2<1$. The timelike radial polynomial is a quintic, which always has exactly one negative root and whose positive roots all lie between the two horizons, in even number. Only three generic patterns remain: transit between the horizons, a horizon-connected branch with a cosmological one, and the same with a bound branch, which requires $\varepsilon^2+\Lambda L^2/3<1$. The reciprocals of the five roots still add up to $1/(2M)$, and the triple roots are the innermost and outermost stable circular orbits, both beyond $6M$. Null geodesics reduce to the Schwarzschild case with an effective impact parameter.
\end{enumerate}

The method has clear limits. Vieta's formulae are necessary conditions on a root pattern. Turning them into a classification requires locating the roots with respect to the horizon; in Kerr spacetime this is the role of Descartes' rule at the horizon, without which the discriminant can mislead. The analysis covers the generic cases: horizon roots, $R(r_{\rm H})=0$, and multiple roots require separate treatment, and the complete radial taxonomy, including all non-generic cases, is given by \citet{CompereLiuLong:2022}. The radial potential, which is quadratic in $\varepsilon$ with a frame-dragging cross term, is useful as a picture in the equatorial plane but not as a classification tool. Finally, the detailed trajectories, the fundamental frequencies and the precession rates need the explicit elliptic-function solutions \citep{FujitaHikida:2009,GrallaLupsasca:2020}.

The same approach applies to any spacetime whose geodesic equations separate into polynomial radial and polar parts. The most immediate extensions are the extremal Kerr black hole, $a=M$, where the two horizons merge and the Cauchy-horizon bounds used above degenerate, and the Kerr naked singularity, $a>M$, where there is no horizon and the sign of $R(0)=-a^2\mathcal Q$ becomes the key boundary value. Beyond Kerr, the Kerr--Newman family keeps a quartic radial polynomial but changes its coefficients, while in the Kerr--de~Sitter family the timelike radial polynomial becomes a sextic. Schwarzschild--de~Sitter spacetime with a negative cosmological constant, which has no cosmological horizon and confines all timelike motion, is a natural companion to Chapter~\ref{Sec:SdS}. In all these cases the logic of the classification stays the same: fix the constants of motion, read off the symmetric functions of the roots, and let the signs decide.

\bibliography{Vieta}
\bibliographystyle{plainnat} 

\end{document}